\documentclass[10pt]{article}

\usepackage{xspace}
\usepackage{ragged2e}
\usepackage{url}
\usepackage{mathtools}
\usepackage{amssymb}
\usepackage{amsthm}
\usepackage{empheq}
\usepackage{latexsym}
\usepackage{enumitem}
\usepackage{eurosym}
\usepackage{dsfont}
\usepackage{appendix}
\usepackage{color} 
\usepackage[unicode]{hyperref}
\usepackage{frcursive}
\usepackage[utf8]{inputenc}
\usepackage[T1]{fontenc}
\usepackage{geometry}
\usepackage{multirow}
\usepackage{todonotes}
\usepackage{lmodern}
\usepackage{subcaption}
\usepackage{anyfontsize}
\usepackage{stmaryrd}
\usepackage{natbib}
\usepackage{cleveref}
\usepackage[english]{babel}
\usepackage[english=british]{csquotes}

\setcitestyle{numbers,open={[},close={]}}

\definecolor{red}{rgb}{0.7,0.15,0.15}
\definecolor{green}{rgb}{0,0.5,0}
\definecolor{blue}{rgb}{0,0,0.7}
\hypersetup{colorlinks, linkcolor={red},citecolor={green}, urlcolor={blue}}
			
\makeatletter \@addtoreset{equation}{section}

\newtheorem{theorem}{Theorem}[section]
\newtheorem{assumption}[theorem]{Assumption}

\newtheorem{example}[theorem]{Example}

\newtheorem{lemma}[theorem]{Lemma}
\newtheorem{proposition}[theorem]{Proposition}

\newtheorem{definition}[theorem]{Definition}
\newtheorem{remark}[theorem]{Remark}

\makeatletter
\newcommand{\uset}[3][0ex]{%
  \mathrel{\mathop{#3}\limits_{
    \vbox to#1{\kern-2\ex@
    \hbox{$\scriptstyle#2$}\vss}}}}
\makeatother

\def \E{\mathbb{E}}
\def \F{\mathbb{F}}

\def \N{\mathbb{N}}

\def \P{\mathbb{P}}

\def \R{\mathbb{R}}

\def\Ac{\mathcal{A}}
\def\Bc{\mathcal{B}}

\def\Fc{\mathcal{F}}

\def\Sc{\mathcal{S}}
\def\Tc{\mathcal{T}}

\def\Zc{\mathcal{Z}}

\def\Cbf{\mathbf{C}}

\def\Cfrak{\mathfrak{C}}

\def\be{\begin{eqnarray}}
\def\ee{\end{eqnarray}}
\def\b*{\begin{eqnarray*}}
\def\e*{\end{eqnarray*}}

\def\1{{\bf 1}}

\def\eps{\varepsilon}

\newcommand{\tnorm}[1]{\left\vert\kern-0.25ex\left\vert\kern-0.25ex\left\vert #1 
    \right\vert\kern-0.25ex\right\vert\kern-0.25ex\right\vert}

\begin{document}

\title{Is there a Golden Parachute\\ in Sannikov's principal--agent problem?\footnote{We are grateful to Yuliy Sannikov for his insightful comments on the first version of this paper. This work benefited from support of the ANR project PACMAN ANR--16--CE05--0027.}}

\author{Dylan Possama\"{\i}\footnote{ETH Z\"urich, department of Mathematics, Switzerland, dylan.possamai@math.ethz.ch} \and Nizar Touzi\footnote{CMAP, \'Ecole Polytechnique, 91128 Palaiseau Cedex, France, nizar.touzi@polytechnique.edu. This author is also grateful for the financial support from the Chaires FiME--FDD and Financial Risks of the Louis Bachelier Institute.}}

\date{\today}

\maketitle

\abstract{
This paper provides a complete review of the continuous-time optimal contracting problem introduced by \citeauthor*{sannikov2008continuous} \cite{sannikov2008continuous}, in the extended context allowing for possibly different discount rates for both parties. The agent's problem is to seek for optimal effort, given the compensation scheme proposed by the principal over a random horizon. Then, given the optimal agent's response, the principal determines the best compensation scheme in terms of running payment, retirement, and lump-sum payment at retirement.

A Golden Parachute is a situation where the agent ceases any effort at some positive stopping time, and receives a payment afterwards, possibly under the form of a lump sum payment, or of a continuous stream of payments. We show that a Golden Parachute only exists in certain specific circumstances. This is in contrast with the results claimed by \citeauthor*{sannikov2008continuous} \cite{sannikov2008continuous}, where the only requirement is a positive agent's marginal cost of effort at zero. Namely, we show that there is no Golden Parachute if this parameter is too large. Similarly, in the context of a concave marginal utility, there is no Golden Parachute if the agent's utility function has a too negative curvature at zero.

In the general case, we prove that an agent with positive reservation utility is either never retired by the principal, or retired above some given threshold (as in \citeauthor*{sannikov2008continuous}'s solution). We show that different discount factors induce a \emph{face-lifted utility function}, which allows to reduce the analysis to a setting similar to the equal--discount rates one. Finally, we also confirm that an agent with small reservation utility does have an informational rent, meaning that the principal optimally offers him a contract with strictly higher utility than his participation value.
\medskip

	\noindent {\bf Key words:} continuous-time principal--agent, optimal control and stopping, face-lifting.
}

\section{Introduction}

Principal--agent problems naturally stem from questions of optimal contracting between two parties, a principal (`she') and an agent (`he'), when the agent's effort cannot be observed or contracted upon. Mathematically, they are formulated as a Stackelberg non--zero sum game, and can also be identified to bi-level optimisation problems in the operations research literature. The number of articles related to this topic is staggering, mainly due to the wide spectrum of concrete problems where this theory is able to provide relevant results, for instance for moral hazard problem in microeconomics with applications to corporate governance, portfolio management, and many other areas of economics and finance. 

\medskip
    The first, and seminal, paper on principal--agent problems in continuous-time is by \citeauthor*{holmstrom1987aggregation} \cite{holmstrom1987aggregation}, who show that the optimal contract is linear in the output process, in a finite horizon setting with CARA utility functions for both parties, and when the agent's effort impacts solely the drift of the output process. This paper is the first to highlight that optimal contracting problems tend to be easier to address in continuous-time, an observation which has been confirmed by the large continuous-time literature in this area. \citeauthor*{holmstrom1987aggregation}'s work was extended by \citeauthor*{schattler1993first} \cite{schattler1993first}, \citeauthor*{sung1995linearity} \cite{sung1995linearity,sung1997corporate}, \citeauthor*{muller1998first} \cite{muller1998first,muller2000asymptotic}, and \citeauthor*{hellwig2002discrete} \cite{hellwig2002discrete,hellwig2007role}. While the aforementioned papers use continuous-time extensions of the celebrated first-order approach from the contract theory literature in static cases, see for instance \citeauthor*{rogerson1985first} \cite{rogerson1985first}, the papers by \citeauthor*{williams2008dynamic} \cite{williams2008dynamic,williams2011persistent,williams2015solvable} and \citeauthor*{cvitanic2009optimal} \cite{cvitanic2006optimal,cvitanic2008principal,cvitanic2009optimal} use the stochastic maximum principle and forward--backward stochastic differential equations to characterise the optimal compensation for more general utility functions, see also the excellent monograph by \citeauthor*{cvitanic2012contract} \cite{cvitanic2012contract}.\footnote{Other early continuous-time contract theory models were proposed by \citeauthor*{adrian2009disagreement} \cite{adrian2009disagreement}, \citeauthor*{biais2007dynamic} \cite{biais2007dynamic}, \citeauthor*{biais2010large} \cite{biais2010large}, \citeauthor*{biais2013dynamic} \cite{biais2013dynamic}, \citeauthor*{capponi2015dynamic} \cite{capponi2015dynamic}, \citeauthor*{demarzo2006optimal} \cite{demarzo2006optimal}, \citeauthor*{demarzo2012dynamic} \cite{demarzo2012dynamic}, \citeauthor*{fong2009evaluating} \cite{fong2009evaluating}, \citeauthor*{he2009optimal} \cite{he2009optimal}, \citeauthor*{hoffmann2010reward} \cite{hoffmann2010reward}, \citeauthor*{ju2012optimal} \cite{ju2012optimal}, \citeauthor*{keiber2003overconfidence} \cite{keiber2003overconfidence}, \citeauthor*{leung2014continuous} \cite{leung2014continuous}, \citeauthor*{mirrlees2013strategies} \cite{mirrlees2013strategies}, \citeauthor*{myerson2008leadership} \cite{myerson2008leadership}, \citeauthor*{ou2003optimal} \cite{ou2003optimal}, \citeauthor*{pages2012bank} \cite{pages2012bank}, \citeauthor*{pages2014mathematical} \cite{pages2014mathematical}, \citeauthor*{piskorski2010optimal} \cite{piskorski2010optimal}, \citeauthor*{piskorski2016optimal} \cite{piskorski2016optimal}, \citeauthor*{sannikov2007agency} \cite{sannikov2007agency}, \citeauthor*{schroder2010continuous} \cite{schroder2010continuous}, \citeauthor*{van2010dynamic} \cite{van2010dynamic}, \citeauthor*{westerfield2006optimal} \cite{westerfield2006optimal}, \citeauthor*{zhang2009dynamic} \cite{zhang2009dynamic}, \citeauthor*{zhou2006principal} \cite{zhou2006principal}, or \citeauthor*{zhu2011sticky} \cite{zhu2011sticky}.}

\medskip
The seminal work of \citeauthor*{sannikov2008continuous} \cite{sannikov2008continuous}, see also \cite{sannikov2013contracts}, represents a genuine breakthrough in this vast literature from various perspectives. First, from the methodological viewpoint, \citeauthor*{sannikov2008continuous} introduced the idea to focus on the dynamic continuation value of the agent as a state variable for the principal's problem. Although this idea was already acknowledged throughout the discrete-time literature on this problem, an illuminating example being \citeauthor*{spear1987repeated} \cite{spear1987repeated}, its systematic implementation in continuous-time offers an elegant solution approach by means of a representation result of the dynamic value function of the agent. Second, the infinite horizon setting considered by \citeauthor*{sannikov2008continuous} revealed remarkable economic implications. Indeed, the main conclusions are that the principal optimally retires the agent, offering him a Golden Parachute, that is to say a lifetime constant continuous stream of consumption, when his continuation utility reaches a sufficiently high level, and that an agent with small reservation utility possesses an informational rent, in the sense that he is offered a contract with strictly higher value. 

\medskip
    The main objective of our paper is twofold. First, we revisit \citeauthor*{sannikov2008continuous}'s seminal work, though putting a stronger weight on technical rigour, which is unfortunately lacking in some key parts of \cite{sannikov2008continuous}. We would like to emphasise that this should not be seen in any case as a reason to underestimate the importance of this paper, given the groundbreaking novelties recalled above. In contrast, our first aim is to try and contribute even more to the success of \cite{sannikov2008continuous} by making it more accessible to a wider community of mathematicians and economists, whose overall understanding of the model may be hindered by the technical gaps in \cite{sannikov2008continuous}. Notice that we are not the first to try and obtain rigorously the results claimed in \cite{sannikov2008continuous}. For instance, \citeauthor*{strulovici2015smoothness} \cite[Section 4.3]{strulovici2015smoothness} offer a more rigorous take on the existence of optimal contracts in the model. However, the authors take for granted the fact that \cite{sannikov2008continuous} proves that the HJB equation for the principal's problem has a smooth solution, while we will argue that the proof has important gaps. Similarly, the unpublished PhD thesis of \citeauthor*{choi2014sannikov} \cite{choi2014sannikov} aims at putting the problem on rigorous foundations. Nonetheless, existence of optimal contracts is not addressed there, and the results rely on the assumption that it is never optimal to retire the agent temporarily, while our approach actually proves that this is the case. We also would like to refer to the recent work of \citeauthor*{decamps2019two} \cite{decamps2019two}, where the authors study a related, but different, contracting problem, and where again the heart of the analysis is technical clarity: this should be an additional illustration that proving rigorously results in this literature is a challenging task. 
    
    \medskip
    Our second goal is to prove that our analysis extends beyond the case where the principal and the agent have the same discount rates. It is an important feature, as most models\footnote{Exceptions are the recent work by \citeauthor*{hajjej2019optimal} \cite{hajjej2019optimal}, where the agent is risk-averse and more impatient than the principal. However, they do not obtain clear results saying that the hypotheses of their verification result \cite[Theorem 4.3]{hajjej2019optimal} can be verified in practice, as well as the work of \citeauthor*{lin2020random} \cite{lin2020random}, but there the emphasis is more on obtaining general methods to attack infinite horizon moral hazard problems.} in the discrete- or continuous-time literature either allow for risk-averse agents who are as patient as the principal, as in \citeauthor*{sannikov2008continuous} \cite{sannikov2008continuous}, \citeauthor*{fong2009evaluating} \cite{fong2009evaluating}, \citeauthor*{myerson2008leadership} \cite{myerson2008leadership}, and \citeauthor*{hajjej2017optimal} \cite{hajjej2017optimal}, or for more impatient, but risk-neutral agents, as in \citeauthor*{demarzo2006optimal} \cite{demarzo2006optimal}, \citeauthor*{biais2007dynamic} \cite{biais2007dynamic}, \citeauthor*{biais2010large} \cite{biais2010large}, \citeauthor*{biais2013dynamic} \cite{biais2013dynamic}, \citeauthor*{he2009optimal} \cite{he2009optimal}, \citeauthor*{piskorski2010optimal} \cite{piskorski2010optimal}, \citeauthor*{piskorski2016optimal} \cite{piskorski2016optimal}, \citeauthor*{demarzo2012dynamic} \cite{demarzo2012dynamic}, \citeauthor*{pages2014mathematical} \cite{pages2014mathematical}, or \citeauthor*{williams2015solvable} \cite{williams2015solvable}. Even more surprisingly, our analysis can also accommodate the case where the principal is actually strictly more impatient than the agent.
%
    As far as we know, our paper is the first offering such a comprehensive analysis. 
    
\medskip
Our main findings are the following. First, in contrast with the overall message from \cite{sannikov2008continuous}, we show that a Golden Parachute may fail to exist in some specific situations. It never happens if the agent's marginal cost of effort at zero is zero, or is sufficiently large. And it never happens if the agent's marginal utility is also concave, and his utility function has sufficiently large negative curvature at zero, with a level depending on the marginal cost of effort at zero. We also highlight the fact that there are three regimes, depending on whether the agent is strictly more impatient than the principal or not:
\begin{itemize}
\item[$(i)$] when the principal is much more impatient than the agent (the actual bound depends on the level of risk-aversion of the agent), the problem degenerates, optimal contracts cease to exist, and the principal can achieve her first-best value with appropriately defined sequences of incentive-compatible contracts;

\item[$(ii)$] when the principal is still more impatient than the agent, but above the aforementioned threshold, and the agent's marginal cost of effort at zero is positive, a Golden Parachute is likely to exist, in the sense that we know that for a large enough continuation utility of the agent, the value of the principal coincides with the value she would obtain by offering a Golden Parachute. Then, the only way a Golden Parachute could fail to exist is when these two functions are equal everywhere, so that the contract does not even start. We also provide sufficient conditions, combining the curvature of the agent's utility function at zero and his marginal cost of effort at zero, for an informational rent to exist;

\item[$(iii)$] when the agent is strictly more impatient than the principal, we do not know whether the two functions mentioned in $(ii)$ coincide or not, and some numerical evidence seems to indicate that they could fail to do so in general. Besides, we prove that in this case, the value of the principal is non-increasing with respect to the continuation utility of the agent, meaning that an informational rent cannot exist.
\end{itemize}
We emphasise that our rigorous presentation involves advanced tools from stochastic control theory and partial differential equations. In particular, the justification of the solution claimed by \citeauthor*{sannikov2008continuous} in \cite{sannikov2008continuous} requires an in-depth analysis of the first-best solution, combined with the theory of viscosity solutions, and it is unclear to us how the proof could be significantly simplified. 

\medskip
Finally, from the methodological and theoretical point of view, we have highlighted a novel phenomenon in (properly renormalised) moral hazard problems with risk-aversion and different discount rates, where the principal's problem has an optimal stopping component, in the sense that she can terminate the contract. Indeed, we prove that the problem could actually be treated as in the case with similar discount rates, provided that the lump sum payment in the principal's problem be evaluated through an appropriate face-lifted version of the agent's inverse utility function. The last function is obtained as the value function of the best termination time of the contract when the agent ceases any effort. We believe that this finding in the context of the present contracting problem applies to a wider class of moral hazard problems with early retirement possibilities.\footnote{The 'face-lifting' phenomenon corresponds to the so-called boundary layer effect in singular optimal control problems, and appeared naturally in various pricing problems in finance, either with hedging constraints or market frictions, see for instance \citeauthor*{broadie1998optimal} \cite{broadie1998optimal}, \citeauthor*{bouchard2000explicit} \cite{bouchard2000explicit}, \citeauthor*{chassagneux2015terminal} \cite{chassagneux2015terminal}, \citeauthor*{guasoni2008consistent} \cite{guasoni2008consistent}, \citeauthor*{soner2002dynamic} \cite{soner2000superreplication,soner2002dynamic,soner2003problem,soner2007hedging}, \citeauthor*{cheridito2005multi} \cite{cheridito2005multi}, and \citeauthor*{schmock2002valuation} \cite{schmock2002valuation}, or for utility maximisation problems, see \citeauthor*{larsen2016facelifting} \cite{larsen2016facelifting}. However all these references consider either `pure' optimal control or stochastic target problems, while in our context, the face-lifting phenomenon occurs because of an optimal stopping problem, and is therefore of a different nature.}

\medskip
The paper is organised as follows. \Cref{sec:formulation} provides a rigorous formulation of the continuous-time contracting problem, with a clear description of the set of contracts, and introduces the face-lifted utility $\overline F$. Our main results are given in \Cref{sect:mainresults}. As such, \Cref{sec:complete} gives all our results on the first- and second-best problems, while \Cref{sec:nogp} provides some conditions under which no Golden Parachute can exist. Next, \Cref{sec:num} presents our numerical illustrations, and \Cref{sec:gap} discusses the gaps in \cite{sannikov2008continuous}. The subsequent sections are dedicated to the proofs of our main results: \Cref{sec:solvfb} concentrates on the first-best problem, while \Cref{sect:reduction} uses the result of \citeauthor*{lin2020random} \cite{lin2020random}, itself an extension of earlier results by \citeauthor*{cvitanic2017moral} \cite{cvitanic2017moral,cvitanic2018dynamic}, which justify rigorously \citeauthor*{sannikov2008continuous}'s \cite{sannikov2008continuous} remarkable reduction of the principal's Stackelberg game problem into a standard control-and-stopping problem. Such a reduction opens the door for the use of standard tools of stochastic control theory. In particular, we treat the case of a very impatient principal in \Cref{sec:SBproof}, which can be addressed directly by exhibiting a sequence of contracts inducing a degenerate situation where both parties achieve as large a payment as possible. The alternative case of reasonably impatient principal is analysed by means of the corresponding dynamic programming equation introduced in \Cref{sec:DPE}, where we also provide a verification result following the standard theory. In \Cref{sect:solution}, we provide a rigorous analysis of the dynamic programming equation, and we isolate a set of conditions which guarantee that the solution is of the form claimed in \cite{sannikov2008continuous}. 

\section{Sannikov's contracting problem}
\label{sec:formulation}

\subsection{Output process, agent's effort, and contract}

This section reports our understanding of the continuous time contracting model in \citeauthor*{sannikov2008continuous} \cite{sannikov2008continuous}.  Let $(\Omega,\Fc,\P^0)$ be a probability space carrying a one-dimensional $\P^0$--Brownian motion $W^0$. For fixed parameters $\sigma>0$ and $X_0\in\R$, the \emph{output process} is defined by
 \[
 X_t:=
 X_0+\sigma W^0_t,\; t\geq 0.
 \]
We denote by $\F$ the $\P^0$-augmentation of the natural filtration of $X$ (or equivalently, of $W^0$), which is known to satisfy the usual conditions. We next introduce distributions $\P^\alpha$ of the output process under effort $\alpha$, so as to induce the dynamics $\mathrm{d}X_t=\alpha_t\mathrm{d}t+\sigma \mathrm{d}W^\alpha_t$, for some $\P^\alpha$--Brownian motion $W^\alpha$. This is naturally accomplished by means of the following argument based on the Girsanov transformation. Let $\Ac$ be the collection of all $\F$-predictable processes $\alpha$ with values in a compact subset $A$ of $[0,\infty)$, containing $0$. For all $\alpha\in\Ac$, we may introduce an equivalent probability measure $\P^\alpha$ so that the process $W^\alpha:=W^0-\int_0^\cdot \frac{\alpha_s}{\sigma}\mathrm{d}s$ is a $\P^\alpha$--Brownian motion, and the process $X$ can be written in terms of $W^\alpha$ as
\[
 X_t=
 X_0+\int_0^t\alpha_s\mathrm{d}s+\sigma W^\alpha_t,\; t\ge 0.
 \]
Any $\alpha\in\Ac$ is called an \emph{effort} process, and is interpreted as an action exerted in order to affect the distribution of the output process from $\P^0$ to $\P^\alpha$. A \emph{contract} is a triple $\Cbf:=(\tau,\pi,\xi)$, where $\tau\in\Tc$, the set of all $\F$--stopping times, $\xi$ is a non-negative $\Fc_\tau$-measurable random variable, and $\pi\in\Pi$, the set of $\F$-predictable non-negative processes. Here, $\tau$ represents a retirement time, $\pi$ is a process of continuous payment rate until retirement, and $\xi$ is a lump-sum payment at retirement, which may be interpreted as a Golden Parachute in the terminology of \citeauthor*{sannikov2008continuous} \cite{sannikov2008continuous}, see \Cref{def:GP} below.

\medskip
We shall introduce later in \Cref{sect:contract} the collection $\mathfrak{C}^0$ of admissible contracts, by imposing some integrability requirements. These contracts allow to formulate the contracting problem which sets the terms of the delegation by the principal (she) of the output process to the agent (he). Namely, the principal seeks to design the optimal contract so as to guarantee that the agent best serves her objectives, while optimising his own interest.

\subsection{The agent's problem}

The agent preferences are defined by 
\begin{itemize}
\item a utility function $u:[0,\infty)\longrightarrow [0,\infty)$ which is increasing, strictly concave, twice continuously differentiable on $(0,\infty)$, satisfies $u(0)=0$ together with the (one-sided) Inada condition $\lim_{x\to \infty}u^\prime(x)=0$ and the growth condition
\begin{equation}\label{cond:u}
 c_0(-1+\pi^{\frac1\gamma}\big) \le u(\pi) \le c_1\big(1+\pi^{\frac1\gamma}\big),
 \;\pi\ge 0,
 \;\mbox{for some}\;
  (c_0,c_1)\in(0,\infty)^2,\;
  \text{\rm and some}\;\gamma>1,
  \end{equation}
which implies that $u(\infty)=\infty$, and $u^{-1}(y)\leq C\big(1+y^{\gamma}\big)$, for any $(y,\pi)\in[0,\infty)$, and for some $C>0$;\footnote{It should be clear that our specification for $u$ is tailored to a power utility function of the form $u(\pi):=\pi^{1/\gamma}/\gamma$, $\pi\geq 0$, for some $\gamma>1$, though we strictly speaking only require such a behaviour at $\infty$.}
\item a cost function $h:[0,\infty)\longrightarrow [0,\infty)$, increasing, strictly convex, continuously differentiable, with $h(0)=0$;
\item a fixed discount rate $r>0$.
\end{itemize}
Given a contract $\Cbf:=(\tau,\pi,\xi)\in\Cfrak^0$ and $\alpha\in \Ac$, the utility obtained by the agent is defined by the problem
 \begin{equation}\label{Agent0}
 V^{\rm A}(\Cbf)
 :=
 \sup_{\alpha\in\Ac} J^{\rm A}(\Cbf,\alpha),
 \;\mbox{where}\;
 J^{\rm A}(\Cbf,\alpha)
 :=
 \E^{\P^\alpha}\bigg[\mathrm{e}^{-r\tau}u(\xi)
                                 +\int_0^\tau r\mathrm{e}^{-rs}\big(u(\pi_s)-h(\alpha_s)\big)\mathrm{d}s
                       \bigg].
 \end{equation}
As $u\ge 0$, and $A$ is bounded, notice that $J^{\rm A}(\Cbf,\alpha)\in\R\cup\{\infty\}$ is well-defined. Moreover, as the agent is allowed to choose zero effort, inducing $J^{\rm A}(\Cbf,0)\ge 0$, it follows that $V^{\rm A}(\Cbf)\ge 0$ for any proposed contract $\Cbf\in\Cfrak^0$. We denote by 
 \[
 \Ac^\star(\Cbf):=\big\{\alpha\in\Ac: V^{\rm A}(\Cbf)=J^{\rm A}(\Cbf,\alpha)\big\},
 \] 
the (possibly empty) set of all optimal responses of the agent. The agent only accepts contracts which provide him with a utility above a fixed threshold $u(R)$, where $R\geq 0$, called participation level. Thus he considers contracts in the subset
 \[
 \Cfrak^0_R
 :=
 \big\{ \Cbf\in\Cfrak^0: V^{\rm A}(\Cbf)\ge u(R) \big\}.
 \]
Observe that the final lump-sum utility for the agent can be written as $u(\xi)=\int_\tau^\infty r\mathrm{e}^{-rt}u(\xi)\mathrm{d}t$, so that it can be equivalently implemented by the payment of the lifetime consumption $\xi$ after retirement at time $\tau$. We shall comment further on this normalisation in \Cref{rem:normalisation} below.

\subsection{The principal's problem}

The principal is risk-neutral with the objective of maximising her overall revenue induced by the agent's effort
 \[
 J^{\rm P}(\Cbf,\alpha)
 :=
 \E^{\P^\alpha}\bigg[-\mathrm{e}^{-\rho\tau}\xi
                                +\int_0^\tau \rho\mathrm{e}^{-\rho s}(\mathrm{d}X_s-\pi_s\mathrm{d}s)
                       \bigg],
 \]
where we extend \cite{sannikov2008continuous}, allowing the principal to have a possibly different discount rate $\rho>0$ from that of the agent. Observe that for any $\alpha\in\Ac$, $\E^{\P^\alpha}\big[\int_0^\tau\mathrm{e}^{-2\rho s}\mathrm{d}s\big]\le\int_0^\infty \mathrm{e}^{-2\rho s}\mathrm{d}s=\frac1{2\rho}<\infty$. Then, by standard It\^o integration theory, we have $\E^{\P^\alpha}\big[\int_0^\tau\mathrm{e}^{-\rho s}\sigma\mathrm{d}W^\alpha_s\big]=0$ for all stopping time $\tau\in\Tc$, implying that 
\[
J^{\rm P}(\Cbf,\alpha)=\E^{\P^\alpha}\bigg[-\mathrm{e}^{-\rho\tau}\xi+\int_0^\tau \rho\mathrm{e}^{-\rho s}(\alpha_s-\pi_s)\mathrm{d}s\bigg],
\]
which is well-defined in $\{-\infty\}\cup\R$, due to the boundedness of $A$ and the non-negativity of $\xi$ and $\pi$. We also notice again that the lump-sum payment $\xi$ at time $\tau$ can be rewritten as $\xi=\int_\tau^\infty \rho\mathrm{e}^{-\rho t}\xi \mathrm{d}t$, and so it can be implemented by the lifetime payment at rate $\xi$ after $\tau$, in agreement with the corresponding interpretation in the agent's problem.

\medskip
The principal's problem is defined as follows: anticipating the agent's optimal response, she chooses the contract which best serves her objective under the participation constraint
 \begin{equation}\label{Principal0}
 V^{\rm P}
 :=
 \sup_{\Cbf\in\Cfrak^0_{\text{\fontsize{4}{4}\selectfont $R$}}}
 \sup_{\alpha\in\Ac^\star(\Cbf)}
 J^{\rm P}(\Cbf,\alpha).
 \end{equation}
 
\subsection{Reformulation and face-lifted utility}
\label{sect:reformulation}

We next re-write our contracting problem equivalently by using the opposite of the inverse of the agent's utility
 \[
 F:=-u^{-1}\mathbf{1}_{[0,u(\infty))}-\infty\mathbf{1}_{\R_{\text{\fontsize{4}{4}\selectfont $+$}}\setminus[0,u(\infty))}.
 \]
Then, denoting $\zeta:=u(\xi)$ and $\eta:=u(\pi)$, the criterion of the agent becomes
 \begin{equation}\label{Agent}
 J^{\rm A}(\mathbf{C},\alpha)
 =
 \E^{\P^\alpha} \bigg[ \mathrm{e}^{-r\tau}\zeta + \int_0^\tau r\mathrm{e}^{-rt}\big(\eta_t-h(\alpha_t)\big)\mathrm{d}t
                        \bigg],\; (\mathbf{C},\alpha)\in\mathfrak{C}^0\times\Ac,
 \end{equation}
where we abuse notations and indifferently identify as a contract the triplet $(\tau,\pi,\xi)$, or the triplet $(\tau, \eta,\zeta)$. We will use this identification implicitly throughout the paper. As for the principal, we have
 \[
 J^{\rm P}(\Cbf,\alpha)
 =
 \E^{\P^\alpha}\bigg[\mathrm{e}^{-\rho\tau}F(\zeta)
                                +\int_0^\tau \rho\mathrm{e}^{-\rho t}\big(\alpha_t+F(\eta_t)\big)\mathrm{d}t\bigg],\; (\mathbf{C},\alpha)\in\mathfrak{C}^0_R\times\Ac.
 \]
Here, the (negative) reward of the principal by stopping at $\tau$ is $F(\zeta)$. Our first result shows that in general, the principal may be able to improve her reward by not ending the contract at some time $\tau$ with a lump-sum payment to the agent, but by instead discouraging the agent from exerting any efforts (which can be understood as an alternative way of ending the contract), and offering him a continuous consumption. The improved (or face-lifted, hereafter) reward is naturally defined by means of the following deterministic control problem
 \begin{equation}\label{Fbar}
 \overline{F}(y_0)
 :=
 \sup_{p\in\Bc_{\text{\fontsize{4}{4}\selectfont $\R_+$}}}  \sup_{ T\in[0, T_{\text{\fontsize{4}{4}\selectfont $0$}}^{\text{\fontsize{4}{4}\selectfont $y_0,p$}}]}\bigg\{\mathrm{e}^{-\rho T}F\big(y^{y_{\text{\fontsize{4}{4}\selectfont $0$}},p}(T)\big)+\int_0^{T} \rho \mathrm{e}^{-\rho t} F\big(p(t)\big)\mathrm{d}t\bigg\},\; y_0\geq 0,
 \end{equation}
 where $\Bc_{\R_{\text{\fontsize{4}{4}\selectfont $+$}}}$ is the set of Borel-measurable maps from $\R_+$ to $\R_+$, and for all $(y_0,p)\in\R_+\times\Bc_{\R_{\text{\fontsize{4}{4}\selectfont $+$}}}$
 \[
 T_0^{y_{\text{\fontsize{4}{4}\selectfont $0$}},p}:=\inf\big\{t\ge 0:y^{y_{\text{\fontsize{4}{4}\selectfont $0$}},p}(t)\le 0\big\}
 \in [0,\infty],
 \]
and the state process $y^{y_{\text{\fontsize{4}{4}\selectfont $0$}},p}$ is defined by the controlled first-order ODE
 \begin{equation}\label{eq:defydet}
 y^{y_{\text{\fontsize{4}{4}\selectfont $0$}},p}(0)=y_0,\; \dot{y}^{y_{\text{\fontsize{4}{4}\selectfont $0$}},p}(t)=r\big(y^{y_{\text{\fontsize{4}{4}\selectfont $0$}},p}(t)-p(t)\big),\; t> 0.
 \end{equation}
To better understand the expression \eqref{Fbar} for the improved payment, notice that for any $p\in\Bc_{\R_{\text{\fontsize{4}{4}\selectfont $+$}}}$, direct integration of this linear ODE leads to
 \[
 y_0
 =
 \mathrm{e}^{-rT}y^{y_{\text{\fontsize{4}{4}\selectfont $0$}},p}(T)+\int_0^{T}\mathrm{e}^{-rt}p(t)\mathrm{d}t,
 \;
 \mbox{for all}
 \;
 y_0\ge 0,\;\mbox{and}\;T\le T_0^{y_{\text{\fontsize{4}{4}\selectfont $0$}},p},
 \]
meaning that for a given state of the world $\omega$, the agent is indifferent between a lump-sum payment $\xi(\omega)$ at some retirement time $\tau(\omega)$, and a continuous payment $p(t)$ on the time interval $[\tau(\omega),\tau(\omega)+ T]$, with zero effort on this time interval, and a retirement deferred to $\tau(\omega)+ T$, where the lump-sum payment is now $\xi^\prime(\omega):=u^{-1}\big(y^{\zeta(\omega),p}(T)\big)$. The restriction $T\le T_0^{y_{\text{\fontsize{4}{4}\selectfont $0$}},p}$ on such deferral policies is induced by the fact that the agent is protected by limited liability, and therefore can only receive non-negative payments. The idea is that while the agent is indifferent between these two alternatives, the discrepancy between the discount rates may be such that the principal can actually benefit from postponing retirement.

\medskip
An immediate consequence is that the value function of the principal can be expressed in its relaxed formulation as
 \begin{equation}\label{Principalbar}
 \overline{V}^{\rm P}
 :=
 \sup_{\Cbf\in\Cfrak_{\text{\fontsize{4}{4}\selectfont $R$}}}
 \sup_{\alpha\in\Ac^\star(\Cbf)}
 \bar{J}^{\rm P}(\Cbf,\alpha),
 \;\mbox{\rm where}\;
 \bar{J}^{\rm P}(\Cbf,\alpha)
: =
 \E^{\P^\alpha}\bigg[\mathrm{e}^{-\rho\tau}\overline{F}(\zeta)
                                +\int_0^\tau \rho\mathrm{e}^{-\rho t}\big(\alpha_t+F(\eta_t)\big)\mathrm{d}t\bigg],
 \end{equation}
where $\Cfrak_R:=\{\mathbf{C}\in\Cfrak: V^{\rm A}(\mathbf{C})\ge u(R)\}$, for some subset $\Cfrak\subset\Cfrak^0$ defined in \Cref{sect:contract} below. The following result states the equivalence of our original contracting problem\footnote{As observed by Yuliy Sannikov in private communication, the principal problem may be directly defined by the relaxed formulation \eqref{Principalbar}.} with $\overline{V}^{\rm P}$, and characterises the face-lifted reward $\overline{F}$ in closed form in terms of the concave conjugate functions
 \[
 F^\star(p)
 :=
 \inf_{y\ge 0} \big\{ yp-F(y) \big\},
 \;\mbox{and}\;
 \overline{F}^\star(p)
 :=
 \inf_{y\ge 0} \big\{ yp-\overline{F}(y) \big\},
 \;p\in\R.
 \]
Notice that $F^\star=0$ on $\R_+$, and $F^\star<0$ on $(-\infty,0)$. Moreover, our conditions \eqref{cond:u} on the agent's utility function is immediately converted for $F^\star$ into
 \begin{equation}\label{cond:Fstar}
 -c_0^\star\big(1+|p|^{\gamma^\star}\big)
 \le
 F^\star(p)
 \le
 c_1^\star\big(1-|p|^{\gamma^\star}\big),
 \;p\le 0,
 \;\mbox{with}\;\frac1\gamma-\frac1{\gamma^\star}=1,
 \;\mbox{for some}\;
(c_0^\star,c_1^\star)\in(0,\infty)^2.
 \end{equation} 

\begin{proposition}\label{prop:Fbar}
We have $V^{\rm P}=\overline{V}^{\rm P}$, and denoting $ \delta:=\frac{r}{\rho}$, the face-lifted reward function satisfies

\medskip
$(i)$ if $\delta\gamma\le 1$, we have $\overline F=0;$ 

\medskip
$(ii)$ if $\delta\gamma>1$, we have $\overline F=F$, if $\delta=1$, and otherwise, under the additional condition that $\underset{y\to\infty}{\lim}\1_{\{\delta>1\}}\frac{F^{\text{\fontsize{4}{4}\selectfont $\prime$}}(y)}{yF^{\text{\fontsize{4}{4}\selectfont $\prime\prime$}}(y)}$ exists, 
\begin{align}\label{Fbarstar}
\overline{F}=\big(\overline{F}^\star\big)^{\star},
\;\mbox{\rm where}\;
 \overline{F}^\star(p)
 &
 =\int_0^{+\infty}\rho\mathrm{e}^{-\rho t} F^\star\big(\delta\mathrm{e}^{(\rho-r)t}p\big)\mathrm{d}t,\; p\le 0.
 \end{align}
In particular $\overline{F}$ is decreasing, strictly concave, $\delta\overline{F}^\prime(0)= F^\prime(0)\mathbf{1}_{\{\delta\ge1\}}$, and $\overline{F}^\star$ satisfies similar bounds to \eqref{cond:Fstar}, with appropriate positive constants $\bar c_0^\star$, and $\bar c_1^\star$, which translate into bounds on $\overline{F}$ similar to those in \eqref{cond:u}, with appropriate positive constants $\bar c_0$, and $\bar c_1$. Moreover, the supremum over $T$ in \eqref{Fbar} is attained at $T_0^{y_{\text{\fontsize{4}{4}\selectfont $0$}},p}$, meaning that
 \begin{equation}\label{Fbarbis}
 \overline{F}(y_0)
 =
 \sup_{p\in\Bc_{\R_{\text{\fontsize{4}{4}\selectfont $+$}}}} \bigg\{\int_0^{T_{\text{\fontsize{4}{4}\selectfont $0$}}^{y_{\text{\fontsize{4}{4}\selectfont $0$}},p}} \rho \mathrm{e}^{-\rho t} F(p(t))\mathrm{d}t\bigg\},\; y_0\geq 0.
 \end{equation}
\end{proposition}
\noindent The equality $V^{\rm P}=\overline{V}^{\rm P}$ in \Cref{prop:Fbar} is a direct consequence of our definition of admissible contracts in \Cref{sect:contract} below. The remaining claims are proved in \Cref{sec:appendix}, and provide the following significant results. In the case $\rho=r$ considered by \citeauthor*{sannikov2008continuous} \cite{sannikov2008continuous}, the principal never gains by postponing retirement and allowing the agent to produce zero effort for a while. On the other hand, when $\rho\neq r$, and $\rho$ is not too large, it is always optimal to postpone retirement and $\overline F$ is a non-trivial majorant of $F$. Finally, when the principal becomes a lot more impatient than the agent, we actually have $\overline F=0$, meaning that she can bring back the cost of permanently retiring the agent to $0$.

\begin{example}\label{ex:power}
Let $u(\pi):=\pi^{1/\gamma}$, and $\rho\neq r$ with $\rho<\gamma r$, then $F(y)=-y^\gamma$, and we compute directly 
\[
F^\star(p)=-(\gamma-1)\bigg(\frac{|p|}{\gamma}\bigg)^{\gamma/(\gamma-1)},
\; 
\overline{F}^\star(p)=-\frac{\rho(\gamma-1)^2}{r\gamma-\rho}\bigg(\frac{r |p|}{\rho\gamma}\bigg)^{\frac\gamma{\gamma-1}},\; p\leq 0,
\]
and it follows from {\rm \Cref{prop:Fbar}} that
\[
\overline{F}(y)=-\bigg(\frac{r\gamma-\rho}{\rho(\gamma-1)}\bigg)^{\gamma-1}\bigg(\frac{\rho y}{r}\bigg)^\gamma,\; y\geq 0.
\]
\end{example}

\begin{remark}\label{rem:normalisation}
The normalisation of the running rewards of the principal and the agent by their corresponding discount rates in {\rm \Cref{Agent0}} and {\rm \Cref{Principal0}}, is not fundamental, \emph{per se}. However, the face-lifted principal's benefit function plays a crucial role to relate equivalent formulations of the problem. Consider for instance the agent's criterion 
\[
\overline{J}_0^{\rm A}(\mathbf{C},\alpha):=\E^{\P^\alpha}\bigg[\mathrm{e}^{-r\tau}u(\xi)+\int_0^\tau \mathrm{e}^{-rs}\big(u(\pi_s)-h(\alpha_s)\big)\mathrm{d}s \bigg],
\]
which differs from $J^{\rm A}$ in \eqref{Agent0} by the form of discount factor $\mathrm{e}^{-rt}$ instead of $r\mathrm{e}^{-rt}$. Similarly, change the principal's criterion to 
\[
\overline{J}^{\rm P}_0(\Cbf,\alpha):=\E^{\P^\alpha}\bigg[-\mathrm{e}^{-\rho\tau}\xi+\int_0^\tau \mathrm{e}^{-\rho t}\big(\alpha_t-\pi_t)\big)\mathrm{d}t\bigg].
\]
Then, following the same argument, the corresponding face-lifted utility function is
 \[
 \overline{F}_0(y_0)
 :=
 \sup_{p\in\Bc_{\text{\fontsize{4}{4}\selectfont $\R_+$}}}  \sup_{ T\in[0, T_{\text{\fontsize{4}{4}\selectfont $0$}}^{y_{\text{\fontsize{4}{4}\selectfont $0$}},p}]}\bigg\{\mathrm{e}^{-\rho T}F\big(y^{y_{\text{\fontsize{4}{4}\selectfont $0$}},p}(T)\big)+\int_0^{T} \mathrm{e}^{-\rho t} F\big(p(t)\big)\mathrm{d}t\bigg\},\; y_0\geq 0,
 \]
with controlled state satisfying for any $p\in\Bc_{\R_{\text{\fontsize{4}{4}\selectfont $+$}}}$, $ y^{y_{\text{\fontsize{4}{4}\selectfont $0$}},p}(0)=y_0$, and $\dot y^{y_{\text{\fontsize{4}{4}\selectfont $0$}},p}(t)=ry^{y_{\text{\fontsize{4}{4}\selectfont $0$}},p}(t)-p(t)$, $t>0$. The corresponding Hamilton--Jacobi equation is 
\[
\min\Big\{\overline{F}_0-F,\;\rho\overline{F}_0-ry\overline{F}_0^\prime+F^\star(\overline{F}_0^\prime)\Big\}=0.
\]
In particular, in the case $\rho=r$ of equal discount rates, we see immediately that $\overline{F}_0(y):=r^{-1}F(ry)$, $y\ge 0$, is a solution of this equation. Consequently the decision of retiring the agent should be discussed by comparing the principal's value function to $\overline{F}_0$ instead of $F$ in this case, see {\rm\Cref{def:GP}} below. In this sense, the setting of {\rm \cite{sannikov2008continuous}} is the only parametrisation of the problem with $\rho=r$ for which the face-lifted retirement reward function $\overline F$ coincides with $F$.
\end{remark}

\subsection{Admissible contracts and Golden Parachute}
\label{sect:contract}

For technical reasons, we introduce further integrability conditions which guarantee that both criteria of the agent and the principal are finite, and more importantly, allow to apply the reduction result of \citeauthor*{lin2020random} \cite{lin2020random}. We denote by $\Cfrak$ the collection of all contracts $\Cbf:=(\tau,\pi,\xi)$, satisfying in addition the following integrability condition
 \begin{equation}\label{integrability}
 \lim_{n\to\infty}\,\sup_{\alpha\in\Ac} 
 \P^\alpha[\tau\ge n]=0,\; \mbox{and}\;
 \sup_{\alpha\in\Ac} 
 \E^{\P^\alpha}\bigg[\big(\mathrm{e}^{-r^{\text{\fontsize{4}{4}\selectfont $\prime$}}\tau}|\xi|\big)^\gamma
                                +\int_0^\tau \big(\mathrm{e}^{-r^{\text{\fontsize{4}{4}\selectfont $\prime$}} s}|\pi_s|\big)^\gamma\mathrm{d}s
                        \bigg] 
 < \infty,
 \end{equation}
for some $r^\prime\in(0,r\wedge\frac{\rho}{\gamma})$. In order to guarantee that the equality $V^{\rm P}=\overline{V}^{\rm P}$ of \Cref{prop:Fbar} holds, we define the set $\mathfrak{C}^0$ as the collection of all triples $(\tau^0,\pi^0,\xi^0)$ such that
 \[
 \tau^0=\tau+T,
 \;\pi^0=\pi\mathbf{1}_{[0,\tau)}+p\mathbf{1}_{[\tau,\tau^{\text{\fontsize{4}{4}\selectfont $0$}})},\;
 \mbox{and}\;
 u(\xi^0)=y^{u(\xi),p}(T),
  \]
for some $(\tau,\pi,\xi)\in\Cfrak$, and $\Fc_\tau$-measurable $p$ with values in $\Bc_{\R_{\text{\fontsize{4}{4}\selectfont $+$}}}$, and $T$ with values in $\big[0,T_0^{u(\xi),p}\big]$.

\medskip 
We now introduce the notion of Golden Parachute which has two different meanings in our relaxed formulation \eqref{Principalbar}:
\begin{itemize}
\item[$(i)$] in \citeauthor*{sannikov2008continuous}'s formulation, the retirement time $\tau$ is not explicitly involved in the model formulation. Instead, a Golden Parachute is defined as a stopping time $\tau$ such that the agent exerts no effort while receiving a constant consumption on $[\tau,\infty)$; 

\item[$(ii)$] our definition of contracts includes a retirement time $\tau$, and we may naturally define a situation of Golden Parachute by $\tau>0$ and $\xi>0$, $\P^0$--a.s. 
\end{itemize}

\begin{definition}\label{def:GP}
We say that the contracting model exhibits a \emph{Golden Parachute}, if there exists an optimal contract $(\tau^\star,\pi^\star,\xi^\star)\in\mathfrak C_R$ for the relaxed formulation of the principal's problem \eqref{Principalbar} such that $\tau^\star>0$, and $\P^0[\xi^\star>0]>0$.
\end{definition}

\noindent In other words, a Golden Parachute corresponds to a situation where there is a high-retirement point for the agent, with either lump-sum payment at retirement or continuous payment after retirement, where retirement means that the agent ceases to exert any effort forever. 

\subsection{The first-best contracting problem}

\label{sec:firstbest}

We conclude this section with the formulation of the first-best version of the contracting problem
\[
V^{\rm P, FB}:=\sup \Big\{ J^{\rm P}(\mathbf{C},\alpha): 
                                              \mathbf{C}\in\mathfrak{C}^{\rm FB},
                                              \;\alpha\in\Ac,
                                              \;\mbox{and}\;
                                              J^{\rm A}(\mathbf{C},\alpha)\ge u(R)
                                     \Big\},
\]
where $\mathfrak{C}^{\rm FB}$ consists of all contracts $(\tau,\pi,\xi)$ where $\tau\in\Tc$ is a stopping time with values in $[0,\infty]$, and $(\pi,\xi)$ satisfy the integrability condition of \eqref{integrability}. By definition of $J^{\rm P}$ and $V^{\rm P}$, we observe that we have 
\begin{equation}\label{SBleFB}
V^{\rm P} \le V^{\rm P, FB} \le \overline{a}.
\end{equation}
The corresponding notion of Golden Parachute is naturally defined in the context of the present first-best contracting problem similar the second-best setting.

\section{Main results}
\label{sect:mainresults}

\subsection{Complete solution of the contracting problems}\label{sec:complete}

We first provide a complete solution of the first-best contracting problem which has a different nature depending on  the relative value of the agent's and the principal's discount rates $\delta=\frac{r}{\rho}$. Namely, the principal's first-best contracting problem is degenerate in the case $\delta\gamma\leq 1$, in the sense that the principal achieves the maximal value $\overline{a}$ by means of a sequence of admissible contracts which offer no intermediate payments, incite the agent to exert maximal effort at all times, and offer him a large lump-sum payment at large retirement time. In the alternative case $\delta\gamma>1$, the problem does not degenerate any longer, and we shall express its solution in terms of the convex dual of the function $G:=-h^{-1}$
\[
G^\star(p):=\sup_{b\in h(A)}\big\{pb-G(b)\big\}=\sup_{a\in A}\big\{a+ph(a)\big\}, 
\;\mbox{and}\;
\overline G^\star(p):=\int_0^{+\infty}\rho\mathrm{e}^{-\rho t} G^\star\big(\delta\mathrm{e}^{(\rho-r)t}p\big)\mathrm{d}t,\; p\leq 0.
\]
Notice that the expression of $\overline G^\star$ in terms of $G^\star$ is exactly the same of that of $\overline F^\star$ in terms of $F^\star$ in \Cref{prop:Fbar}.$(ii)$.

\begin{theorem}\label{th:firstbest}{\rm (First-best contracting)}

\medskip
\noindent 
{$(i)$} If $\delta\gamma\leq 1$, then, $V^{\rm P, FB}=\bar a$, and there is no admissible contract which achieves this value.

\medskip
\noindent $(ii)$ If $\delta\gamma>1$, then 
	\begin{enumerate}[leftmargin=4em]
	\item[{$(ii$-$1)$}]
	$V^{\rm P, FB}
	=
	(\overline F^\star-\overline G^\star)^\star(u(R))
	:=
	\inf_{\lambda\le 0}\big\{\lambda u(R)-(\overline F^\star-\overline G^\star)(\lambda)\big\}$, with infimum achieved at the unique solution $\lambda^\star\ge 0$ of the first order condition $u(R)=(\overline F^\star-\overline G^\star)^\prime(\lambda^\star)$.
	\item[{ $(ii$-$2)$}] $V^{\rm P, FB}|_{R=0}=\bar a$, and  $\lambda^\star>0$ if and only if $R>0$. In this case the first-best optimal contract sets the agent to his participation constraint, and is given by 
	\[
	\tau^\star=\infty,
	\;\mbox{\rm and}\;\pi_t^\star=-F^\prime\bigg(\frac{\mathrm{e}^{(r-\rho)t}}{\delta\lambda^\star}\bigg),\; 	a^\star_t\in \hat A\bigg(\frac{\mathrm{e}^{(r-\rho)t}}{\delta\lambda^\star}\bigg),\; t\geq 0.
	\]

	\item[{ $(ii$-$3)$}] Denoting $V^{\rm P, FB}=v^{\rm FB}\big(u(R)\big)$, we have
	\[
	0\leq (v^{\rm FB}-\overline F)(y)\leq \overline G^\star\circ\overline F^\prime(y),\; y\geq 	0,
	\;\mbox{\rm so that}\;
	\lim_{y\to+\infty}(v^{\rm FB}-\overline F)(y)=0.
	\]
	In particular

\medskip
	$\bullet$ if $\delta\in(\frac1\gamma, 1]$ and $\beta>0$, we have 
	$
	v^{\rm FB}(y) =\overline F(y), 
	\;\forall y\geq (\overline F^\prime)^{-1}\big(-1/(\delta\beta)\big),$
	and therefore a \emph{Golden Parachute} exists for the first-best contracting problem$;$
	
	\medskip
	$\bullet$ otherwise if $\delta>1$ or $\beta=0$, we have $v^{\rm FB} >\overline F$ on $\R_+$ and there is no \emph{first-best Golden Parachute}.
	\end{enumerate}
\end{theorem}

\begin{remark}\label{rem:FB}
$(i)$ {\rm\Cref{th:firstbest}}.$(ii$-$3)$ states that when $\beta>0$ and $\delta\in(1/\gamma, 1]$, the first-best value function coincides with $\overline F$ after the finite value $(\overline F^\prime)^{-1}\big(-1/(\delta\beta)\big)$. As the value function at the origin is $v^{\rm FB}(0)=\bar a>\bar F(0)=0$, this shows that a \emph{first-best Golden Parachute} exists in this case. 

\medskip
\noindent $(ii)$ By {\rm\Cref{SBleFB}}, we also have in the case $\beta>0$ and $\delta\in (1/\gamma,1]$ that $V^{\rm P}=\overline F$ whenever $u(R)\ge(\overline F^\prime)^{-1}\big(-1/(\delta\beta)\big)$. However, we shall see that $V^{\rm P}|_{R=0}=0$. Therefore, unlike the first-best situation in $(i)$, we cannot discard the possibility that the second-best contracting principal's value coincides everywhere with $\overline F$, see {\rm\Cref{sec:nogp}}.
\end{remark}

\noindent We next focus on the second-best contracting problem. Recall from \Cref{prop:Fbar} that $\overline{F}=0$ when $\delta\gamma\leq 1$. Our first result shows that, similar to the first-best situation of \Cref{th:firstbest}, the solution of the contracting problem is degenerate in this case. We prove this by exhibiting a sequence of admissible contracts which induces a utility as large as we want for the agent, and reaches the highest possible level for the principal, namely $\bar a$. Roughly speaking, these contracts make small intermediate payments, enforce the highest possible effort for the agent at all times, and promise to pay him an extremely high value after an extremely long time. By exploiting the large discrepancy between the discount rates of the agent and the principal, we show that the continuation utilities of both parties reach their maximum.

\medskip
The case $\delta\gamma>1$ is more interesting. Similar to {\rm\citeauthor*{sannikov2008continuous}} \cite{sannikov2008continuous}, we denote by $V(y)$ the maximal value obtained by the principal when the utility of the agent is held at the level $y\geq 0$. We provide a characterisation of this function and of the solution of the contracting problem by means of the second-order differential equation
 \begin{equation}\label{DPE:main}
 v(0)=0,
 \;\mbox{and}\;
 v-\delta y v^\prime + F^\star(\delta v^\prime) - \mathfrak{I}(v^\prime,v^{\prime\prime})^+=0,
 \;\mbox{on}\;[0,\infty),
 \end{equation}
where the second-order differential operator is given by
 \begin{equation}\label{I0}
 \mathfrak{I}(v^\prime,v^{\prime\prime})
 :=
 \sup_{z\in\R,\; \hat a\in\partial h^{\text{\fontsize{4}{4}\selectfont $\star$}}(z)}
\big\{\hat a+h(\hat a)\delta v^\prime+\eta z^2\delta v^{\prime\prime}
\big\},
\;\mbox{for all}\;C^2\;\mbox{function}\;v,
\;\mbox{where}\;
\delta:=\frac{r}{\rho},\; \eta:=\frac12r\sigma^2,
 \end{equation}
 and $\partial h^\star$ denotes the subgradient of the convex dual $h^\star$ of $h$:
 \begin{equation}\label{hstar}
 \partial h^\star(z):=\{a\in A:h^\star(z)=za-h(a)\},
 ~\mbox{with}~
 h^\star(z)
 :=
 \sup_{a\in A} \{za-h(a)\},\;
 z\in\R.
 \end{equation}
The following assumption is needed for our main results concerning the second-best contracting problem.
\begin{assumption}\label{hyp:existence}
$F^\prime$ is concave, and
\begin{itemize}
\item if $\delta\in(1/\gamma, 1)$, then $A=[0,\bar a],$ $\beta:=h'(0)>0$, and $h\in C^4$ satisfies
\[
\min\big\{(h^{\prime\prime})^2+h^\prime h^{(3)},2(h^{\prime\prime})^2-h^\prime h^{(3)},2(h^{\prime\prime})^3-h^\prime h^{\prime\prime}h^{(3)}-(h^\prime)^2h^{(4)}\big\}\geq 0;
\]
\item if $\delta>1$, then $\lim_{y\to\infty}F^\prime(y)/\big(yF^{\prime\prime}(y)\big)$ exists.
\end{itemize}
\end{assumption}

\noindent Notice that $\overline{F}^\prime$ inherits the concavity of $F^\prime$. To see this, recall from standard convex duality that the concavity of $F^\prime$ is equivalent to that of $(F^*)^\prime$. Then, it follows from the expression of $\overline{F}^*$ in \eqref{Fbarstar} that $(\overline{F}^*)^\prime$ is concave, which then implies that $\overline{F}^\prime$ is concave. We shall provide more comments on these assumptions in \Cref{rem:assumptions} below. 

\begin{theorem}\label{thm:existence}{\rm (Second-best contracting)} Let {\rm \Cref{hyp:existence}} hold.
\begin{enumerate}
\item[$(i)$] Let $\delta\gamma\le 1$. Then $\overline V^{\rm P}=\bar a$, and there is no optimal contract achieving this value.
\item[$(ii)$] Let $\delta\gamma>1$, and define the stopping region $\Sc:=\{v=\overline F\}$. Then, under {\rm\Cref{hyp:existence}}, we have   
\begin{enumerate}
\item[{$(ii$-$1)$}]
the principal's value function $V$ is the unique viscosity solution of \eqref{DPE:main}, in the class of functions such that $0\le (v-\overline F)(y)\le C$, $y\ge 0$, for some $C>0.$
\item[{$(ii$-$2)$}] $V\in C^1([0,+\infty))$, is strictly concave, ultimately decreasing, and $C^2$ on $[0,+\infty)$, except at at most one point. Moreover $V^\prime(0)=V^\prime(0)\1_{\{\delta\le 1\}}\geq 0$, and whenever $F^\prime(0)=0$ and $\mathfrak{I}(0,\overline F^{\prime\prime}(0))>0$, we have $V^\prime(0)>0$.
\item[{ $(ii$-$3)$}] if $\beta>0$ and $\delta\leq 1$, then there is always some $y_{\rm gp}\in\big[0,\big(\overline F^\prime\big)^{-1}(-\frac{1}{\delta\beta})\big]$ such that $[y_{\rm gp},+\infty)\subset \Sc.$
\item[{ $(ii$-$4)$}] more generally, $\Sc=\{0\}\cup[y_{\rm gp},\infty)$ for some $y_{\rm gp}\in[0,\infty];$
\item[{$(ii$-$5)$}] if $\Sc=\{0\}\cup[y_{\rm gp},\infty)$ for some $y_{\rm gp}<\infty$, then
\[
\overline{V}^{\rm P}=\sup_{y\geq u(R)}V(y),
\]
and the supremum is attained at some $\hat y\geq u(R)$. Let $\hat z:[0,\infty)\longrightarrow \R$ be a $($measurable$)$ maximiser of $\mathfrak{I}(V^\prime,V^{\prime\prime})$, and $\hat\pi:[0,\infty)\longrightarrow \R$ a $($measurable$)$ minimiser of $F^\star(\delta V^\prime)$, then there exists a unique weak solution to the {\rm SDE} corresponding to $\widehat Y:=Y^{\hat y, \hat z(\hat Y),\hat \pi (\hat Y)}$ $($see {\rm\Cref{u(xi)}} for the definition of this process$)$. In particular, defining
\[
\hat\tau:=\inf\big\{t\geq 0:\widehat Y_t\not\in(0,y_{\rm gp})\big\},
\]
the contract $\big(\hat\tau, \hat \pi(\widehat Y),u^{-1}(\widehat Y_{\hat\tau})\big)$ is an optimal contract for the relaxed principal problem \eqref{Principalbar}.
\end{enumerate}
\end{enumerate}
\end{theorem}

\medskip
\begin{remark}\label{rem:2ndbest}
$(i)$ {\rm \Cref{thm:existence}}.$(ii$-$2)$ states that $V$ is non-increasing when $\delta>1$ so that, unlike the findings of {\rm\citeauthor*{sannikov2008continuous}} when $\delta=1$, no informational rent exists in this case, \emph{i.e.} the agent's value function is never above the requested participation value $u(R)$. An informational rent can only exist when $\delta\in(1/\gamma,1]$, and it does exist when $F^\prime(0)=0$ and $\overline{F}^{\prime\prime}(0)$ is sufficiently small.

\medskip
\noindent $(ii)$ In the case $\delta>1$, {\rm \citeauthor*{sannikov2008continuous} \cite[pp. 959]{sannikov2008continuous}} mentions that if {\rm\textquote[{}][,]{the agent had a higher discount rate than the principal, then with time the principal's benefit from output outweighs the cost of the agent's effort}} and that {\rm \textquote[{}][.]{it is sensible to avoid permanent retirement by allowing the agent to suspend effort temporarily}}. Our result shows that, under the assumption that $\overline F^\prime$ is concave, this statement is not correct.

\medskip
\noindent $(iii)$ As explained in {\rm \Cref{rem:FB}}.$(ii)$, {\rm \Cref{thm:existence}}.$(ii$-$3)$ does not directly imply that when $\delta\in(1/\gamma,1]$ and $\beta>0$, a \emph{Golden Parachute} automatically exists, since even if $y_{\rm gp}$ is then always finite, it could still be that $y_{\rm gp}=0$. If indeed, $y_{\rm gp}>0$, then the claim made by {\rm\citeauthor*{sannikov2008continuous}} in the case $\delta=1$ that a \emph{Golden Parachute} always exists is correct, and can be extended to $\delta\in(1/\gamma,1]$ and $\beta>0$. However, in the case $\delta>1$, it is unclear in general whether $y_{\rm gp}$ is finite or not, and we do not have sufficient conditions under which it is $($though we present conditions under which $y_{\rm gp}=0$ in {\rm \Cref{sec:nogp}}$)$.

\medskip
\noindent $(iv)$ The case $\Sc=\{0\}$ is not covered by {\rm \Cref{thm:existence}.$(ii$-$3)$--$(ii$-$5)$}, due to the fact that in this case, the optimal retirement time $\tau^\star$ may be infinite with positive probability, and therefore cannot satisfy the integrability requirement in {\rm \Cref{integrability}}. This is however not a critical issue. Indeed, these integrability conditions on admissible stopping times are taken from the general result in {\rm \cite{lin2020random}}. But a detailed reading of their arguments shows that they only require it in order to be able to treat moral hazard problems where the agent is allowed to control the volatility of the output process, for which they need a theory for second-order {\rm backward SDEs} with random horizon, which is obtained in {\rm \citeauthor*{lin2020second} \cite{lin2020second}}, but does not allow for infinite horizon. In our problem of interest, the agent only controls the drift of $X$, meaning that the classical theory of {\rm backward SDEs} is sufficient, and these objects are known to be well-posed even with infinite horizon, see for instance {\rm \citeauthor*{papapantoleon2016existence} \cite{papapantoleon2016existence}}. With these results in hand, we can straightforwardly extend the general reduction result of {\rm \Cref{sect:reduction}} to include possibly infinite retirement times, and then obtain a verification result general enough to cover these situations. As this is not central to our message, we refrain to go to this level of generality.
\end{remark}

\begin{remark}\label{rem:assumptions}
$(i)$ The concavity of $F^\prime$ means that the risk aversion of the agent is large enough $($in the power utility setting of {\rm\Cref{ex:power}}, it is equivalent to $\gamma\geq 2)$. Mathematically, this assumption allows to prove that the value function of the principal inherits this property, which in turn is a sufficient condition to ensure that whenever it touches $\overline F$, it coincides with it forever, see {\rm\Cref{lem:existence}.$(iv)$}. Economically, this means that the principal never retires temporarily the agent. 

\medskip The \emph{raison d'\^etre} of this assumption is that it is crucial in our proof of regularity for the value function of the principal. We emphasise that our results could be generalised to accommodate the possibility for the principal to temporarily retire the agent when his continuation utility belongs to finitely many intervals, since our main issue is that without it, we cannot in general rule out the possibility that $v$ and $\overline F$ could have uncountably many contact points on some compact interval, and the regularity becomes then a very unclear issue, at least to us. We have however refrained from doing this, since we do not have a satisfactory sufficient condition under which we could prove that there are strictly more than one, but finitely many contact points.

\medskip
\noindent $(ii)$ The case $\delta\in(1/\gamma,1)$ requires additional conditions in order to ensure that the value function of the principal has a concave first-order derivative, see \emph{Step $3$} in the proof of {\rm\Cref{lemma:regdelta1}} $(1/\gamma<\delta<1)$. The condition $\beta>0$ allows to prove using {\rm\Cref{th:firstbest}}.$(ii$-$3)$ that $v=\overline F$ outside a compact set, which in turn allows us to prove the required technical comparisons theorems in {\rm \Cref{sec:comp}} on such a compact set. We believe that the argument can be extended to the general case $\beta=0$, but we refrained from doing so for simplicity.
\end{remark}
\begin{example}
The condition on the derivatives of $h$ in {\rm\Cref{hyp:existence}} is satisfied if and only
\begin{itemize}
\item $\beta(q-2)=0$ and $q\geq 3/2$, for $\displaystyle h(a):=\frac{a^q}q+\beta a$, $a\in A$;
\item $\alpha^2\geq \beta$, for $ \displaystyle h(a):=\frac{a^3}3+\alpha\frac{a^2}2+\beta a,\; a\in A$.
\end{itemize}
\end{example}

\subsection{Some cases of non-existence of a Golden Parachute}\label{sec:nogp}

This section specialises the discussion to the economic question of whether a Golden Parachute is optimal in the second-best optimal contracting problem. In particular, our results here contrast with the main findings of \citeauthor*{sannikov2008continuous} \cite{sannikov2008continuous} that a \emph{Golden Parachute} always exists whenever $\beta:=h^\prime(0)>0$ and $\delta=1$. We exploit here the observation that $\mathbf{L}\overline{F}=-\mathfrak{I}(\overline{F}^\prime,\overline{F}^{\prime\prime})^+$, so that $\mathbf{L}\overline{F}=0$ if and only if $\mathfrak{I}(\overline{F}^\prime,\overline{F}^{\prime\prime})\le 0$. As both $\overline{F}$ and $\overline{F}^\prime$ are concave in our context, it follows that the map $\mathfrak{I}(\overline{F}^\prime,\overline{F}^{\prime\prime})$ is non-increasing on $\R_+$. Consequently the existence of a Golden Parachute means that $\overline F(y)$ solves \Cref{DPE:main} for sufficiently large $y$, but fails to solve it on a right-neighbourhood of $0$. Equivalently, there is No Golden Parachute (NGP) in either one of the following cases
\begin{enumerate}[label=$({\rm NGP}\arabic*)$,leftmargin=5em]
\item $\mathfrak{I}(\overline{F}^\prime,\overline{F}^{\prime\prime})>0$, on $[0,\infty)$, so that in this case $V>\overline F$ on $(0,\infty);$\label{ngp1}
\item $\mathfrak{I}(\overline{F}^\prime,\overline{F}^{\prime\prime})(0)\le 0$, so that in this case $V=\overline F$ on $\R_+$.\label{ngp2}
\end{enumerate}

\noindent We observe that \ref{ngp2} may hold in the context of \Cref{thm:existence} where it is stated that $V=\overline{F}$ for values larger than $(\overline F^\prime)^{-1}(-\frac{1}{\delta\beta})$, so that there is no Golden Parachute despite the fact that $V$ coincides with $\overline{F}$, see \Cref{ex:sann} below.
The following result provides a sufficient conditions for \ref{ngp1} to hold.
\begin{proposition}\label{prop:NGP}
Let $F^\prime$ be concave, $h^\prime(0)=0$, and $A\supset[0,a_o]$, for some $a_o\in(0,\bar a]$. Then {\rm \ref{ngp1}} holds.
\end{proposition}

\begin{proof}
Since $A$ contains an interval, $h$ is strictly convex, $\overline{F}$ is concave, and $h^\prime(0)=0$, we have that
\begin{align*}
\mathfrak{I}\big(\overline{F}^\prime(y),\overline{F}^{\prime\prime}(y)\big)^+
 &=\sup_{z\ge0,\; \hat a\in\hat A(z)} \big\{\hat a+h(\hat a)\delta\overline{F}^\prime(y)+\eta z^2 \delta\overline{F}^{\prime\prime}(y)\big\}\\
 &\geq\sup_{a\in [0, a_o]}\big\{a+h(a)\delta\overline{F}^\prime(y)+\eta \big(h^\prime(a)\big)^2 \delta\overline{F}^{\prime\prime}(y)\big\},\; y>0.
\end{align*}
Now notice that since $h^\prime(0)=0$, the derivative at $a=0$ of the map inside the supremum above is equal to $1>0$. Therefore, this map is increasing on a right-neighbourhood of $0$, and thus for any $y>0$, we have $\mathfrak{I}\big(\overline{F}^\prime(y),\overline{F}^{\prime\prime}(y)\big)^+>0$.

\end{proof}
\begin{remark}
Assume for simplicity that $F^\prime(0)=0$, then $\overline{F}^\prime(0)=0$ by {\rm\Cref{prop:Fbar}}. Then 
\[
\sup_{z\ge h^{\text{\fontsize{4}{4}\selectfont $\prime$}}(0)} \big\{z^{-2}\max\hat A(z)\big\}\le \frac{\bar a}{{h^\prime(0)^2}},\; \mbox{with}\; \bar a:=\max A,
\] 
so that, as $\overline{F}^{\prime}$ is concave, the existence of a {\rm Golden Parachute} implies that $h^\prime(0)^2<\bar a/\big(-\eta \overline{F}^{\prime\prime}(0)\big)$. In other words, {\rm \Cref{prop:NGP}} states that {\rm \ref{ngp1}} holds for sufficiently large $h^\prime(0) $.
\end{remark}

\begin{example}\label{ex:sann}
{\rm \citeauthor*{sannikov2008continuous} \cite[Figure 1]{sannikov2008continuous}} considers the situation $\delta=1$, $($so that $\overline{F}=F)$, and
 \[
 F(y)=-y^2,\; y\ge 0,
 \; h(a):=\frac12 ha^2 +\beta a,
\;
 a\in A=\R_+,
 \; \mbox{\rm for some positive constants $h$ and $\beta$}.
 \]
Notice that, given {\rm\citeauthor*{sannikov2008continuous}}'s conclusion that a {\rm Golden Parachute} exists, the unboundedness of $A$ is not problematic, as the optimal effort remains bounded, so that the problem is unchanged by restricting to the corresponding compact subset of $A$. Under the present specification, we have \[
F^{\prime\prime}(y)=-2, \; \text{\rm and}\; \sup_{z\ge\beta} \big\{z^{-2}\max\hat A(z)\big\}=\sup_{a\ge 0}\bigg\{ \frac{a}{(a+\beta)^2h}\bigg\}=\frac1{4h\beta}.
\] 
Then, since $F^\prime(0)=0$ in this case, \ref{ngp2} holds if and only if $8\beta\eta h \ge 1.$
\end{example}

\subsection{Numerical illustration}\label{sec:num}

We next provide some numerical results with the cost of effort function from \Cref{ex:sann}, and utility function $u(\pi):=\pi^\gamma,$ $\gamma>1$. We of course choose the model parameters so that \ref{ngp2} is not satisfied, since in this case the solution is $\overline F$ everywhere. \Cref{fig:sann} takes the parameters in \cite{sannikov2008continuous} (with $\gamma=2$, $\eta=0.05$, $h=0.5$, $\beta=0.4$, and $\delta=1$), and shows the archetypical case where a Golden Parachute exists. 
%
\begin{figure}[!ht]
\begin{center}
    \includegraphics[width=0.5\textwidth,height=6cm]{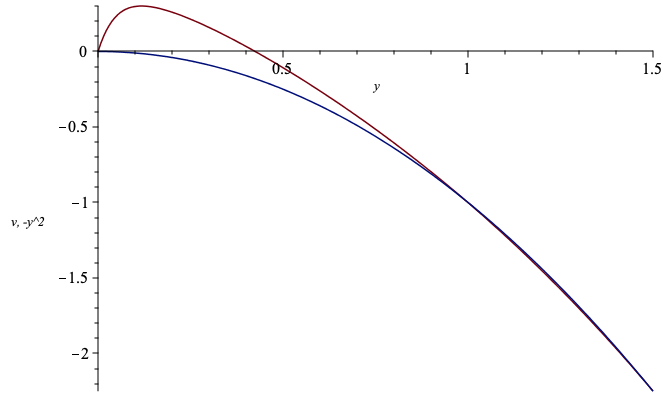}
    \caption{$v$ (red), $F$ (blue)}
    \label{fig:sann}
\end{center}
  \end{figure}
  
The next two sets of figures show what happens when $\delta\neq 1$. More precisely, \Cref{fig:GP3/4} (with $\gamma=3/2$, $\eta=h=1$, $\beta=0.01$, and $\delta=3/4$) shows a case where $v$ becomes equal to $\overline F$ after a while and a Golden parachute does exist, while, at least numerically, \Cref{fig:noGPd} (with $\gamma=3$, $\eta=h=1$, $\beta=0.01$, and $\delta=2$), seems to show that $v$ remains always above $\overline F$, and that no Golden Parachute exists.

\begin{figure}[!ht]
\begin{center}
\begin{subfigure}[b]{0.41\textwidth}
 \includegraphics[width= \textwidth,height=6cm]{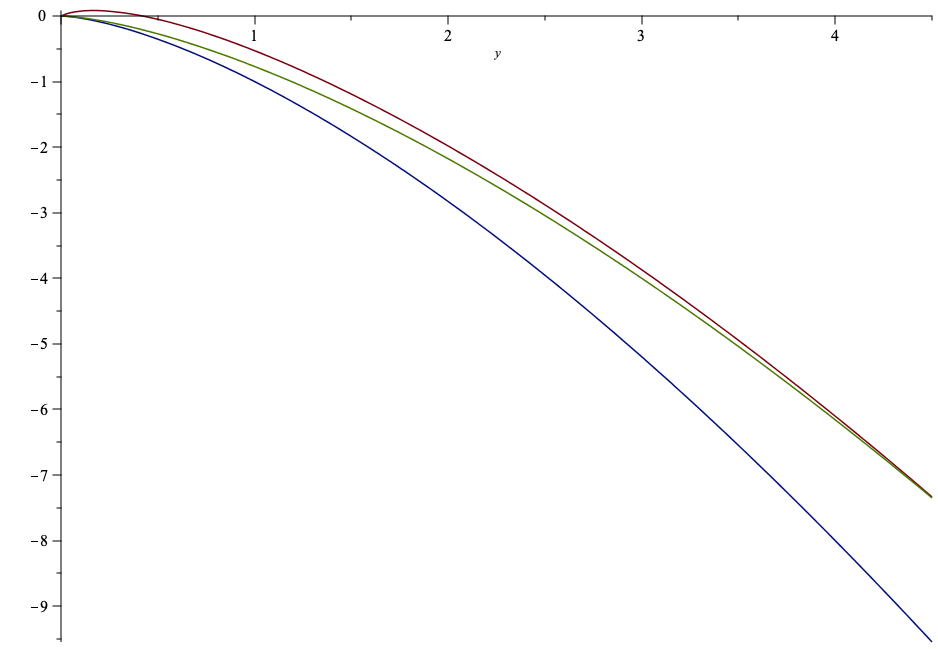}
    \caption{$v$ (red), $\overline F$ (green), $F$ (blue)}
    \label{fig:GP3/4}
\end{subfigure}
  \begin{subfigure}[b]{0.46\textwidth}
 \includegraphics[width= \textwidth,height=6cm]{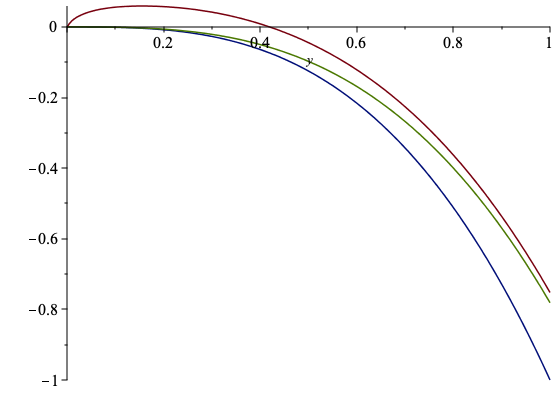}
    \caption{$v$ (red), $\overline F$ (green), $F$ (blue)}
    \label{fig:noGPd}
      \end{subfigure}
      \end{center}
\end{figure}

\subsection{Sannikov's solution}\label{sec:gap}

In this subsection, we specialise the discussion to the case $\delta=1$ as in \cite{sannikov2008continuous}. Notice that the {\rm HJB} equation considered by {\rm \citeauthor*{sannikov2008continuous}} in {\rm \cite[Equation $(5)$]{sannikov2008continuous} } is the same as our {\rm \Cref{DPE:main}} when restricted to the continuation region
 \[
 v- y v^\prime + F^\star(v^\prime) - \mathfrak{I}(v,v^\prime)^+=0,
 \;y\in[0,y_{\rm gp}],\;
 v(0)=F(0), 
 \;v(y_{\rm gp})=F(y_{\rm gp})\;\mbox{and}\;v'(y_{\rm gp})=F^\prime(y_{\rm gp}),
 \]
which corresponds to the natural guess that the stopping region $\mathcal{S}=\{v=F\}$ is of the form $\{0\}\cup [y_{\rm gp},\infty)$, with some free boundary point $y_{\rm gp}<\infty$ to be determined so as to guarantee that the smooth-fit condition $v^\prime(y_{\rm gp})=F^\prime(y_{\rm gp})$ holds. Such a guess is more naturally justified by the optimal stopping component of the principal's problem in our formulation. We shall also see that it is necessary in order to apply the verification argument of \Cref{prop:verif} below (which in fact requires $C^1$ regularity and $C^2$ except at a zero measure set of points).

\medskip
A few pages later, namely in {\rm \cite[Equation $(6)$]{sannikov2008continuous}}, the author rewrites this {\rm ODE} with $\mathfrak{I}$ instead of $\mathfrak{I}^+$
 \begin{equation}\label{Sannikov0}
 v-y v^\prime + F^\star(v^\prime) - \mathfrak{I}(v,v^\prime)=0,
 \;y\in[0,y_{\rm gp}],\;
 v(0)=F(0), 
 \;v(y_{\rm gp})=F(y_{\rm gp})\;\mbox{and}\;v^\prime(y_{\rm gp})=F^\prime(y_{\rm gp}).
 \end{equation}
This is motivated by the natural guess that the principal is expected to induce a positive effort for the agent on the continuation region. More importantly, direct manipulations allow to reformulate the last equation equivalently as
 \begin{equation}\label{I0:inverse}
 v^{\prime\prime}
 =
 \inf_{z\ge h^{\text{\fontsize{4}{4}\selectfont $\prime$}}(0),\;\hat{a}\in\hat{A}(z)}\bigg\{
 \frac{v-yv^\prime+F^\star(v^\prime)-\hat{a}-h(\hat{a})v^\prime}
        {\eta z^2}\bigg\},
 \end{equation}
thus reducing the equation to an explicit non-linear second-order ODE under the additional restriction of a positive marginal cost of effort, that is to say when $h^\prime(0)>0$.
 
 \medskip
Next, assuming that $y_{\rm gp}<\infty$, the potential explosion of the solution due to the super-linear feature of $F^\star$ is bypassed, as the concavity of $v$ implies that $v^\prime$ is bounded in $[v^\prime(y_{\rm gp}),v^\prime(0)]$.  \Cref{thm:existence}.$(ii$-$3)$ shows that this claim is true, modulo the fact that it may happen that $y_{\rm gp}=0$. Then, it follows from the standard Cauchy--Lipschitz theorem that the last ODE, with initial data $v(0)=0$ and $v^\prime(0)=b$, has a unique classical solution for any choice of $b$, say $v_b$. Next, \citeauthor*{sannikov2008continuous} argues that it is possible to choose $b$ so that this solution $v_b$ indeed solves \Cref{Sannikov0}. \citeauthor*{sannikov2008continuous}'s proof of this claim is not rigorous, and we show in the subsequent analysis that we can circumvent this difficulty by a different approach, under the extra assumption that $F^\prime$ is concave. Notice also that our main results given in {\rm \Cref{sec:complete}} and {\rm \Cref{sec:nogp}}, show that
\begin{itemize}
\item for $\beta=0$, there is no $y_{\rm gp}\geq 0$ such that the solution of the dynamic programming equation \eqref{DPE:main} agrees with $F$ on $[y_{\rm gp},\infty)$, see \ref{ngp1} of \Cref{prop:NGP};
\item for $\beta>0$ sufficiently small, we prove under additional conditions that the solution of \eqref{DPE:main} may exhibit the behaviour claimed by \citeauthor*{sannikov2008continuous}. In fact, we shall prove that the stopping region $\mathcal{S}$ is either reduced to $\{0\}$, or is of the form $\{0\}\cup[y_{\rm gp},\infty)$ for some $y_{\rm gp}\ge 0$. This requires some involved technical arguments which are displayed in \Cref{sect:solution} below;
\item when the curvature at zero $u^{\prime\prime}(0)$ of the agent's utility is sufficiently large negative, the stopping region $\mathcal{S}$ is always reduced to $\{0\}$ for whatever value of $\beta>0$. See \ref{ngp2} of \Cref{prop:NGP}.
\end{itemize}

\section{Solving the first-best contracting problem}\label{sec:solvfb}

\begin{proof}[Proof of Theorem \ref{th:firstbest}.$(i)$] We first solve the first best contracting problem when $\delta\gamma>1$. Notice that in view of \eqref{SBleFB}, the result follows directly from \Cref{thm:existence}.$(i)$. We however provide a simpler proof when $\delta\gamma<1$.First, the limited liability constraint on the payments made to the agent, and the fact that $A$ is bounded by $\bar a$ imply immediately that for any $(\mathbf{C},\alpha)\in\mathfrak{C}^{\rm FB}\times\Ac$, we have $
J^{\rm P}(\mathbf{C},\alpha)\leq \bar a.$ Moreover, the only way this can be an equality is to choose $\alpha=\bar a$, and $\mathbf C:=(\tau,\pi,\xi)$ such that $\pi=0$ and $F(\xi)\mathrm{e}^{-\rho\tau}=0$, which means that either $\tau=\infty$, or $\xi=0$. As none of these contracts can satisfy the participation constraint of the agent, there is no admissible contract which achieves this bound $\overline{a}$.

\medskip
For any $\eps>0$, let us consider the following contract: $\tau^\eps:=-\log(\eps)/\eps$, $\pi^\eps:=0$, $\xi^\eps:=\eps^{-1}\mathrm{e}^{\gamma r\tau^{\text{\fontsize{4}{4}\selectfont $\eps$}}}$, with the level of effort $\alpha^\eps:=\bar a$. Since these contracts are defined by deterministic components, they automatically satisfy the integrability condition of \eqref{integrability}. Notice also that when $\eps$ goes to $0$, both $\tau^\eps$ and $\xi^\eps$ converge to $\infty$. Therefore, we can choose $\eps$ small enough and find a constant $C>0$, independent of $\eps$, such that $u(\xi^\eps)\geq C(\xi^\eps)^{1/\gamma}$. The utility received by the agent is then
\[
\mathrm{e}^{-r\tau^{\text{\fontsize{4}{4}\selectfont $\eps$}}}u\big(\xi^\eps\big)-h(\bar a)\big(1-\mathrm{e}^{-r\tau^{\text{\fontsize{4}{4}\selectfont $\eps$}}}\big)\geq \frac C{\eps^\frac1\gamma}-h(\bar a)\underset{\eps\to 0}{\longrightarrow}\infty,
\]
so that the agent's participation constraint is satisfied for $\eps$ small enough. The principal's utility is
\[
-\mathrm{e}^{-\rho\tau^{\text{\fontsize{4}{4}\selectfont $\eps$}}}\xi^\eps+\bar a\big(1-\mathrm{e}^{-\rho\tau^{\text{\fontsize{4}{4}\selectfont $\eps$}}}\big)=\mathrm{e}^{\rho(1-\delta \gamma)\frac{\log(\eps)}\eps}\eps^{ -1}+\bar a\big(1-\mathrm{e}^{-\rho\tau^{\text{\fontsize{4}{4}\selectfont $\eps$}}}\big)\underset{\eps\to 0}{\longrightarrow}\bar a,
\]
since $\delta\gamma< 1$, which ends the proof in this case.
\end{proof}

\begin{proof}[Proof of Theorem \ref{th:firstbest}.$(ii)$] 
We start with $(ii$-$1)$--$(ii$-$2)$. By the standard Karush--Kuhn--Tucker method, we rewrite the first-best problem as
 \begin{align*}
V^{\rm P,FB}
&=
\inf_{\lambda\le 0}
 \bigg\{
\lambda u(R)
 +\sup_{(\tau,\alpha)\in\Tc\times \Ac}
 \E^{\P^\alpha}\!\bigg[
   -\mathrm{e}^{-\rho\tau}F^\star\big(\lambda\mathrm{e}^{(\rho-r)\tau}\big)
   +\!\int_0^\tau \rho\mathrm{e}^{-\rho t} \big(G^\star-F^\star\big)\big(\delta\lambda\mathrm{e}^{(\rho-r)t}\big)
                                  \mathrm{d}t
 \bigg]
 \bigg\}
 \\
 &=
 \inf_{\lambda\le 0}
 \Big\{\lambda u(R)+\sup_{T\ge 0} f_T(\lambda)\Big\},
 \;\text{\rm with}\;
f_T(\lambda):= -\mathrm{e}^{-\rho T}F^\star\big(\lambda\mathrm{e}^{(\rho-r)T}\big)
   +\int_0^T \rho\mathrm{e}^{-\rho t} \big(G^\star-F^\star\big)\big(\delta\lambda\mathrm{e}^{(\rho-r)t}\big)
                                  \mathrm{d}t.
\end{align*}
As $G^\star\geq 0$, and $F^\star$ is concave, we have 
\begin{align*}
\frac{\partial f_T(\lambda)}{\partial T}
&\ge\rho\mathrm{e}^{-\rho T}
       \big(F^\star\big(\lambda\mathrm{e}^{(\rho-r)T}\big)
              -F^\star\big(\lambda\delta\mathrm{e}^{(\rho-r)T}\big)
              -\lambda(1-\delta)\mathrm{e}^{(\rho-r)T}
                                           (F^\star)^\prime\big(\lambda\mathrm{e}^{(\rho-r)T}\big)
       \big)
\ge 0
,\; T\geq 0.
\end{align*}
Then, $\sup_{T\ge 0} f_T(\lambda)=f_\infty(\lambda)
 :=
 \int_0^\infty \rho\mathrm{e}^{-\rho t} \big(G^\star-F^\star\big)\big(\delta\lambda\mathrm{e}^{(\rho-r)t}\big)\mathrm{d}t
 =
 (\overline G^\star-\overline F^\star)(\lambda)$, and therefore 
 \[
 V^{\rm P,FB}
 =
 \inf_{\lambda\le 0}\big\{\lambda u(R)-(\overline F^\star-\overline G^\star)(\lambda)\big\}
 =
 (\overline F^\star-\overline G^\star)^\star\big(u(R)\big).
 \]
\medskip
In particular, $V^{\rm P,FB}|_{R=0}=\inf_{\lambda\le 0}(\overline G^\star-\overline F^\star)(\lambda)$. Notice that $\overline G^\star-\overline F^\star$ is convex and $(\overline G^\star-\overline F^\star)^\prime(0)=(G^\star-F^\star)^\prime(0)=G^\star(0)=\bar a>0$, where we used the fact that $(F^\star)^\prime(0)=(F^\prime)^{-1}(0)=0$, and the observation that $G^\star(p)=\bar a+ph(\bar a)$ for $p\ge 0$ small enough. Then $V^{\rm P,FB}|_{R=0}=(\overline G^\star-\overline F^\star)(0)=(G^\star-F^\star)(0)=G^\star(0)=\bar a$. The remaining claims follow from the fact that $\lim_{\lambda\to\infty}(\overline G^\star-\overline F^\star)(\lambda)=\infty$, as $G^\star$ has linear growth ($A$ is compact), and $F^\star$ grows as $(-p)^{\gamma/(\gamma-1)}$ at $-\infty$. 

\medskip
For $(ii$-$3)$, by the non-negativity of $\overline G^\star$ and convex duality between $\overline F$ and $\overline F^\star$, we have for any $y\geq 0$
\[
0\leq v^{\rm FB}-\overline F(y)=\inf_{p\leq 0}\big\{yp-\overline{F}^\star(p)+\overline G^\star(p)\big\}-\inf_{p\leq 0}\big\{yp-\overline{F}^\star(p)\big\}\leq \overline G^\star\big(\overline F^\prime(y)\big),
\]
which proves the first part of the desired results. To prove the remaining claims, we consider several cases.

\medskip
- For $\delta\le 1$ and $\beta>0$. Notice that $G^\star(p)=0$ for $p\leq -\beta^{-1}$, and $G^\star(p)=\bar a + ph(\bar a)$ for $p\geq -\big(h^\prime(\bar a)\big)^{-1}.$ In particular, when $p\leq -(\beta\delta)^{-1}$ it follows from the fact that $\delta\le 1$ that $\overline G^\star(p)=0$. As $\overline F$ is strictly concave and satisfies similar growth conditions at $+\infty$ as $F$, it cannot behave as an affine function for $y$ large, and thus $\lim_{y\to+\infty}\overline F^\prime(y)=-\infty$. Hence, for $y\geq \big(\overline F^\prime\big)^{-1}\big(-(\delta\beta)^{-1}\big)$, we have $\overline G^\star\big(\overline F^\prime(y)\big)=0$, and $v^{\rm FB}$ and $\overline F$ coincide as required.

\medskip
- In the case $\delta>1$ and $\beta>0$, we compute directly for $p$ sufficiently negative
\[
\overline G^\star(p)=(\delta|p|)^{-\frac1{\delta-1}}\int_{-\beta^{-1}}^{-(h^\prime(\bar a))^{-1}}|x|^{-1-\frac{1}{1-\delta}}G^\star(x)\frac{\mathrm{d}x}{1-\delta}+\bar a\big(-\delta h^\prime(\bar a)p\big)^{-\frac1{\delta-1}}+ph(\bar a)\big(-\delta h^\prime(\bar a)p\big)^{-\frac\delta{\delta-1}},
\]
which goes to $0$ as $p$ goes to $-\infty$ as claimed. 

\medskip
- For $\delta>1$ and $\beta=0$, the same computations as above lead to
\[
\overline G^\star(p)=(\delta|p|)^{-\frac1{\delta-1}}\int_{\delta p}^{-(h^\prime(\bar a))^{-1}}|x|^{-1-\frac{1}{1-\delta}}G^\star(x)\frac{\mathrm{d}x}{1-\delta}+\bar a\big(-\delta h^\prime(\bar a)p\big)^{-\frac1{\delta-1}}+ph(\bar a)\big(-\delta h^\prime(\bar a)p\big)^{-\frac\delta{\delta-1}}.
\]
The only unclear term is the first one. As $\lim_{-\infty}G^\star=0$, we may find for all $\eps>0$ a constant $M_\eps>0$ such that for any $G^\star\leq \eps$ on $(-\infty,-M_\eps]$. We then have for $p\le -M_\eps$
\begin{align*}
0&\leq (\delta|p|)^{-\frac1{\delta-1}}\int_{\delta p}^{-(h^\prime(\bar a))^{-1}}|x|^{-1-\frac{1}{1-\delta}}G^\star(x)\mathrm{d}x
\\
&\leq (\delta|p|)^{-\frac1{\delta-1}}\int_{-M_\eps}^{-(h^\prime(\bar a))^{-1}}|x|^{-1-\frac{1}{1-\delta}}G^\star(x)\mathrm{d}x
+\eps(\delta|p|)^{-\frac1{\delta-1}}\int_{\delta p}^{-M_\eps}|x|^{-1-\frac{1}{1-\delta}}\mathrm{d}x
\\
&=(\delta|p|)^{-\frac1{\delta-1}}\int_{-M_\eps}^{-(h^\prime(\bar a))^{-1}}|x|^{-1-\frac{1}{1-\delta}}G^\star(x)\mathrm{d}x
+\eps(\delta-1)-\eps(\delta|p|)^{-\frac1{\delta-1}}(\delta-1)M_\eps^{\frac1{\delta-1}}
\underset{p\searrow-\infty}{\longrightarrow}\eps(\delta-1).
\end{align*}
Then $\lim_{p\to-\infty}\overline{G}^\star(p)=0$ by arbitrariness of $\eps$. 

\medskip
- It remains to study the case $\delta\le 1$ and $\beta=0$. The case $\delta=1$ is obvious as $\overline{G}^\star=G^\star$. For $\delta<1$, we can again fix some $\eps>0$, and some $N_\eps>0$ such that 
\[
0\leq\overline G^\star(p)\leq\frac{\eps(\delta|p|)^{\frac1{1-\delta}}}{1-\delta}\int_{-\infty}^p|x|^{-1-\frac1{1-\delta}}\mathrm{d}x=\eps\delta^{\frac1{1-\delta}},
\;\mbox{for all}\;p\le -N_\eps.
\]
This implies again have $\lim_{p\to-\infty}\overline{G}^\star(p)=0$ by the arbitrariness of $\eps>0$.
\end{proof}

\section{Reduction to a mixed control-and-stopping problem}
\label{sect:reduction}

In order to solve the second-best contracting problem, we use the general approach of \citeauthor*{lin2020random} \cite{lin2020random}\footnote{See \Cref{foot:LRTY}. The methodology developed in \cite{lin2020random} extends the finite maturity setting of \citeauthor*{cvitanic2018dynamic} \cite{cvitanic2018dynamic} and is largely inspired by the method developed in \citeauthor*{sannikov2008continuous} \cite{sannikov2008continuous}.} which justifies the remarkable solution approach introduced by \citeauthor{sannikov2008continuous} \cite{sannikov2008continuous}, reducing the Stackelberg game problem of the principal \eqref{Principalbar} to a standard stochastic control one. 

\medskip
To do this, observe that the Hamiltonian of the agent's problem is given by the convex conjugate function $h^\star$ introduced in \eqref{hstar}, and that the corresponding subgradient contains all possible optimal agent responses
 \[
 \hat{A}(z):=
 \partial h^\star(z)
 =
 \{a\in A:h^\star(z)=za-h(a)\}.
 \]
As $A$ is closed and $h$ is strictly convex\footnote{If $A$ is an interval, then the strict convexity of $h$ guarantees that $\hat A(z)$ is a singleton. However, for a general closed subset $A$, the maximiser may not be unique.}, notice that 
 \begin{equation}\label{hatA}
 \hat{A}(z)\neq\emptyset,
 \;\mbox{whenever}\;
 h^\star(z)<\infty,
 \;\mbox{and}\;
 \hat A(z)=\{0\},
 \;\mbox{for}\;z< h^\prime(0),
 \end{equation}
because $a\longmapsto za-h(a)$ is decreasing whenever $z<h^\prime(0)$. We also abuse notations slightly, and for any $\F$-predictable, real-valued process $Z$ and any $\alpha\in\Ac$, we write $\alpha\in\hat A(Z)$ whenever $\alpha_t\in\hat A(Z_t)$, $\mathrm{d}t\otimes\mathrm{d}\P^0$--a.e. Then, the lump-sum payment $\xi=u^{-1}(\zeta)$ promised by the principal at $\tau$ takes the form
 \be\label{u(xi)}
 \zeta
 = 
 Y^{Y_{\text{\fontsize{4}{4}\selectfont $0$}},Z,\pi}_\tau
 =
 Y_0+ r\int_0^\tau Z_t \mathrm{d}X_t+\int_0^\tau\big(Y^{Y_{\text{\fontsize{4}{4}\selectfont $0$}},Z,\pi}_t-h^\star(Z_t)-\eta_t\big)\mathrm{d}t,
 \ee
where $Y^{Y_{\text{\fontsize{4}{4}\selectfont $0$}},Z,\pi}$ represents the continuation utility of the agent given a continuous consumption stream $\pi=u^{-1}(\eta)$ and $Z$ satisfies the integrability condition
 \begin{equation}\label{integ:YZ}
 \sup_{\alpha\in\Ac} 
 \E^{\P^\alpha}\bigg[ \sup_{0\leq t\le\tau}\; \big(\mathrm{e}^{-r^{\text{\fontsize{4}{4}\selectfont $\prime$}}t}\big|Y^{Y_{\text{\fontsize{4}{4}\selectfont $0$}},Z,\pi}_t\big|\big)^p\bigg]
 <
 \infty,
 \;\mbox{and}\;
 \sup_{\alpha\in\Ac} 
 \E^{\P^\alpha}\bigg[ \bigg(\int_0^\tau (\mathrm{e}^{-r^{\text{\fontsize{4}{4}\selectfont $\prime$}} t}|Z_t|)^2\mathrm{d}t\bigg)^{\frac{p}{2}}\bigg]
 <
 \infty.
 \end{equation}

\begin{remark}\label{rem:y=0}
As observed by {\rm \citeauthor*{sannikov2008continuous} \cite{sannikov2008continuous}}, notice that the non-negativity condition on $u$ and $h$ implies that the so-called limited liability condition $Y^{Y_{\text{\fontsize{4}{4}\selectfont $0$}},Z,\pi}\ge 0$ is satisfied. Indeed, as the dynamics of the process $Y^{Y_{\text{\fontsize{4}{4}\selectfont $0$}},Z,\pi}$ are given by $\mathrm{d}Y^{Y_{\text{\fontsize{4}{4}\selectfont $0$}},Z,\pi}_t=r\big(Y^{Y_{\text{\fontsize{4}{4}\selectfont $0$}},Z,\pi}_t+h^\star(Z_t)-\eta_t\big)\mathrm{d}t+\sigma rZ_t \mathrm{d}W^0_t$, under the agent's optimal response, we see that $0$ is an absorption point for the continuation utility with optimal effort $0$. 
\end{remark}
\noindent By the main reduction result of \cite{lin2020random},\footnote{\label{foot:LRTY} By the growth condition \eqref{cond:u} on $u$, the integrability condition \eqref{integrability} implies that 
\[
\sup_{\alpha\in\Ac}\E^{\P^\alpha}\bigg[\big(\mathrm{e}^{-r^{\text{\fontsize{4}{4}\selectfont $\prime$}}\tau}u(\xi)\big)^\gamma+\int_0^\tau \big(\mathrm{e}^{-r^{\text{\fontsize{4}{4}\selectfont $\prime$}} s}u(\pi_s)\big)^\gamma\mathrm{d}s\bigg]<\infty,
\] 
which is precisely the integrability condition required by \citeauthor*{lin2020random} \cite{lin2020random}.}  we may rewrite the principal's problem \eqref{Principalbar} as
 \be\label{Principal-Reduced}
 V^{\rm P}
 =
 \sup_{Y_{\text{\fontsize{4}{4}\selectfont $0$}}\ge u(R)}
 V(Y_0),
 \;\mbox{where}\;
 V(Y_0)
 :=
\sup_{(\tau,Z,\pi)\in\Zc(Y_{\text{\fontsize{4}{4}\selectfont $0$}}),\; \hat a\in\hat A(Z)}
 J(\tau,\pi,Z,\hat a),
 \ee
and 
 \begin{equation}\label{JP-Reduced}
 J(\tau,\pi,Z,\hat a)
 :=
 \E^{\P^{\hat a}}\bigg[ \mathrm{e}^{-\rho\tau}\overline{F}\big(Y^{Y_{\text{\fontsize{4}{4}\selectfont $0$}},Z,\pi}_\tau\big)
                                   +\int_0^\tau\rho\mathrm{e}^{-\rho t}\big(\hat a_t+F(\eta_t)\big)\mathrm{d}t
                         \bigg].
 \end{equation}
Here $\Zc(Y_0)$ is the collection of all triples $(\tau, Z,\pi)$ such that $\tau$, $\hat a$ and $\xi=-\overline{F}(Y^{Y_{\text{\fontsize{4}{4}\selectfont $0$}},Z,\pi}_\tau)$ satisfy the integrability conditions \eqref{integrability}, for some $\hat a\in\hat A(Z)$, and therefore also \eqref{integ:YZ}, together with the limited liability condition $Y^{Y_{\text{\fontsize{4}{4}\selectfont $0$}},Z,\pi}\ge 0$ of \Cref{rem:y=0}.

\medskip
The last control problem only involves the dynamics of $Y^{Y_{\text{\fontsize{4}{4}\selectfont $0$}},Z,\pi}$ under the optimal response of the agent (due to the principal's criterion which does not involve anymore the state variable $X$)
 \begin{equation}\label{SDE:principal}
 \mathrm{d}Y^{Y_{\text{\fontsize{4}{4}\selectfont $0$}},Z,\pi}_t
 =
 r\big(Y_t^{Y_{\text{\fontsize{4}{4}\selectfont $0$}},Z,\pi}+h(\hat a_t)-\eta_t\big)\mathrm{d}t+rZ_t \sigma \mathrm{d}W^{\hat a}_t,
 \;\P^{\hat a}\mbox{--a.s., for all}\;\hat a\in\hat A(Z).
 \end{equation}

\section{Degenerate second-best value for a (very) impatient principal}
\label{sec:SBproof}

We now provide the proof of \Cref{thm:existence}.$(i)$ by using the problem reduction from the previous section. Notice first that whenever $V^{\rm P}=\bar a$, then the same argument as in the proof of \Cref{th:firstbest}.$(i)$ shows that there cannot exist an optimal contract, and that the first-best and second-best value coincide. Our proof is based on an explicit construction of a sequence of contracts following the idea used in the proof of \Cref{th:firstbest}.$(i)$: we want to have a retirement time going to $\infty$, associated with a large lump-sum payment. However, because we are now in the second-best case, we need to offer the agent contracts which are incentive-compatible with the level of effort $\bar a$, meaning that these contracts cannot be deterministic. This can however be achieved by choosing a large enough constant control process $Z$ in \Cref{SDE:principal}. The price to pay now with such contracts is that the continuation utility of the agent may reach $0$ in finite time with positive probability, thus preventing the principal from offering a large lump-sum payment. This thus requires to carefully control the probability of early termination of the contract, and we show that by offering the agent a sufficiently large utility, this probability can be made arbitrarily small.

\medskip
\begin{proof}[{Proof of \Cref{thm:existence}.$(i)$}] Let us fix some $y_0>0$, $z>h^\prime(\bar a)$. It is immediate that in this case $\hat A(z)=\{\bar a\}$. For arbitrary $\eps\in (0,r\wedge 1)$, consider the continuous payment $\pi^\eps_t:=u^{-1}\big(\eps Y^\eps_t\big)$, $t\geq 0$, where $Y^\eps:=Y^{y_{\text{\fontsize{4}{4}\selectfont $0$}}/\sqrt{\eps},z,\pi^{\text{\fontsize{4}{4}\selectfont $\eps$}}}$ is the corresponding continuation utility of the agent, which is given by
\[
Y_t^{\eps}=\frac{y_0}{\sqrt{\eps}}+\int_0^t\big((r-\eps)Y_s^{\eps}+rh(\bar a)\big)\mathrm{d}s+rz\sigma W_t^{\bar a},\; t\geq 0.
\]
Notice that $Y^{\eps}$ is an Ornstein--Uhlenbeck process under $\P^{\bar a}$, whose defining SDE can be solved explicitly
\[
Y_t^{\eps}=\mathrm{e}^{(r-\eps)t}\frac{y_0}{\sqrt{\eps}}+\frac{r}{r-\eps}h(\bar a)\big(\mathrm{e}^{(r-\eps)t}-1\big)+rz\sigma\int_0^t\mathrm{e}^{(r-\eps)(t-s)}\mathrm{d}W_s^{\bar a}
,\; t\geq 0.
\]
Let now $T_0^{\eps}:=\inf\big\{t>0:Y_t^{\eps}=0\big\}$, and consider the contract $\mathbf{C}_\eps$ with retirement time $\tau^{\eps}:=\big(-\frac{\log(\eps)}\eps\big)\wedge T_0^{\eps}$, continuous payments $\pi^\eps$, and terminal payment $\xi^\eps:=u^{-1}\big(Y_{\tau^{\text{\fontsize{4}{4}\selectfont $\eps$}}}^{\eps}\big)$. We know from the general results in \Cref{sect:reduction} that such a contract provides the agent with utility $y_0/\sqrt{\eps}$, which he will accept for $\eps$ small enough, regardless of the level of his participation constraint. Indeed, all the integrability requirements are obviously satisfied here, since $z$ is deterministic, $\tau^{\eps}$ is bounded, and from the explicit formula for $Y^{\eps}$.

\medskip
We now compute the principal's utility induced by this contract
\[
J^{\rm P}(\mathbf{C}_\eps,\bar a)
=
\E^{\P^{\bar a}}\big[\mathrm{e}^{-\rho \tau^{\text{\fontsize{4}{4}\selectfont $\eps$}}}F\big(Y_{\tau^{\text{\fontsize{4}{4}\selectfont $\eps$}}}^{\eps}\big)\big]+\E^{\P^{\bar a}}\bigg[\int_0^{\tau^{\text{\fontsize{4}{4}\selectfont $\eps$}}}\rho\mathrm{e}^{-\rho t}F\big(\eps Y_{t}^{\eps}\big)\mathrm{d}t\bigg]+\bar a \Big(1-\E^{\P^{\bar a}}\big[\mathrm{e}^{-\rho \tau^{\text{\fontsize{4}{4}\selectfont $\eps$}}}\big]\Big).
\]
\emph{Step $1$.} For $\eps<r$, we have $T_0^\eps>\bar T_0^\eps:=\inf\big\{t>0:\bar Y^\eps_t=0\big\}$, where $\bar Y_t^\eps:=y_0/\sqrt{\eps}+rh(\bar a)t+rz\sigma W^{\bar a}_t$, $t\geq 0$. The law of $\bar T_0^\eps$ is well-known (see for instance \citeauthor*{karatzas1991brownian} \cite[Equation (5.13)]{karatzas1991brownian}), and we have
\begin{equation}\label{eq:probht}
\P^{\bar a}\big[T_0^{\eps}<\infty\big]\leq \P^{\bar a}\big[\bar T_0^\eps<\infty\big]=\exp\bigg(-\frac{2h(\bar a)y_0}{r\sigma^2 z^2\sqrt{\eps}}\bigg)\underset{\eps\to0}{\longrightarrow} 0.
\end{equation}
This implies that
\begin{align*}
\E^{\P^{\bar a}}\big[\mathrm{e}^{-\rho \tau^{\text{\fontsize{4}{4}\selectfont $\eps$}}}\big]
=
\mathrm{e}^{\rho \frac{\log(\eps)}{\eps}}\P^{\bar a}\big[T_0^{\eps}=\infty\big]
+\E^{\P^{\bar a}}\Big[\mathrm{e}^{-\rho \tau^{\text{\fontsize{4}{4}\selectfont $\eps$}}}\mathbf{1}_{\{T_{\text{\fontsize{4}{4}\selectfont $0$}}^{\text{\fontsize{4}{4}\selectfont $\eps$}}<\infty\}}\Big]
\leq 
\mathrm{e}^{\rho \frac{\log(\eps)}{\eps}}\P^{\bar a}\big[T_0^{\eps}=\infty\big]
+\P^{\bar a}\big[T_0^{\eps}<\infty\big]
\underset{\eps\to 0}{\longrightarrow} 0.
\end{align*}
\emph{Step $2$.} Next, we have that there exists some $C>0$, which may change value from line to line, but is independent of $\eps$, such that for any $t\in \big[0,T_0^{\eps}\big]$
\begin{align*}
0\leq -\mathrm{e}^{-\rho t}F\big(Y_{t}^{\eps}\big)&\leq C\mathrm{e}^{-\rho t}\big(1+\big|Y_{t}^{\eps}\big|^\gamma\big)
\leq C\mathrm{e}^{-\rho t}\bigg(1+\big(1+\eps^{-\gamma/2}\big)\mathrm{e}^{\gamma(r-\eps) t}+\mathrm{e}^{\gamma(r-\eps)t}\bigg|\int_0^t\mathrm{e}^{-(r-\eps)s}\mathrm{d}W^{\bar a}_s\bigg|^\gamma\bigg).
\end{align*}
\begingroup
\allowdisplaybreaks
Then, as the last stochastic integral is a Gaussian random variable, and $\delta\gamma\leq 1$, we see that
\begin{align}\label{eq:FB11}
\notag 0
&\leq 
-\E^{\P^{\bar a}}\big[\mathrm{e}^{-\rho \tau^{\text{\fontsize{4}{4}\selectfont $\eps$}}}F\big(Y_{\tau^{\text{\fontsize{4}{4}\selectfont $\eps$}}}^{\eps}\big)\big]\\
\notag
&\leq 
C\E^{\P^{\bar a}}\bigg[ \mathrm{e}^{-\rho \tau^{\text{\fontsize{4}{4}\selectfont $\eps$}}}\bigg(1+\big(1+\eps^{-\gamma/2}\big)\mathrm{e}^{\gamma(r-\eps)\tau^{\text{\fontsize{4}{4}\selectfont $\eps$}}}+\mathrm{e}^{\gamma(r-\eps)\tau^\text{\fontsize{4}{4}\selectfont $\eps$}}\bigg|\int_0^{\tau^{\text{\fontsize{4}{4}\selectfont $\eps$}}}\mathrm{e}^{-(r-\eps)s}\mathrm{d}W^{\bar a}_s\bigg|^\gamma\bigg)\bigg]
\\
\notag
&\leq 
C\mathrm{e}^{\rho\frac{\log(\eps)}\eps }\bigg(1+\big(1+\eps^{-\gamma/2}\big)\mathrm{e}^{-\gamma(r-\eps) \frac{\log(\eps)}\eps}+\mathrm{e}^{-\gamma(r-\eps)\frac{\log(\eps)}{\eps}}\E^{\P^{\bar a}}\bigg[ \bigg|\int_0^{-\frac{\log(\eps)}{\eps}}\mathrm{e}^{-(r-\eps)s}\mathrm{d}W^{\bar a}_s\bigg|^\gamma\bigg]\bigg)
\\
\notag
&\quad + C\E^{\P^{\bar a}}\bigg[\mathbf{1}_{\{T_{\text{\fontsize{4}{4}\selectfont $0$}}^{\text{\fontsize{4}{4}\selectfont $\eps$}}<\infty\}}\mathrm{e}^{-\rho \tau^{\text{\fontsize{4}{4}\selectfont $\eps$}}}\bigg(1+\big(1+\eps^{-\gamma/2}\big)\mathrm{e}^{\gamma(r-\eps)\tau^{\text{\fontsize{4}{4}\selectfont $\eps$}}}+\mathrm{e}^{\gamma(r-\eps)\tau^{\text{\fontsize{4}{4}\selectfont $\eps$}}}\bigg|\int_0^{\tau^{\text{\fontsize{4}{4}\selectfont $\eps$}}}\mathrm{e}^{-(r-\eps)s}\mathrm{d}W^{\bar a}_s\bigg|^\gamma\bigg)\bigg]
\\
\notag
&\leq 
C\mathrm{e}^{\rho\frac{\log(\eps)}\eps }\bigg(1+\big(1+\eps^{-\gamma/2}\big)\mathrm{e}^{-(r-\eps)\gamma \frac{\log(\eps)}\eps}+\mathrm{e}^{-\gamma(r-\eps)\frac{\log(\eps)}{\eps}}\Big(1-\mathrm{e}^{2(r-\eps)\frac{\log(\eps)}{\eps}}\Big)^{\frac\gamma2}\bigg]\bigg)
\\
&\quad 
+ C\big(1+\eps^{-\gamma/2}\big)\P^{\bar a}\big[T_0^\eps<\infty\big]
+ C\E^{\P^{\bar a}}\bigg[\mathbf{1}_{\{T_{\text{\fontsize{4}{4}\selectfont $0$}}^{\text{\fontsize{4}{4}\selectfont $\eps$}}<\infty\}}\bigg|\int_0^{\tau^{\text{\fontsize{4}{4}\selectfont $\eps$}}}\mathrm{e}^{-(r-\eps)s}\mathrm{d}W^{\bar a}_s\bigg|^\gamma\bigg].
\end{align}
It can be checked directly that since $\delta\gamma\leq 1$, the first term on the right-hand side of \Cref{eq:FB11} goes to $0$ as $\eps$ go to $0$. By \eqref{eq:probht}, the second term also goes to $0$ as $\eps$ goes to $0$. Finally, for the third one, it follows from the Cauchy--Schwarz's inequality and Burkholder--Davis--Gundy's inequality that
\begin{align*}
\E^{\P^{\bar a}}\bigg[\mathbf{1}_{\{T_{\text{\fontsize{4}{4}\selectfont $0$}}^{\text{\fontsize{4}{4}\selectfont $\eps$}}<\infty\}}\bigg|\int_0^{\tau^{\text{\fontsize{4}{4}\selectfont $\eps$}}}\mathrm{e}^{-(r-\eps)s}\mathrm{d}W^{\bar a}_s\bigg|^\gamma\bigg]
&\leq 
\Big(\P^{\bar a}\big[T_0^\eps<\infty\big]\Big)^{\frac12}\E^{\P^{\bar a}}\bigg[\bigg|\int_0^{\tau^{\text{\fontsize{4}{4}\selectfont $\eps$}}}\mathrm{e}^{-(r-\eps)s}\mathrm{d}W^{\bar a}_s\bigg|^{2\gamma}\bigg]^{\frac12}\\
&\leq 
C\Big(\P^{\bar a}\big[T_0^\eps<\infty\big]\Big)^{\frac12}\bigg(\int_0^{\infty}\mathrm{e}^{-2(r-\eps)s}\mathrm{d}s\bigg)^{\frac\gamma2}\leq C\Big(\P^{\bar a}\big[T_0^\eps<\infty\big]\Big)^{\frac12}
\underset{\eps\to 0}{\longrightarrow} 0.
\end{align*}
\endgroup
\emph{Step $3$.} It remains to control the continuous payment term
\begin{align*}
0
&\leq 
-\E^{\P^{\bar a}}\bigg[\int_0^{\tau^{\text{\fontsize{4}{4}\selectfont $\eps$}}}\rho\mathrm{e}^{-\rho t}F\big(\eps Y_{t}^{\eps}\big)\mathrm{d}t\bigg]
\\
&\leq 
-\int_0^{\infty}\rho\mathrm{e}^{-\rho t}
                    \E^{\P^{\bar a}} \Big[\mathbf{1}_{\{\eps|Y_{\text{\fontsize{4}{4}\selectfont $t$}}^{\text{\fontsize{4}{4}\selectfont $\eps$}}|\le 1\}}
                                                    F\big( \eps|Y_{t}^{\eps}|\big)
                                             \Big]\mathrm{d}t
 +C\int_0^{\infty}\rho\mathrm{e}^{-\rho t}
                          \E^{\P^{\bar a}}\Big[ \mathbf{1}_{\{\eps|Y_{\text{\fontsize{4}{4}\selectfont $t$}}^{\text{\fontsize{4}{4}\selectfont $\eps$}}|> 1\}}
                                                          \big(1+\eps^\gamma|Y_{t}^{\eps}|^\gamma\big)\Big]\mathrm{d}t
 \\
&\leq -\int_0^{\infty}\rho\mathrm{e}^{-\rho t}
                    \E^{\P^{\bar a}} \Big[F\big(1\wedge |\eps Y_{t}^{\eps}|\big)
                                             \Big]\mathrm{d}t
+C\eps^\gamma\int_0^{\infty}\rho\mathrm{e}^{-\rho t}\E^{\P^{\bar a}}\big[ |Y_{t}^{\eps}|^\gamma\big]\mathrm{d}t+C\int_0^{\infty}\rho\mathrm{e}^{-\rho t}\P^{\bar a}\big[\eps|Y_t^\eps|\geq 1\big]\mathrm{d}t.
\end{align*}
Notice next that we have that for any $t\geq 0$
\begin{align*}
\eps^\gamma|Y_t^\eps|^\gamma\leq C\bigg(\eps^{\gamma/2}+\eps^\gamma\mathrm{e}^{\gamma (r-\eps)t}+\eps^\gamma\bigg|\int_0^t\mathrm{e}^{(r-\eps)(t-s)}\mathrm{d}W^{\bar a}_s\bigg|^\gamma\bigg).
\end{align*}
Therefore, we have since $\gamma>1$
\begin{align*}
0\leq \eps^\gamma\int_0^{\infty}\rho\mathrm{e}^{-\rho t}\E^{\P^{\bar a}}\big[ |Y_{t}^{\eps}|^\gamma\big]\mathrm{d}t&\leq C \eps^{\gamma/2}+ C\rho\frac{\eps^{\gamma-1}}{\gamma}+C\rho\eps^\gamma\int_0^\infty\mathrm{e}^{-\gamma\eps t}\Big(1-\mathrm{e}^{-2(r-\eps)t}\Big)^\frac{\gamma}{2}\mathrm{d}t\leq  C\eps^{\gamma/2}+ 2C\rho\frac{\eps^{\gamma-1}}{\gamma}\underset{\eps\to0}{\longrightarrow}0.
\end{align*}
Finally, since for any $t\geq 0$, $\eps|Y_t^\eps|$ converges $\P^{\bar a}$--a.s. to $0$, it is immediate by dominated convergence that 
\[
\int_0^{\infty}\rho\mathrm{e}^{-\rho t}\P^{\bar a}\big[\eps|Y_t^\eps|\geq 1\big]\mathrm{d}t\underset{\eps\to0}{\longrightarrow}0, \; \text{and}\;  \int_0^{\infty}\rho\mathrm{e}^{-\rho t}\E^{\P^{\bar a}} \Big[F\big(1\wedge |\eps Y_{t}^{\eps}|\big)
                                             \Big]\mathrm{d}t\underset{\eps\to0}{\longrightarrow}0,
\]
which concludes the proof.
\end{proof}

\begin{remark}
When $\delta\gamma<1$, we may follow all the steps of the last proof, but take instead $\pi^\eps=0$. Then, all terms appearing still converge to $0$ when $\eps$ goes to $0$, and we may conclude already in \emph{Step $2$} that $\lim_{\eps \to 0}J^{\rm P}(\mathbf{C}_\eps,\bar a)=\bar a$.
\end{remark}
\section{Dynamic programming equation and verification}
\label{sec:DPE}

This section prepares for the proof of the remaining main results of \Cref{sect:mainresults} by applying the dynamic programming approach to solve the mixed control-and-stopping problem \eqref{Principal-Reduced}--\eqref{JP-Reduced}. \Cref{hyp:existence} is in force throughout. Notice that this problem is stationary in time due to the infinite horizon feature, and the time homogeneity of the dynamics of $Y$. By standard stochastic control theory, together with \Cref{rem:y=0}, the corresponding HJB equation is
 \begin{equation}\label{DPE:Sannikovrho} 
 v(0)=0,
 \; \mbox{and}\;
 \min\big\{ v-\overline{F}
                  ,
                   \mathbf{L} v
         \big\}
 =0, 
 \;\mbox{on}\;(0,\infty),
 \end{equation}
where for any $y>0$
 \[
 \mathbf{L} v(y)
 :=
 v-\delta y v^\prime(y) + F^\star\big(\delta v^\prime(y)\big) - \mathfrak{I}\big(v^\prime(y),v^{\prime\prime}(y)\big)^+
 =
  v(y)-F(y)-\mathbf{T}F\big(y,\delta v^\prime(y)\big)
  -\mathfrak{I}\big(v^\prime(y),v^{\prime\prime}(y)\big)^+,
  \]
and the second order differential operator $\mathfrak{I}$ is as introduced in \Cref{I0}, and can be rewritten thanks to \Cref{hatA} as
  \begin{align}
  \mathfrak{I}(p,q)
  =&
  \infty\mathbf{1}_{\{q>0\}}
  +\mathbf{1}_{\{q\le 0\}}
  \sup_{z\ge h^{\text{\fontsize{4}{4}\selectfont $\prime$}}(0),\; \hat a\in\hat A(z)} \big\{\hat a+\delta h(\hat a)p+\delta \eta z^2 q\big\},
  \;(p,q)\in\R^2,
 \label{I(p,q)}
 \end{align}
where $\mathbf{T}F(y,p) := yp-F(y)-F^\star(p),
 \;
 y\ge 0,\;p\in\R.$ Observe by definition that
 \be\label{TFge0}
 \mathbf{T}F(y,p) \ge 0,
 \;\mbox{and}\;
 \mathbf{T}F\big(y,F^\prime(y)\big)=0,
 \;\mbox{for all}\;
 y\ge 0.
 \ee
Moreover, inspecting the proof of \Cref{prop:Fbar} in \Cref{sec:appendix}, the face-lifted reward function $\overline{F}$ of \eqref{Fbar} satisfies 
 \be\label{ODE:Fbar}
 \overline{F}-F-\mathbf{T}F\big(\cdot,\delta \overline{F}^\prime\big)
 =0,
 \;\mbox{on}\;\R_+.
 \ee 
\begin{remark}\label{rem:DPE}
{$(i)$} Notice that $\mathbf{L}\overline{F}\leq 0$ on $\R_+$, as a direct consequence of \eqref{ODE:Fbar}.

\medskip
\noindent {$(ii)$} {\rm \Cref{DPE:Sannikovrho}} is equivalent to  
 \begin{equation}\label{DPE:Sannikov1}
 v(0)=0,\;
 \mbox{\rm and}\;
 \mathbf{L} v = 0,
 \;\mbox{\rm on}\;(0,\infty),
 \end{equation}
which agrees exactly with {\rm \Cref{DPE:main}}.
Indeed, if $v$ is a solution of \eqref{DPE:Sannikovrho}, then $\mathbf{L} v=0$ on $\Sc^c$, where $\Sc:=\{v=\overline{F}\}$ is the so-called stopping region, and  $\mathbf{L} v=\mathbf{L} \overline{F}\ge 0$ on $\Sc$, which implies that $\mathbf{L} v= 0$ on $\R_+$ by part $\rm (i)$ of the present remark. Conversely, assuming that $\mathbf{L} v=0$, we see that $v=F+\mathbf{T}F(\cdot,\delta v^\prime)+\mathfrak{I}(v^\prime,v^{\prime\prime})^+\ge F+\mathbf{T}F(\cdot,\delta v^\prime)$, and therefore $v$ is a super-solution of \eqref{ODE:Fbar}. By {\rm \Cref{lem:vsupersolODEFbar}}, this implies that $v\ge\overline{F}$, and we conclude that $v$ solves {\rm\Cref{DPE:Sannikovrho}}.
\end{remark}
\noindent We next provide a verification argument which is the standard justification of the importance of the dynamic programming equation \eqref{DPE:Sannikovrho}, and which guides the subsequent technical analysis to solve the contracting problem.
  \begin{proposition}\label{prop:verif}
Let $v\in C^1(\R_+)$ be $C^2$ on $\R_+$ except at a finite number of points.
 
\medskip
\noindent { $(i)$} Assume $v$ is a super-solution of \eqref{DPE:Sannikovrho}, \emph{i.e.} $v(0) \ge 0$, and $\mathbf{L} v \ge 0$. Then $v\ge V$ on $\R_+$. 

\medskip
\noindent {$(ii)$} If $v(0)=0$ and $\mathbf{L} v=0$ on the continuation region $\Sc^c:=\{v>\overline{F}\}$, then $v=V$ under the additional conditions
\begin{enumerate}
\item[$\bullet$] for any $y>0$, there exists a maximiser $\hat z(y)$ of $\mathfrak I(\delta v^\prime,\delta v^{\prime\prime})^+(y)$ such that the {\rm SDE} \eqref{SDE:principal}, with, for any $t\geq 0$, $u(\pi^\star_t):=(F^\star)^\prime\big(\delta v^\prime(Y_t)\big)$, $Z^\star_t:=\hat z(Y_t)$, and $\hat a^\star_t\in \hat A(Z^\star_t)$, has a weak solution$;$
\item[$\bullet$] defining $\tau^\star:=\inf\big\{t\geq 0:Y^{y,Z^\star,\pi^\star}_t\not\in\Sc^c\big\}$, the triplet $(\tau^\star,Z^\star,\pi^\star)$ belongs to $\Zc(Y_0)$.
\end{enumerate}
\noindent {$(iii)$} If in addition $v$ is ultimately decreasing, then the principal's value function is $V^{\rm P}=v(Y_0^\star)$, for some $Y_0^\star\ge u(R)$ with optimal contract $\xi^\star$ given by
 \[
 u(\xi^\star)
 :=
 Y_0^\star + r\int_0^{\tau^{\text{\fontsize{4}{4}\selectfont $\star$}}}Z^\star_t\mathrm{d}X_t+r\int_0^{\tau^{\text{\fontsize{4}{4}\selectfont $\star$}}}\big(Y_t-h^\star(Z^\star_t)-u(\pi^\star)\big)\mathrm{d}t.
 \]
\end{proposition}

\begin{proof}
{$(i)$} We first prove that $v\ge V$. For an arbitrary $Y_0\ge 0$, and $(\tau,Z,\pi)\in\Zc(Y_0)$ with corresponding $\hat a\in\hat A(Z)$, we introduce for any integer $n$, $\tau_n:=\tau\wedge\inf\{t\ge 0:Y_t\ge n\}$, and we directly compute by It\^o's formula that
 \begin{align*}
 v(Y_0)
 &=
 \mathrm{e}^{-\rho\tau_{\text{\fontsize{4}{4}\selectfont $n$}}}v(Y_{\tau_{\text{\fontsize{4}{4}\selectfont $n$}}}) 
 - \int_0^{\tau_{\text{\fontsize{4}{4}\selectfont $n$}}} \mathrm{e}^{-\rho t}\bigg(-\rho v+\partial_tv
                                                  +(y+h(\hat a_t)-u(\pi_t))rv_y
                                                  +\frac12\sigma^2r^2Z_t^2v_{yy}
                                           \bigg)(Y_t)\mathrm{d}t
 - \int_0^{\tau_{\text{\fontsize{4}{4}\selectfont $n$}}}\mathrm{e}^{-\rho t}v_y(Y_t)rZ_t\sigma \mathrm{d}W^{\hat a}_t
 \\
 &\ge
 \mathrm{e}^{-\rho\tau_{\text{\fontsize{4}{4}\selectfont $n$}}}\overline{F}(Y_{\tau_{\text{\fontsize{4}{4}\selectfont $n$}}}) 
 + \int_0^{\tau_{\text{\fontsize{4}{4}\selectfont $n$}}} \mathrm{e}^{-\rho t}\big(\mathbf{L} v(Y_t)+\hat a_t-\pi_t\big)\mathrm{d}t
                           - \int_0^{\tau_{\text{\fontsize{4}{4}\selectfont $n$}}}\mathrm{e}^{-\rho t}v_y(Y_t)rZ_t\sigma \mathrm{d}W^{\hat a}_t                             
 \\
 &\ge
 \mathrm{e}^{-\rho\tau_{\text{\fontsize{4}{4}\selectfont $n$}}}\overline{F}(Y_{\tau_{\text{\fontsize{4}{4}\selectfont $n$}}}) 
 + \int_0^{\tau_{\text{\fontsize{4}{4}\selectfont $n$}}} \mathrm{e}^{-\rho t}\big(\hat a_t-\pi_t\big)\mathrm{d}t
                           -  \int_0^{\tau_{\text{\fontsize{4}{4}\selectfont $n$}}}\mathrm{e}^{-\rho t}v_y(Y_t)rZ_t\sigma \mathrm{d}W^{\hat a}_t.
 \end{align*}
Since $v_y$ is bounded on $[0,\tau_n]$ and $Z$ satisfies \eqref{integ:YZ}, this implies that 
 \[
 v(Y_0)
 \ge
 \E^{\P^{\hat a}}\bigg[
 \mathrm{e}^{-\rho\tau_n}\overline{F}(Y_{\tau_{\text{\fontsize{4}{4}\selectfont $n$}}}) 
 + \int_0^{\tau_{\text{\fontsize{4}{4}\selectfont $n$}}} \mathrm{e}^{-\rho t}\big(\hat a_t-\pi_t\big)\mathrm{d}t 
 \bigg]
 \longrightarrow
 \E^{\P^{\hat a}}\bigg[
 \mathrm{e}^{-\rho\tau}\overline{F}(Y_{\tau}) 
 + \int_0^{\tau} \mathrm{e}^{-\rho t}\big(\hat a_t-\pi_t\big)\mathrm{d}t 
 \bigg],
 \;\mbox{as}\;
 n\longrightarrow\infty,       
 \]
where the last convergence follows from the fact that 
\[
\big|\mathrm{e}^{-\rho\tau_{\text{\fontsize{4}{4}\selectfont $n$}}}\overline{F}(Y_{\tau_{\text{\fontsize{4}{4}\selectfont $n$}}})\big|\leq C\big(1+ \mathrm{e}^{-\rho\tau_{\text{\fontsize{4}{4}\selectfont $n$}}}Y_{\tau_{\text{\fontsize{4}{4}\selectfont $n$}}}^\gamma\big)\le C\bigg(1+\sup_{0\leq t\le\tau}\big\{\mathrm{e}^{-\frac{\rho}{\gamma} t}Y_t\big\}^\gamma\bigg),
\] 
by the estimate stated in \Cref{prop:Fbar}, together with the integrability conditions on $\pi$ in \eqref{integrability} and on $Y$ in \eqref{integ:YZ}. By the arbitrariness of  $(\tau,Z,\pi)\in\Zc(Y_0)$, this shows that $v(Y_0)\ge V(Y_0)$.

\medskip
To prove $(ii)$, we now repeat the previous argument starting from the control $(\tau^\star,Z^\star,\pi^\star)$ introduced in the statement, and denoting $Y^\star$ the induced controlled state process. As $Z^\star_t$ and $u(\pi^\star_t)$ are maximisers of $I(\delta v^\prime,\delta v^{\prime\prime})(Y^\star_t)$ and $F^\star(\delta v^\prime(Y^\star_t))$, respectively, we see that for any $\hat a^\star\in\hat A(Z^\star)$
 \begin{align*}
 v(Y_0)
 =
 \E^{\P^{\hat a^\star}}\bigg[
 \mathrm{e}^{-\rho\tau_{\text{\fontsize{4}{4}\selectfont $n$}}}v(Y^\star_{\tau^\star_{\text{\fontsize{4}{4}\selectfont $n$}}}) 
 + \int_0^{\tau^\star_{\text{\fontsize{4}{4}\selectfont $n$}}} \mathrm{e}^{-\rho t}\big(\hat\alpha^\star_t-\pi^\star_t\big)\mathrm{d}t 
 \bigg]
 &\underset{n\to\infty}{\longrightarrow}
 \E^{\P^{\hat a^\star}}\bigg[
 \mathrm{e}^{-\rho\tau^\star}v\big(Y^\star_{\tau^\star}\big) 
 + \int_0^{\tau^\star} \mathrm{e}^{-\rho t}\big(\hat a^\star_t-\pi^\star_t\big)\mathrm{d}t 
 \bigg]
 \\
 &
 = \E^{\P^{\hat a^\star}}\bigg[
 \mathrm{e}^{-\rho\tau^\star}\overline{F}\big(Y^\star_{\tau^\star}\big) 
 + \int_0^{\tau^\star} \mathrm{e}^{-\rho t}\big(\hat a^\star_t-\pi^\star_t\big)\mathrm{d}t 
 \bigg],
 \end{align*}
since $v=\overline{F}$ on the boundary of $\Sc$. 

\medskip
$(iii)$ Finally, $v$ is concave by \Cref{lem:existence} (since the argument for concavity applies to any continuous viscosity solution). As it is assumed to be ultimately decreasing, the existence of a maximiser $Y^\star_0$ of $v(y)$ on $[u(R),\infty)$ follows, and we obtain that $V^{\rm P}=\sup_{Y_{\text{\fontsize{4}{4}\selectfont $0$}}\ge u(R)}v(y)=v(Y_0^\star)$.
\end{proof}

\noindent We finally prove \Cref{thm:existence}.$(ii$-$5)$ by using the verification result provided in \Cref{prop:verif}, and admitting the regularity of $V$ stated in the remaining items of \Cref{thm:existence}.$(ii)$.

\medskip
\begin{proof}[Proof of \Cref{thm:existence}.$(ii$-$5)$, admitting $(ii$-$2)$]
The existence of $\hat y$ is immediate by the strict concavity of $V$, and the fact that it is ultimately decreasing. Then, the rest of the proof simply requires to check that the assumptions in \Cref{prop:verif} are satisfied here. First of all, notice that the map $\hat z$ is bounded from above since $A$ is compact, and from below by $\beta$, and it is continuous on $\Sc^c$ because $V$ is $C^2$ there. Similarly, the map $\hat \pi$ is bounded on $\Sc^c$, from below by $0$ and from above as well because $\Sc^c$ is a bounded set under our assumptions. The existence of a unique weak solution for $\widehat Y$ is then direct from \citeauthor*{stroock1997multidimensional} \cite[Corollary 6.4.4]{stroock1997multidimensional}.\footnote{The drift of $\widehat Y$ is not bounded as required in \cite[Corollary 6.4.4]{stroock1997multidimensional}, because of the term $r\widehat Y$. However, it suffices to apply the result to $(\mathrm{e}^{rt}\widehat Y_t)_{t\geq 0}$.} Notice in addition that $\widehat Y$ has moments of any order under $\P^0$ (and thus under any $\P^\alpha$, $\alpha\in\Ac$, recall that $A$ is compact). It remains to verify that $\hat \tau$ satisfies \eqref{integrability}. However, $\widehat Y$ is a one-dimensional Markov process for which the boundaries $0$ and $y_{\rm gp}$ are regular and accessible, it is therefore well-known that $\hat\tau$ is finite with probability $1$. Since $A$ is compact, the densities $\mathrm{d}\P^\alpha/\mathrm{d}\P^0$ all have moments of any order, uniformly in $\alpha\in\Ac$, from which it is immediate that \eqref{integrability} holds.
\end{proof}

\section{Analysis of the dynamic programming equation}
\label{sect:solution}

This section provides the proof of \Cref{thm:existence}.$(ii$-$1)$--$(ii$-$4)$. We shall assume throughout that \Cref{hyp:existence} is in force. Our first result below identifies the second-best value function as the unique viscosity solution of \Cref{DPE:Sannikovrho} in an appropriate class, and gives relatively direct properties.

\begin{lemma}\label{lem:existence}
The second-best value function $V$ is the unique continuous viscosity solution of {\rm \Cref{DPE:Sannikovrho}}, in the class of maps $v$ such that $0\le (v-\overline F)(y)\le C$, $y\ge 0$, for some $C>0$. Moreover, $V$ is strictly concave, ultimately decreasing, and satisfies the following properties

\medskip
$(i)$ $0\leq V(y)-\overline F(y)\leq \overline G^\star\big(\overline F^\prime(y)\big),\; y\geq 0;$

\medskip
$(ii)$ $V^\prime_+(0)\ge 0$ for $\delta\le 1$ and $V^\prime_+(0)=0$ for $\delta>1$. Moreover, $V^\prime(0)>0$ whenever $F^\prime(0)=0$ and $\mathfrak I\big(0, \overline F^{\prime\prime}(0)\big)>0;$

\medskip
$(iii)$ if $\beta>0$ and $\delta\leq 1$, then $V=\overline F$ on $[y_o,+\infty)$ for some finite  $y_o\in\big[0,(\overline F^\prime)^{-1}\big(-\frac{1}{\delta\beta}\big)\big];$ 

\medskip
$(iv)$ if $\overline F^\prime$ is concave, then $\big\{V=\overline F\}=\{0\}\cup[y_o,+\infty),\; \text{for some}\; y_o\in[0,+\infty].$ Moreover $y_o=\infty$ when $\beta=0$.
\end{lemma}
 
\begin{proof} We proceed in several steps.

\medskip
\noindent \emph{Step $1$.} By standard arguments in control theory, $V$ is a (discontinuous) viscosity solution of \Cref{DPE:Sannikovrho}, see for instance \citeauthor*{touzi2013optimal} \cite[Theorem 7.4]{touzi2013optimal}. As $\overline F\le V\le V^{\rm FB}$ by \eqref{SBleFB}, the bound $(i)$ for $V$ is inherited from the similar bound \Cref{th:firstbest}.$(ii$-$3)$ on the value of the first best contracting problem, and implies that $V$ is ultimately decreasing, and $0\leq V(y)-\overline F(y)\leq C,$ for some positive constant $C$. The last bounds allow apply the comparison result of \Cref{lem:comparison} to deduce that $V$ is the unique viscosity solution of \Cref{DPE:Sannikovrho}, and that it is continuous.

\medskip
\noindent \emph{Step $2$.} We next prove that $V$ is strictly concave. To prove concavity, suppose to the contrary that $V$ is strictly convex on some non-empty open interval $(y_0,y_1)\subset\R_+$, then we would have that $-V^{\prime\prime}<0$ in the viscosity sense on $(y_0,y_1)$, and thus that $-\mathfrak{I}(V^\prime,V^{\prime\prime})^+=-\infty$ on $(y_0,y_1)$ $($still in the viscosity sense$)$, contradicting the fact that $V$ is a continuous viscosity solution of {\rm \Cref{DPE:Sannikovrho}}. The strict concavity follows the same line of argument as in \cite{sannikov2008continuous}. Suppose to the contrary that $V(y)=b_0+by$ for $y$ in some interval $[y_0,y_1]\subset\R_+$, then
  \[
  b_0+(1-\delta)by+F^\star(\delta b)-\mathfrak{I}(b,0)^+=0,
  \;y\in[y_0,y_1].
  \]
 $\bullet$ If $\delta\neq 1$, this implies that $b=0$, and then $b_0=\mathfrak{I}(0,0)^+=\bar a>0$. In particular $\mathfrak{I}^+=\mathfrak{I}$. We next argue that this ODE is in addition uniformly elliptic. This is immediate when $\beta>0$. For $\beta=0$, we have $[0,\bar a_0]\subset A$ by \Cref{hyp:existence}, and
\begin{align*}
\mathfrak I\big( v^\prime, v^{\prime\prime}\big)^+ &\ge \sup_{z\ge0,\; \hat a\in\hat A(z)\cap[0,\bar a_{\text{\fontsize{4}{4}\selectfont $0$}}]} \big\{\hat a+\delta h(\hat a)v^\prime+\delta \eta z^2 v^{\prime\prime}(y)\big\}=\sup_{a\in [0,\bar a_{\text{\fontsize{4}{4}\selectfont $0$}}]}\Phi(a,y),\;
\Phi(a,y) := a+\delta h(a)v^\prime(y)+\delta \eta h^\prime(a)^2 v^{\prime\prime}(y).
\end{align*}
As $\Phi(0,y)=0$, $\partial_a\Phi(0,y)=1$, for any $y>0$, we see that for any compact subset of $(0,\infty)$, the supremum in $\mathfrak I( v^\prime, v^{\prime\prime})^+$ is attained on $[\eps,\bar a_0]$ for some $\eps>0$, independent of $y$ (but of course depending on the chosen compact set). Hence, the ODE can always be written in explicit form on any compact subset of $(0,\infty)$, and the standard Cauchy--Lipschitz existence and uniqueness theory applies. By uniqueness of the solution of {\rm \Cref{DPE:Sannikovrho}} with boundary condition $v(y_0)=b_0$ and $v^\prime(y_0)=0$, we deduce that $v=b_0$ on $[0,y_0]$, contradicting the boundary condition $v(0)=0$. 

\medskip
\noindent $\bullet$ If $\delta=1$, we also see by the same argument that $v(y)=b_0+by$ on $[0,y_0]$, so that $v(0)=0$ implies that $b_0=0$, and we get $F^\star(\delta b)-\mathfrak{I}( b,0)^+=0$ and therefore $F^\star(\delta b)=\mathfrak{I}( b,0)^+=0$, which again cannot happen.  

\medskip
\noindent \emph{Step $3$.} We now prove the remaining statements. $(i)$ and $(iii)$ are inherited from the similar bounds for $v^{\rm FB}$ stated in \Cref{th:firstbest}.$(ii)$. To obtain $(ii)$, notice that the non-negativity of $\mathfrak I^+$ and $-F^\star$ implies that $F^\star(V^\prime(0))=0$, meaning that $V^\prime(0)\ge 0$ (where $V'(0)$ exists as the right derivative of a concave function), and that $V$ is a viscosity super-solution of
\[
v(y)-\delta yv^\prime(y)\geq 0,\; y>0.
\]
Since $V$ is continuous and admits a right-derivative $V_+^\prime$, the last inequality holds Lebesgue--almost everywhere, i.e. $V(y)-\delta yV_+^\prime(y)\geq 0,\; y>0$. Dividing by $y>0$ and sending $y$ to $0$, we deduce the first part of the claim. Next as $\mathfrak I(p,q)$ is non-decreasing in $p$, we deduce that $0\ge \mathfrak I( V^\prime(0), V^{\prime\prime}(0))\ge\mathfrak{I}(0, V^{\prime\prime}(0))$. However, under our assumptions, $\mathfrak I(0, \overline F^{\prime\prime}(0))>0$. Then it follows from the non-decrease of $\mathfrak I(p,q)$ in $q$ that $V^{\prime\prime}(0)<\overline F^{\prime\prime}(0)$. Then, as $V\ge \overline F$, and $V(0)=\overline F(0)$, we have $V^\prime(0)>F^\prime(0)=0$, as required.

\medskip
It remains to prove $(iv)$. Fix $0<y_0<y_1$, such that $v=\overline F$ on $\{y_0,y_1\}$ and $v>\overline F$ on $(y_0,y_1)$. Since $y_0$ is then a global minimum point of $v-\overline F$, we can use the viscosity super-solution property of $v$ with $\overline F$ as a test function to deduce
\[
0\leq v(y_0)-\delta y_0\overline F^\prime(y_0)+F^\star\big(\delta \overline F^\prime(y_0)\big)-\mathfrak I\big( \overline F^\prime(y_0),\overline F^{\prime\prime}(y_0)\big)^+=-\mathfrak I\big( \overline F^\prime(y_0),\overline F^{\prime\prime}(y_0)\big)^+.
\]
This implies that $\mathfrak I\big( \overline F^\prime(y_0),\overline F^{\prime\prime}(y_0)\big)\leq 0$, or equivalently
\begin{equation}\label{eq:proofnizar1}
\hat a(z)+h\big(\hat a(z)\big)\delta\overline F^\prime(y_0)+\eta z^2\delta\overline F^{\prime\prime}(y_0)\leq 0,\; \forall z\geq \beta,\; \forall \hat a\in \hat A(z).
\end{equation}
Next, there is some $\overline y\in(y_0,y_1)$ such that $\big(v-\overline F\big)(\bar y)=\max_{y\in(y_0,y_1)}\big(v-\overline F\big)(y)$. Using this time the viscosity sub-solution property of $v$ at $\bar y$ with $\overline F$ as a test function, we deduce
\[
0\geq v(\bar y)-\delta \bar y\overline F^\prime(\bar y)+F^\star\big(\delta \overline F^\prime(\bar y)\big)-\mathfrak I\big( \overline F^\prime(\bar y),\overline F^{\prime\prime}(\bar y)\big)^+=\big(v-\overline F\big)(\bar y)-\mathfrak I\big( \overline F^\prime(\bar y),\overline F^{\prime\prime}(\bar y)\big)^+.
\]
Since $\big(v-\overline F\big)(\bar y)>0$, this implies that $\mathfrak I\big( \overline F^\prime(\bar y),\overline F^{\prime\prime}(\bar y)\big)>0$, and
\begin{equation}\label{eq:proofnizar2}
\hat a(\bar z)+h\big(\hat a(\bar z)\big)\delta\overline F^\prime(\bar y)+\eta \bar z^2\delta\overline F^{\prime\prime}(\bar y)\geq 0,
\;\mbox{for some}\;
\bar z\geq \beta,\;\mbox{and}\;\hat a\in\hat A(\bar z).
\end{equation}
Notice also that we cannot have $\hat a(\bar z)=0$, nor $\bar z=0$ (this can be seen by arguing exactly as in the proof of \Cref{lem:existence}, depending on whether $\beta>0$ or not). Using then \Cref{eq:proofnizar1} with $z=\bar z$, we deduce from \Cref{eq:proofnizar2} that
\[
0\leq \delta h\big(\hat a(\bar z)\big)\big(\overline F^\prime(y_0)-\overline F^\prime(\bar y)\big)+\delta\eta \bar z^2\big(\overline F^{\prime\prime}(y_0)-\overline F^{\prime\prime}(\bar y)\big).
\]
Since both $\overline F^\prime$ and $\overline F^{\prime\prime}$ are non-decreasing, and $\hat a(\bar z)$ and $\bar z$ are positive, the above inequality is impossible.

\medskip
Finally, let $\beta=0$. By \Cref{prop:NGP}.$(i)$, $\overline F$ never solves the ODE, and we claim that this implies that $ V> \overline F$ on $(0,\infty)$. Indeed, notice that any contact point $y_0$ of $V$ and $\overline F$ is a local minimiser of the difference $V-\overline F$, so that $ V^\prime=\overline{F}^\prime$ at such a point. Then, as $\mathfrak I^+\ge 0$, it follows from \Cref{DPE:Sannikovrho} that $\mathbf{T}F\big(y_0, \delta\overline F^\prime(y_0)\big)=0$ which cannot happen unless $y_0=0$.
\end{proof}

\begin{remark}\label{rem:DPEstar}
By the strict concavity of $V$, it is natural to introduce the concave dual function $V^\star(p)
 :=
 \inf_{y\ge 0} \; \{yp-V(y)\}$, $p\in\R$. Then, if in addition $V$ is a $C^2$ solution of the dynamic programming equation, $V^\star$ solves the dual equation
 \begin{equation}\label{DPEvstar}
  \mathbf{L}^\star v^\star(p)
  :=
  v^\star(p)-F^\star(\delta p)
  +(\delta-1)p (v^{\star})^\prime(p)
  +\mathfrak{I}\bigg(p,\frac{1}{(v^{\star})^{\prime\prime}(p)}\bigg)^+
  =
  0,
  \; p\in\R.
  \end{equation}
This follows by evaluating \eqref{DPE:Sannikov1} at the point $y=(V^\prime)^{-1}(p)$ and by computing that $V^{\prime\prime}(y)=1/(V^{\star})^{\prime\prime}(p)$.\footnote{Such a transformation can also be conducted if the solution is expressed in the sense of viscosity solutions (as it will be needed later), but one has to be careful as strict convexity is not sufficient, see \citeauthor*{alvarez1997convex} \cite[Proposition 5]{alvarez1997convex} and the remark after its proof.}
\end{remark}
\noindent Our next task here is to prove that the second-best value function has the required regularity so as to apply the verification argument of \Cref{prop:verif}. 
\begin{lemma}\label{lemma:regdelta1}
The second-best value function $V$ is $C^1$ on $[0,+\infty)$, and $C^2$ on $[0,+\infty)$, except at at most one point.
\end{lemma}
\noindent The proof of this regularity result reported separately in the gradually more involved  $\delta=1$, $\delta>1$, and $\delta<1$. 
\begin{proof}[Proof of \Cref{lemma:regdelta1} $(\delta=1)$]
\emph{Step $1$.} Suppose first that $V^\prime_+(0)>0$. Then $V$ is increasing on a right-neighbourhood on $[0,\eps)$, for some $\eps>0$. As $F^\star=0$ on $\R_+$, we see that $V$ is a viscosity solution of $v(y)-y v^\prime(y)-\mathfrak I\big( v^\prime(y), v^{\prime\prime}(y)\big)^+=0,$ $y\in[0,\eps)$. Notice that $V(y)-yV^\prime_{\pm}(y)>0$, by the strict concavity of $V$. Then $V$ is actually a viscosity solution of
\[
v(y)-y v^\prime(y)-\mathfrak I\big( v^\prime(y), v^{\prime\prime}(y)\big)=0.
\;y\in[0,\eps),
\]
Arguing as in the proof of \Cref{lem:existence} \emph{Step $2$}, we see that the latter can be reduced to a uniformly elliptic ODE, which admits a unique classical solution with initial conditions $v(0)=0$ and $v^\prime(0)=V^\prime_+(0)$. In particular, this solution must coincide with $V$, which is therefore $C^2$ on $[0,\eps)$. This extends easily to the interval $[0,\bar y)$ where
\begin{equation}\label{eq:defbary}
\bar y:=\inf\big\{y>\eps/2:\psi_1(y)=0\big\},
\;\psi_\delta(y):=V(y)-\delta yV^\prime_+(y)+F^\star\big(\delta V_+^\prime(y)\big),
\end{equation}
since we have $\mathfrak I\big( V^\prime(y), V^{\prime\prime}(y)\big)>0$ on that interval. Notice also that we must have that $V$ is decreasing on a left-neighbourhood of $\bar y$.

\medskip 
\noindent \emph{Step $2$.} We next consider the alternative case where $V^\prime_+(0)= 0$, see \Cref{lem:existence} (ii). We shall keep the dependence on $\delta$ as the following argument is valid in the case $\delta\ge 1$ and will then be needed for the next proof. As $V\ge \overline F$, we have $0= V^\prime_+(0)\geq \overline F^\prime(0)=F^\prime(0)/\delta$ by \Cref{lemma:facelift}, and we consider two cases:
\begin{itemize}
\item either $F^\prime(0)=0$, then we must also have $V^\prime_+(0)=\overline F^\prime(0)=0$, so that the ODE at the starting point $y=0$ reduces to the second order part. Since $V\ge\overline F$ and $\mathfrak I^+$ is non-decreasing with respect to both its variables, we deduce first that that $\mathfrak I^+\big(0, \overline F^{\prime\prime}(0)\big)\leq 0$, and as both $\overline F$ and $\overline F^\prime$ are concave 
$\mathfrak I^+\big(\overline F^{\prime}(y), \overline F^{\prime\prime}(y)\big)\le\mathfrak I^+\big(0, \overline F^{\prime\prime}(0)\big)\leq 0$. Given that $\overline F$ is solves the ODE $\overline F(y)-y \overline F^\prime(y)=0$, this implies that $\overline F$ solves \Cref{DPE:Sannikovrho} on $\R_+$. By uniqueness, $V=\overline F$ which is therefore $C^2$ as desired; 

\item alternatively, if $F^\prime(0)<0$, then the last argument applies when $\mathfrak I^+(\overline F^\prime(0),\overline F^{\prime\prime}(0))=0$. It remains to consider the case where $\mathfrak I^+(\overline F^\prime(0),\overline F^{\prime\prime}(0))>0$. As $V$ is decreasing, strictly concave, and $\delta\geq 1$ the map $\psi_\delta$ is increasing on a right-neighbourhood of $0$, say $[0,\eps)$, and $\psi_\delta>\psi_\delta(0)=0$ on $(0,\eps)$. 
We may now argue as in the proof of \Cref{lem:existence} \emph{Step $2$} that our equation, with initial conditions $v(0)=0$ and $v^\prime(0)=V^\prime_+(0)$, reduces to a uniformly elliptic ODE, which admits a unique classical solution $V$ on $[0,\bar y)$, where $\bar y$ is defined exactly as in \Cref{eq:defbary} above
\[
\bar y:=\inf\big\{y>\eps/2:\psi_\delta(y)=0\big\}.
\]
\end{itemize}

\noindent \emph{Step $3$.} By direct calculation, we see that
\[
\underset{y\downarrow y^\prime}{\rm liminf}\; \bigg\{\frac{\psi_1(y)-\psi_1(y^\prime)}{y-y^\prime}\bigg\}=-\big(y^\prime-(F^\star)^\prime\big(V^\prime_+(y^\prime)\big)\big)\underset{y\downarrow y^\prime}{\rm limsup}\; \bigg\{\frac{V_+^\prime(y)-V_+^\prime(y^\prime)}{y-y^\prime}\bigg\}.
\]
By concavity of $V$, this implies that a minimum of $\psi_1$ can only be attained at a point $y$ such that $F^\prime(y)=V^\prime_+(y)$, and that the corresponding value is $V(y)- F(y)$. Then, $\psi_1(y)=0$ if we have $V^\prime_+(y)-F^\prime(y)=V(y)- F(y)=0$. Hence, we have $V=F$ at $\bar y$, which by \Cref{lem:existence}$ (iv)$ implies that $V=F$ on $[\bar y,+\infty)$. As a by-product of the last last argument, we have that $V$ is $C^1$ at $\bar y$, and is $C^2$ on $[0,+\infty)\setminus\{\bar y\}$.
\end{proof}

\begin{proof}[Proof of \Cref{lemma:regdelta1} $(\delta>1)$] In this case, it follows from \Cref{lem:existence}.$(iii)$ that $V^\prime_+(0)\leq 0$. Then following \emph{Step $2$} of the previous proof, we either have $F^\prime(0)=0$ and $V=\bar F$, or $F^\prime(0)<0$ and we can then find $\eps>0$ such that $V$ is a classical solution of our nonlinear ODE on $[0,\overline y)$, with $\bar y:=\inf\big\{y>\eps/2:\psi_\delta(y)=0\big\}$.

\medskip
In contrast with the argument in \emph{Step $3$} of the case $\delta=1$, the variations of $\psi_\delta$ for $\delta>1$ seem more difficult to access. Consider instead the dual ODE defined on $(-\infty,\bar p]$ with $\bar p:=V^\prime_+(\bar y)\le 0$:
\[
-\bar v^\star(p)+(1-\delta )p(\bar v^\star)^\prime(p)+F^\star(\delta p)=0,\; p<\bar p,\; \bar v^\star(\bar p)= \bar p\bar y-v(\bar y).
\]
Arguing as in the proof of \Cref{lemma:facelift}, it is easily seen that this ODE has a unique solution given by
\[
\bar v^\star(p):=\overline F^\star(p)-C(\bar y)|p|^{\frac1{1-\delta}},\; p\leq \bar p,
\;\mbox{with}\;
C(\bar y):=\big(\overline F^\star\big(V^\prime_+(\bar y)\big)-\bar yV^\prime_+(\bar y)+V(\bar y)\big)\big|V^\prime_+(\bar y)\big|^{\frac{1}{\delta-1}} \ge 0,
\]
as $\bar v^\star(\bar p)\le\overline F^\star(\bar p)$. Moreover, arguing as in the proof of \Cref{lem:propbarFstar} we see that $\bar v^\star$ is strictly concave and increasing.  By standard convex duality, the corresponding Fenchel transform $\bar v(y):=\inf_{p\leq \bar p}\big\{py-\bar v^\star(p)\big\}$ is the unique strictly concave and decreasing solution of 
\[
\bar v(y)-\delta y\bar v^\prime(y)+F^\star\big(\delta\bar v^\prime(y)\big)=0,\; y>\bar y, \; \bar v(\bar y)=V(\bar y),
\]
and we have that $v^\prime(\bar y)=V^\prime_+(\bar y)$.

\medskip
Next, notice that for any $y\geq \bar y$, we have $\mathfrak I\big(\bar v^\prime(y),\bar v^{\prime\prime}(y)\big)\leq 0$. Indeed, this is true at $\bar y$ by definition, and we claim that $\bar v^\prime$ is concave, which ensures that this remains true for any value of $y$, since $\mathfrak I$ is non-increasing in both its variables. Indeed, notice that $
\bar v^{\prime\prime\prime}(y)=-\frac{\bar v^{\prime\prime}(y)}{(\bar v^\star)^{\prime\prime}(\bar v^\prime(y))}(\bar v^\star)^{\prime\prime\prime}\big(\bar v^\prime(y)\big),$
so that it is equivalent to show that $(\bar v^\star)^\prime$ is concave. Now simply notice that for any $p<\bar p$, we have
\[
(\bar v^\star)^{\prime\prime\prime}(p)=(\overline F^\star)^{\prime\prime\prime}(p)+\frac{\delta (2\delta-1)C(\bar y)}{(1-\delta)^3}|p|^{\frac\delta{1-\delta}-1}\leq 0,
\]
since $(\overline F^\star)^\prime$ is concave (as $\overline F^\prime$ is concave), and $\delta>1$. 

\medskip
Overall, we have proved that $V$ is $C^2$ on $[0,\bar y)$, and we have constructed a $C^2$ solution of \Cref{DPE:Sannikovrho} on $[\bar y,+\infty)$ whose value and derivative at $\bar y$ coincide with those of $V$. This defines a $C^1$ solution of \Cref{DPE:Sannikovrho} which is $C^2$ everywhere except maybe at $\bar y$. It is also obvious that since $\delta>1$, we have 
\[
0\leq \bar v(y)-\overline F(y)\leq 
C(\bar y)|\overline F^\prime(y)|^{-\frac 1{\delta-1}}\longrightarrow 0
~\mbox{as}~y\to\infty.
\]
n particular, we have found a continuous viscosity solution of \Cref{DPE:Sannikovrho} which satisfies the growth condition $0\leq \bar v(y)-\overline F(y)\leq C$, so that it must coincide with $V$ by the comparison result of \Cref{lem:comparison}. 
\end{proof}

\begin{proof}[Proof of \Cref{lemma:regdelta1} $(\frac1{\gamma}<\delta<1)$]

\emph{Step $1$.} By \Cref{lem:existence}.$(iii)$ we have $V^\prime_+(0)\geq 0$ since $\delta<1$.

$\bullet$ First, if $V^\prime_+(0)=0$, then we also have $\overline F^\prime(0)=0$ by \Cref{lemma:facelift}, since $\delta<1$. Since $V\ge \overline F$, and since $\mathfrak I^+$ is non-decreasing with respect to both its variables, we deduce that $\mathfrak I^+(0, \overline F^{\prime\prime}(0))\leq 0,$ so that $\overline F$ solves \Cref{DPE:Sannikovrho} on $\R_+$, as both $\overline F$ and $\overline F^\prime$ are concave. By uniqueness, $V=\overline F$ which is therefore $C^2$ as desired, thus completing the proof in this case. 

$\bullet$ If instead $V^\prime_+(0)>0$. In this case we argue exactly as in \emph{Step $1$} in the proof of \Cref{lemma:regdelta1} to deduce that $V$ is $C^2$ on the interval $[0,\underline y)$, where $\underline y$ is expressed in terms of the function $\psi_\delta$ of \eqref{eq:defbary} as $\underline y:=\inf\big\{y>\eps/2:\psi_\delta(y)=0\big\},$ and that $V$ is decreasing on a left-neighbourhood of $\underline y$.

\medskip
\noindent \emph{Step $2$.} We continue with the case $V'(0)>0$. The idea will now be to regularise $V$. Recall from \Cref{lem:existence} (ii) that $V$ must be equal to $V=\overline F$ on $[\bar y,\infty)$, and consider for all $n\in\N^\star$ the control problem
\begin{equation}\label{eq:vndef}
 v_n(y):=\sup_{(\tau,Z,\pi)\in\Zc_n,\; \hat a\in\hat A(Z)}\E^{\P^{\hat a}}\bigg[\mathrm{e}^{-\rho \tau^{\text{\fontsize{4}{4}\selectfont $y$,$Z$,$\pi$}}}\overline F\big(Y_{\tau^{\text{\fontsize{4}{4}\selectfont $y$,$Z$,$\pi$}}}^{y,Z,\pi}\big)+\int_0^{\tau^{\text{\fontsize{4}{4}\selectfont $y$,$Z$,$\pi$}}}\mathrm{e}^{-\rho t}\big(\hat a_t-\pi_t\big)\mathrm{d}t\bigg],\; y\in[0,\bar y],
\end{equation}
where $\tau^{y,Z,\pi}:=\inf\big\{t\geq 0:Y_t^{y,Z,\pi}\notin(0,\bar y)\big\}$ and $\Zc_n$ is defined as $\Zc(y)$ with the extra requirement that the controls take value in $[1/n,+\infty)$. By standard control theory, $v_n$ is a viscosity solution of the equation
\begin{equation*}
v_n(y)-\delta yv_n^\prime(y)+F^\star\big(\delta v_n^\prime(y)\big)-\mathfrak I_n\big( v_n^\prime(y), v_n^{\prime\prime}(y)\big)=0,\; y\in(0,\bar y),\; v_n(0)=0,\; v_n(\bar y)=\bar F(\bar y),
\end{equation*}
where 
\begin{align}
\mathfrak I_n(p,q)
&=+\infty\mathbf{1}_{\{q>0\}}+\mathbf{1}_{\{q\leq 0\}}\sup_{z\geq \beta\vee1/n,\; \hat a\in\hat A(z)}\big\{\hat a(z)+\delta ph(\hat a(z))+\eta \delta z^2 q\big\},\; (p,q)\in\R^2.
\label{mathfrakIn}
\end{align}
Notice that $\mathfrak I_n\longrightarrow\mathfrak I^+$ as $n=+\infty$. By \Cref{lem:comparisonb}, $v_n$ is in fact the unique viscosity solution of this equation.

\medskip
By the same argument as in the proof of \Cref{lem:existence}, we see that $v_n$ strictly concave on $(0,\bar y)$. Moreover, $v_n\in C^2$ (and uniqueness actually holds in $W^{2,\infty}\big((0,\bar y)\big)$). This follows from \cite[Theorem I.6]{lions1983optimal}, using the following remarks: first we can bound the controls $Z$ and $\pi$ by some arbitrarily large constant, since the optimal one will be bounded as a function of $\tilde v_n$ and its derivatives on the compact set $[0,\bar y]$; second, the ODE \Cref{eq:dppeqn} is uniformly elliptic for any $n\in\N^\star$. By stability of viscosity solutions (see for instance \cite[Section 6]{crandall1992user}), it follows that $v_n$ converges uniformly on $[0,\bar y]$, and thus there is some $\underline y_n$, converging to $\underline y$, above which $v_n$ is non-increasing. By standard convex duality, the Fenchel transform $v^\star_n$ of $v_n$ is $C^1$ on $\big(v^\prime_n(\bar y),v^\prime_n(\underline y_n)\big)$, and it follows from \citeauthor*{alvarez1997convex} \cite[Proposition 5]{alvarez1997convex} (or more precisely the generalisation in \citeauthor*{imbert2006convexity} \cite{imbert2006convexity} which does not require coercitivity conditions), that $v^\star_n$ is a viscosity solution of the ODE
\begin{equation}\label{eq:dualn}
v_n^\star-(1-\delta)p(v_n^\star)^\prime-F^\star(\delta p)+\mathfrak I_n\bigg( p,\frac{1}{(v^\star_n)^{\prime\prime}}\bigg)=0,\;\text{\rm on}\; \big(v^\prime_n(\bar y),v^\prime_n(\underline y\vee \underline y_n)\big),\; v_n^\star\big(v^\prime_n(\bar y)\big)=\bar y,\; v_n^\star\big(v^\prime_n(\underline y\vee \underline y_n)\big)=\underline y\vee\underline y_n.
\end{equation}
Since this is an explicit second-order ODE, the fact that $v_n^\star$ is $C^1$ implies that it is $C^2$ on the above interval, and even $C^3$.

\medskip
Notice also that $V$ must be above $v_n$ by definition. Since these two maps are also concave and coincide at $0$ and $\bar y$, we must have $V^\prime_+(\underline y\vee\underline y_n)\geq v^\prime_n(\underline y\vee\underline y_n)$ and $V^\prime_+(\bar y)\leq v^\prime_n(\bar y)$. Thus, using the Cauchy--Lipschitz theorem, we can uniquely extend $v_n^\star$ to a continuous function on $[V^\prime_-(\bar y),V^\prime_+(\underline y)]$ which is $C^2$ on $(V^\prime_-(\bar y),V^\prime_+(\underline y))$.

\medskip
\noindent \emph{Step $3$.} Clearly, the function $w_n:=-(v_n^\star)^\prime$ is a classical solution of \Cref{eq:DPEdiff} for all $n\ge 1$. Under the conditions of \Cref{hyp:existence}, we verify in \Cref{lemma:fundconc} and \Cref{lem:comparison2} that all the conditions of \cite[Theorem 1]{alvarez1997convex} or more precisely its extension in \cite[Theorem 1]{imbert2006convexity}, are satisfied. Then it follows that $w_n$ is convex, or equivalently, $(v_n^\star)^\prime$ is concave.

\medskip
\noindent \emph{Step $4$.} We now show that $V^\prime$ is concave. By the stability of viscosity solutions again and the comparison results in \Cref{lem:comparisondual,lem:comparison2}, we know that $(v_n^\star)_{n\in\N^\star}$ and $\big((v_n^\star)^\prime\big)_{n\in\N^\star}$ converge uniformly on compact sets respectively to maps $\tilde v^\star$ and $\hat v^\star$. The uniform convergence ensures that $\tilde v^\star$ coincides on $[\underline y,\bar y]$ with $V^\star$, and that $\hat v^\star=(V^\star)^\prime$. By the concavity of $(v_n^\star)^\prime$, this shows in particular that $V^\star$ is $C^1$ and has a concave derivative. Now, it is obvious that any continuous viscosity solution to the PDE satisfied by $V^\star$ has to be strictly concave, and therefore so is $V^\star$. As such, we deduce that $V$ is $C^1$ on $[\underline y,\bar y]$ and that its derivative is concave there, thanks to the concavity of $(V^\star)^\prime$. Thanks to this property, we know that the operator $\mathfrak I$ evaluated along the derivatives of $V$ can only decrease, and thus that $\mathfrak I^+$ remains forever equal to $0$ once it reaches it for the first time at $\underline y$. Arguing as in the proof of \Cref{lemma:regdelta1}, we deduce that $V$ must coincide with $\overline F$ starting from $\underline y$, and that it is $C^1$ at this point and $C^2$ everywhere else, which ends the proof.
\end{proof}

\appendix
\appendixpage

\section{Face-lifted principal's reward}
\label{sec:appendix}

\subsection{Very impatient principal may reduce her loss to zero}

We first consider the case $\rho\ge\gamma r$ of \Cref{prop:Fbar}.$(i)$. Notice that we always have $\overline F(0)=0$, and that since $F$ is non-positive, we have $\overline F\leq 0$. Besides, by our assumptions on $u$, there exists $M>0$ and $C>0$, such that for any $y\geq M$, $F(y)\geq -Cy^\gamma$. Fix some some $y_0>0$ and some $\eps>0$, and consider then the following control $p(t):=\mathbf{1}_{[t^\star,\infty)}(t)\rho\eps y(t),\; t\geq 0,$ where $t^\star$ is the first instant at which $y^{y_{\text{\fontsize{4}{4}\selectfont $0$}},0}$ reaches the value $M$. We immediately have that $y^{y_{\text{\fontsize{4}{4}\selectfont $0$}},p}(t)=y_0\mathrm{e}^{rt}\mathbf{1}_{[0,t^\star)}(t) + M\mathrm{e}^{r(1-\rho\eps)(t-t^\star)}\mathbf{1}_{[t^\star,\infty)}(t),$ $t\geq 0.$ In particular, $T_0^{y_{\text{\fontsize{4}{4}\selectfont $0$}},p}=\infty$, and for $T>t^\star$
\begin{align*}
\overline F(y_0)&\geq \mathrm{e}^{-\rho T}F\big(y^{y_{\text{\fontsize{4}{4}\selectfont $0$}},p}(T)\big)+\int_0^{T} \rho \mathrm{e}^{-\rho t} F\big(p(t)\big)\mathrm{d}t\\
&\geq -CM^\gamma\mathrm{e}^{\gamma r(\rho\eps-1)t^\star-\rho rT(\frac1r-\gamma\rho +\eps\gamma)}-\frac{C\eps^\gamma\rho^{\gamma}M^\gamma\mathrm{e}^{-\rho t^\star}}{\frac1r-\gamma\rho+\gamma\eps}\Big(1-\mathrm{e}^{-\rho r(T-t^\star)(\frac1r-\gamma\rho +\gamma\eps)}\Big)
\;\underset{T\to\infty}{\longrightarrow}\;
\frac{-C\eps^\gamma\rho^{\gamma}M^\gamma\mathrm{e}^{-\rho t^\star}}{\frac1r-\gamma\delta+\gamma\eps},
\end{align*}
by the condition $\rho\ge\gamma r$. As $\gamma>1$, the last limit converges to $0$ as $\eps\searrow 0$.

\subsection{Non-degenerate face-lifted utility}

By standard control theory, and similar to \Cref{rem:DPE}.$(ii)$, we can show that Hamilton--Jacobi equation corresponding to the mixed control-stopping problem $\overline F$ reduces to
\begin{equation}\label{eq:facelift}
w-\delta yw^\prime+F^\star\big(\delta w^\prime\big)=0,
\;\mbox{on}\;(0,\infty),\; 
w(0)=0,
\end{equation}
as the last equation implies that $w-F = \delta yw^\prime-F-F^\star\big(\delta w^\prime\big)\ge 0$. Now, this ODE has $0$ as a trivial solution, and when $\delta=1$, it has a unique strictly concave solution given by $F$. The following lemma addresses the general case.
\begin{lemma}\label{lemma:facelift}
Let $\delta\neq 1$ and $\gamma\delta>1$, Denote by $w^\star$ be the function introduced in \eqref{Fbarstar}, and let $w:=(w^\star)^\star$ be its concave conjugate. Then $w$ is a solution of \eqref{eq:facelift} satisfying $\bar c_0(-1+y^\gamma)\le w(y)\le \bar c_0(-1+y^\gamma)$. Moreover, $w^\prime(0)
=
F^\prime(0)\frac{1}{\delta}\1_{\{\delta\ge 1\}}$.
\end{lemma}

\begin{proof}
Notice first that if $w$ solves \Cref{eq:facelift}, then, whenever $w^\prime(0)$ is finite, by letting $y$ go to $0$, we get $F^\star\big(\delta w^\prime(0)\big)=0$. This implies that $\delta w^\prime(0)\geq F^\prime(0)$, and as $w\ge F$ we deduce that $w^\prime(0)\geq \frac{F^\prime(0)}{1\vee\delta}$, an inequality obviously satisfied when $w^\prime(0)=\infty$, which is the only possible infinite value by concavity of $w$. We consider two cases to verify that
\begin{equation}\label{Fbarprime}
w^\prime(0)
=
f_\delta:=
F^\prime(0)\frac{1}{\delta}\1_{\{\delta\ge 1\}}.
\end{equation}
\begin{itemize}
\item $\delta\ge 1$. Assume to the contrary that $\delta w^\prime(0)>F^\prime(0)$, then $\delta w^\prime>F^\prime(0)$ on $[0,\eps)$ for some $\eps>0$. This in turn implies that $F^\star(\delta w^\prime)=0$ on $[0,\eps)$, and equation \eqref{eq:facelift} reduces to $w(y)-\delta y w^\prime(y)=0$, on $[0,\eps)$, and we get
 \begin{equation}\label{barF:near0}
 w(y)
 =
 {\rm Const.}  \frac{w(y_0)}{y_0^{\frac{1}{\delta}}} y^{\frac1\delta},
 0<y_0\le y<\eps.
 \end{equation}
As $w(0)=0$ and $w\le 0$, we see that $y_0^{-\frac{1}{\delta}}w(y_0)\underset{y_0\searrow 0}{\longrightarrow} 0$. Then $w=0$ on $[0,\eps)$, contradicting the strict concavity of $w$.

\item {$\delta< 1$}. Since $w^\prime(0)\ge F^\prime(0)$, we have again $F^\star(\delta w^\prime)=0$ on $[0,\eps)$, for some $\eps>0$, and by arguing as in the previous case, we arrive to \eqref{barF:near0}, and we see that $w^\prime(0)=0$ is necessary in order to avoid explosion of $w(y_0)/y_0^{\frac{1}{\delta}}$ near $0$.
\end{itemize}
\noindent By \eqref{Fbarprime}, any strictly concave solution of \eqref{eq:facelift} is decreasing, and the Fenchel dual function $w^\star(p):=\inf_{y\geq 0}\big\{py-w(y)\big\}$, whose domain is $(-\infty,f_\delta]$, satisfies the ODE
\begin{equation}\label{ODEwstar}
-w^\star(p)+(1-\delta )p\big(w^\star\big)^\prime(p)+F^\star(\delta p)=0,\; p<f_\delta,\; w^\star(f_\delta)=0.
\end{equation}
This linear ODE has the generic solution, for any $C\in\R$ and $\eps>0$
\begin{equation}\label{eq:sol}
w^\star(p)=(-p)^{-\frac1{\delta-1}}\bigg(C-\frac1{1 -\delta}\int_p^{f_{\text{\fontsize{4}{4}\selectfont $\delta$}}-\eps}\frac{F^\star(\delta x)}{(-x)^{1+\frac{1}{1-\delta}}}\mathrm{d}x\bigg),\; p<f_\delta.
\end{equation}
We next determine $C$ and $\eps$ so as to match boundary condition $w^\star(f_\delta)=0$.

\medskip
\noindent \emph{Case $1$: $f_\delta<0$.}
In this case, we may take $\eps=C=0$, thus reducing to the unique solution
\[
w^\star(p)=\frac{(-p)^{-\frac1{\delta-1}}}{\delta-1}\int_p^{f_{\text{\fontsize{4}{4}\selectfont $\delta$}}}\frac{F^\star(\delta x)}{(-x)^{1+\frac{1}{1-\delta}}}\mathrm{d}x,\; p\leq f_\delta.
\]
In this case, we necessarily have $\delta>1$, therefore, by \Cref{lem:propbarFstar}, $w^\star$ is strictly concave, increasing, and $w^\star\leq F^\star$. This immediately proves that $w$ is unique, strictly concave, decreasing, and above $F$. Besides, the explicit formula we obtained shows by direct integration and using \Cref{cond:Fstar} that $w^\star$ satisfies also \Cref{cond:Fstar} with appropriate constants, which directly implies the required inequalities for $w$.

\medskip
\noindent \emph{Case $2$: $f_\delta=0$ and $\delta>1$.} We can again take $\eps=C=0$ to match the boundary condition $w^\star(0)=0$, as 
\[
\frac{(-p)^{-\frac1{\delta-1}}}{1 -\delta}\int_p^{0}\frac{F^\star(\delta x)}{(-x)^{1+\frac{1}{1-\delta}}}\mathrm{d}x\underset{p\to0-}{=}o(1).
\]
Moreover, by \Cref{lem:propbarFstar}, $w^\star$ is strictly concave, increasing, and $w^\star\leq F^\star$. This immediately proves that $w$ is unique, strictly concave, decreasing, and above $F$. We deduce that $w$ satisfies the required inequalities as in the previous case.

\medskip
\noindent \emph{Case $3$: $f_\delta=0$ and $\delta<1$.} In this case, it can be checked that for any $C\in\R$ and $\eps>0$, we have
\[
\lim_{p\to 0-}(-p)^{-\frac1{\delta-1}}\bigg(C-\frac1{1 -\delta}\int_p^{-\eps}\frac{F^\star(\delta x)}{(-x)^{1+\frac{1}{1-\delta}}}\mathrm{d}x\bigg)=0.
\]
We therefore have, \emph{a priori}, infinitely many possible solutions to the ODE. However, notice that the growth imposed on $w$ translates into $\bar c_0^\star(-1+|p|^{\frac\gamma{\gamma-1}})\le w^\star(p)\le \bar c_1^\star(1+|p|^{\frac\gamma{\gamma-1}})$, and this implies that 
\begin{equation}\label{growth1-delta}
|p|^{-\frac1{1-\delta}}w^\star(p)
\le 
\bar c_1^\star \Big(|p|^{-\frac1{1-\delta}}
                        +|p|^{\frac{1-\gamma\delta}{(1-\delta)(\gamma-1)}}
                 \Big)
\underset{p\to-\infty}\longrightarrow 0,
\end{equation}
as $\gamma\delta>1$. Then, $C$ and $\eps$ must be such that $
C=\frac1{1-\delta}\int_{-\infty}^{-\eps}F^\star(\delta x)(-x)^{-1-\frac{1}{1-\delta}}\mathrm{d}x,
$ where the finiteness of the last integral is satisfied in our setting again by $\gamma\delta>1$. Consequently, $w^\star$ is uniquely determined and given by
\begin{equation}\label{wstardelta<1}
w^\star(p)=\frac{(-p)^{-\frac1{\delta-1}}}{(1-\delta)}
\displaystyle\int_{-\infty}^p\frac{F^\star(\delta x)}{(-x)^{1+\frac{1}{1-\delta}}}\mathrm{d}x,\; p\leq 0,
\end{equation}
and we can again use \Cref{lem:propbarFstar} to conclude. Finally, it follows from a direct change of variable that the above explicit solution coincides with the expression in \Cref{Fbarstar} in all above cases.
\end{proof}

\noindent The following is the easy inequality in a generic verification theorem for $\overline F$.
\begin{lemma}\label{lem:vsupersolODEFbar}
Let $w$ be a $C^2$ super-solution of $w-\delta y w^\prime+F^\star(\delta w^\prime)\ge 0$ on $\R_+$. Then $w\ge\overline{F}$.
\end{lemma}

\begin{proof}
We first observe that $w-F\ge\mathbf{T}F(y,\delta w^\prime)\ge 0$. We next compute for all $p\in\Bc_{\R_{\text{\fontsize{4}{4}\selectfont $+$}}}$ and $T\le T_0^{y_{\text{\fontsize{4}{4}\selectfont $0$}},p}$ that
 \begin{align*}
 w(y_0)
 \ge
 \mathrm{e}^{-\rho T}F\big(y^{y_{\text{\fontsize{4}{4}\selectfont $0$}},p}(T)\big)
 +\int_0^T \rho \mathrm{e}^{-\rho t} F\big(p(t)\big)\mathrm{d}t,
 \end{align*}
by the super-solution property of $w$. The arbitrariness of $p\in\Bc_{\R_{\text{\fontsize{4}{4}\selectfont $+$}}}$ and $T\le T_0^{y_{\text{\fontsize{4}{4}\selectfont $0$}},p}$  implies that $w\ge\overline{F}$.
\end{proof}

\begin{lemma}\label{lem:verifFbar}
Let $\delta\neq 1$ and $\delta\gamma>1$. Assume further that {\rm\Cref{hyp:existence}} holds. Then $\overline{F}=(\overline{F}^\star)^\star$, where $\overline{F}^\star$ is given explicitly in \eqref{Fbarstar}, is the unique solution of {\rm \Cref{eq:facelift}} in the class of  functions satisfying $\bar c_0(-1+y^\gamma)\le\overline{F}(y)\le \bar c_1(-1+y^\gamma)$. Moreover, $\overline{F}$ is a strictly concave decreasing majorant of $F$, with $\overline F^\prime(0)=\delta^{-1} F^\prime(0)\1_{\{\delta>1\}}$.
\end{lemma}

\begin{proof}
We show by a standard verification argument that $w=\overline{F}$ where $w$ is the solution of \eqref{eq:facelift} whose concave dual was derived explicitly in the \Cref{lemma:facelift}. By \Cref{lem:vsupersolODEFbar}, $w$ is an upper bound for $\overline{F}$, i.e. $\overline{F}\le w$. On the other hand, consider for any $y>0$ the maximiser in the definition of $F^\star(\delta w^\prime(y))$, that is to say $(F^\star)^\prime\big(\delta w^\prime(y)\big)$ as a feedback control $\dot y^\star_t = r\big( y^\star_t - p^\star_t\big),$ where $p^\star_t := (F^\star)^\prime\big(\delta w^\prime(y^\star_t)\big).$ Direct differentiation of \eqref{eq:facelift} provides that for any $y>0$, $(1-\delta)w^\prime(y)=\big(y-(F^\star)^\prime\big(\delta w^\prime(y))\big)\delta w^{\prime\prime}(y)$, so that
 \[
 \dot y^\star_t
 =
 r\bigg(\frac1\delta-1\bigg)y^\star_t h(y^\star_t),
 \;t\ge 0,\;\mbox{where}\;h(y):=\frac{w^\prime(y)}{yw^{\prime\prime}(y)},\; y>0.
 \]
Since $h\ge 0$, we see that $y^\star$ is decreasing when $\delta>1$, and is therefore well-defined at least until the hitting time of zero $T^\star:=T_0^{y_{\text{\fontsize{4}{4}\selectfont $0$}},p^{\text{\fontsize{4}{4}\selectfont $\star$}}}<\infty$. In contrast, when $\delta<1$, $y^\star$ is increasing until some explosion time $\bar T$, and $T^\star=\infty$. Following the same calculation as in the first step of the present proof, we see that under the control $p^\star$, all inequalities are turned into equalities, leading for any $T\in [0,\bar T]$ to
 \begin{equation}\label{transversality}
 w(y_0)
 =
 \mathrm{e}^{-\rho T\wedge T^\star}w(y^\star_{T\wedge T^\star})
 +\int_0^{T\wedge T^\star} \rho \mathrm{e}^{-\rho t}F(p^\star_t)\mathrm{d}t.
 \end{equation}
First, by the previous step, when $\delta>1$, we have $T^\star<\infty$, and we obtain by sending $T$ to $\infty$ and using the boundary condition $w(0)=0$ that $(T^\star,p^\star)$ attains the upper bound $w(y_0)$, and is therefore an optimal control for the problem $\overline{F}$. In the alternative case $\delta<1$, we have $T^\star=\infty$. In the rest of this proof, we show that 
 \begin{equation}\label{claim:verifFbar}
 \bar h 
 :=
 \sup_{y\ge \hat y}\{h(y)\} 
 < \frac1{\gamma(1-\delta)},
 \;\mbox{for some}\;
 \hat y>0.
 \end{equation}
Then, $y^\star$ is defined on $\R_+$, \emph{i.e.} $\bar T=\infty$, and since $\hat T:=\inf\{t\ge 0:y^\star_t\ge \hat y\}<\infty$, we deduce from the growth of $w$ that for some $C>0$, whose value may change from line to line, and any $t\geq \hat T$
 \begin{align*}
 \mathrm{e}^{-\rho (t-\hat T)}|w(y^\star_t)|
 \le
 C\mathrm{e}^{-\rho (t-\hat T)}\big(1+|y^\star_t|^\gamma\big)
 &\le 
 C\mathrm{e}^{-\rho (t-\hat T)}\Big(1+\mathrm{e}^{\rho\gamma(1-\delta)\int_{\hat T}^t h(y^\star_{\text{\fontsize{4}{4}\selectfont $s$}})\mathrm{d}s}\Big)
 \le
 C\mathrm{e}^{-\rho (t-\hat T)}\Big(1+\mathrm{e}^{\rho\gamma(1-\delta)\bar h(t-\hat T)}\Big)
 \underset{t\to\infty}{\longrightarrow} 0,
 \end{align*}
as $1-\gamma(1-\delta)\bar h>0$. Sending $T$ to $\infty$ in \eqref{transversality} this provides again that $(T^\star,p^\star)=(\infty,p^\star)$ attains the upper bound $w(y_0)$. In order to verify \eqref{claim:verifFbar}, we prove equivalently that the concave dual $w^\star$ satisfies
\begin{equation} \label{claim:verifFbarstar}
 \sup_{p\le\hat p}\bigg\{\frac{p(w^{\star})^{\prime\prime}(p)}{(w^{\star})^\prime(p)} \bigg\}
 <
 \frac1{\gamma(1-\delta)},\;\mbox{for some}\;
 \hat p<0.
 \end{equation}
Differentiating the ODE \eqref{ODEwstar} satisfied by $w^\star$, and using the expression of $w^\star$ from \Cref{wstardelta<1}, we see that 
\[
\frac{p(w^{\star})^{\prime\prime}}{(w^{\star})^\prime} 
=
\frac{\delta}{1-\delta}
- \Psi(\delta p),
\;\mbox{with}\;
\Psi(p)
:=
\frac{-p\psi^\prime(p)}{\psi(p)},
\;\mbox{and}\;
\psi(p):=(-p)^{-\frac1{1-\delta}}F^\star(p)-\int_{-\infty}^p \frac{F^{\star}(u)}{(1-\delta) (-u)^{1+\frac{1}{1-\delta}}}\mathrm{d}u.
\]
Notice that $\lim_ {p\to-\infty}\psi(p)=0$, and that $\psi$ is easily shown to be non-negative and non-decreasing. Therefore, if $-p\psi^\prime(p)$ does not go to $0$ as $p$ goes to $-\infty$, we have that $\lim_{p\to-\infty}\Psi(\delta p)=\infty$, and \Cref{claim:verifFbarstar} automatically holds. Now if $-p\psi^\prime(p)\underset{p\to-\infty}\longrightarrow 0$, it follows from l'H\^opital's rule and our assumptions that
\[
\lim_{y\to\infty}\bigg\{\frac{F^\prime(y)}{yF^{\prime\prime}(y)}\bigg\}= \lim_{p\to-\infty}\bigg\{\frac{p(F^{\star})^{\prime\prime}(p)}{(F^{\star})^\prime(p)}\bigg\}= \frac{\delta}{1-\delta}
 -\lim_{p\to-\infty}\bigg\{\frac{-\psi^\prime(p)-p\psi^{\prime\prime}}{\psi^\prime(p)}\bigg\}.
\]
Then, using again l'H\^opital's rule, we deduce
 \begin{align*}
 \limsup_{p\to-\infty}\bigg\{\frac{p(w^{\star})^{\prime\prime}(p)}{(w^{\star})^\prime(p)}\bigg\}
 &=
 \frac{\delta}{1-\delta}
 -\liminf_{p\to-\infty}\bigg\{\frac{-p\psi^\prime(p)}{\psi(p)}\bigg\}
 = 
 \frac{\delta}{1-\delta}
 -\liminf_{p\to-\infty}\bigg\{\frac{-\psi^\prime(p)-p\psi^{\prime\prime}}{\psi^\prime(p)}\bigg\}
 =
 \lim_{y\to\infty}\bigg\{\frac{F^\prime(y)}{yF^{\prime\prime}(y)}\bigg\}.
 \end{align*}
Then, assuming to the contrary that \eqref{claim:verifFbarstar} does not hold means that, for fixed $\gamma_0\in(1+\gamma(1-\delta),\gamma)$, we may find $y_0>0$ such that $\frac{F^{\text{\fontsize{4}{4}\selectfont $\prime\prime$}}(y)}{F^{\text{\fontsize{4}{4}\selectfont $\prime$}}(y)}\le (\gamma_0-1)\frac1y$, for $y\ge y_0$. Integrating twice and recalling that $F\le 0$, this implies that 
 \[
 F(y)
 \ge
 F(y_0)+\frac{y_0F^\prime(y_0)}{\gamma_0}\bigg(\bigg(\frac{y}{y_0}\bigg)^{\gamma_{\text{\fontsize{4}{4}\selectfont $0$}}}-1\bigg),
 \;\mbox{for all}\;
 y\ge y_0,
 \]
which in turn leads to the following contradiction $\frac{y_{\text{\fontsize{4}{4}\selectfont $0$}}F^\prime(y_{\text{\fontsize{4}{4}\selectfont $0$}})}{\gamma_{\text{\fontsize{4}{4}\selectfont $0$}}}\le\lim\frac{F(y)}{y^{\gamma_{\text{\fontsize{4}{4}\selectfont $0$}}}}=-\infty$ by our assumption on the growth of $F$ together with the fact that $\gamma_0<\gamma$.
\end{proof}

\noindent We end this section with the result used in the proof of \Cref{lemma:facelift}.
\begin{lemma}\label{lem:propbarFstar}
Let $\delta \neq 1$, and let $\overline{F}^\star$ be a solution of 
 \begin{equation}\label{barFstar}
 -\overline{F}^\star+(1-\delta)p\big(\overline{F}^\star\big)^\prime+F^\star(\delta p)
 =0,\;p< \frac{F^\prime(0)}{\delta}\mathbf{1}_{\{\delta>1\}},\;
 \overline{F}^\star\bigg(\frac{F^\prime(0)}{\delta}\mathbf{1}_{\{\delta>1\}}\bigg)=0.
 \end{equation}
Then $\overline{F}^\star\le F^\star$, $\overline{F}^\star$ is strictly concave and increasing.
\end{lemma}

\begin{proof}
Denote $\phi:=F^\star-\overline{F}^\star$, and notice that \Cref{barFstar} says that for any $p<F^\prime(0)/\delta\mathbf{1}_{\{\delta>1\}}=:f_\delta$
 \[
 \phi(p)
 =
 F^\star(p)-F^\star(\delta p)+(\delta-1)p\big(\overline{F}^\star\big)^\prime(p)
 \ge
 (1-\delta)p\phi^\prime(p),
 \]
by the concavity of $F^\star$. Now define for $p<f_\delta$, $\psi(p):=(-p)^{\frac1{\delta-1}}\phi(p)$, we have $\psi^\prime(p)=\frac{(-p)^{\frac1{\delta-1}-1}}{1-\delta}\big( \phi(p)- (1-\delta)p\phi^\prime(p)\big).$ We need to distinguish two cases, depending on $\delta>1$ or $\delta<1$. First, if $\delta>1$, we have that $\psi$ is non-increasing, and thus
\[
(-p)^{\frac{1}{\delta-1}}\phi(p)\geq \lim_{p\to f_\delta}\Big\{(-p)^{\frac{1}{\delta-1}}\phi(p)\Big\}=0,\; p<f_\delta,
\]
since $\overline{F}^\star(f_\delta)=F^\star(f_\delta)=0$ (recall that $F^\star$ is $0$ above $F^\prime(0)$, which is itself below $f_\delta$, since $\delta>1$). Similarly, when $\delta<1$, by arguing as in \eqref{growth1-delta}, we arrive at the conclusion $\phi\geq 0$, and thus that $\overline{F}^\star\le F^\star$, as desired. Next, by direct differentiation of \Cref{barFstar}, then substituting the expression of $(\overline{F}^\star)^\prime$ from \Cref{barFstar}, and finally using the strict concavity of $F^\star$, we deduce that for any $p<f_\delta$
 \begin{align*}
 (\delta-1)^2p^2(\overline{F}^\star)^{\prime\prime}(p)
  &=
 \delta(\delta-1)p(F^\star)^\prime(\delta p)
 -\delta\big(F^\star(\delta p)-\overline{F}^\star(p)\big)<  \delta\big(\overline{F}^\star(p)-F^\star(p)\big)\leq 0,
 \end{align*}
thus proving the strict concavity of $\overline{F}^\star$. Finally, since $\overline{F}^\star$ is strictly concave, remains below $F^\star$ which is increasing on $(-\infty,F^\prime(0)]$, and $0$ on $[F^\prime(0),f_\delta\wedge F^\prime(0)]$, then $\overline{F}^\star$ must also be increasing on its domain.
\end{proof}

\section{Comparison theorems}\label{sec:comp}

\subsection{The dynamic programming equation and its dual}
This section starts by providing the main technical result which was needed to justify the existence of a solution of the dynamic programming equation in \Cref{lem:existence}. Namely, we show that, despite the exploding feature of $F^\star$, the dynamic programming equation satisfies a comparison result.

\begin{lemma}\label{lem:comparison}{\rm (Comparison)}
Let $u$ and $v$ be respectively an upper-semicontinuous viscosity sub-solution and a lower-semicontinuous viscosity super-solution of \eqref{DPE:Sannikovrho}, such that for $\varphi\in\{u,v\}$ and for some $b>0$, $0\leq \varphi(y)-\overline F(y)\leq b\log\big(1+\log(1+y)\big),\;y\geq 0.$ Then $u\le v$ on $\R_+$.
\end{lemma}
\begin{remark}
The specific upper bound in the statement of {\rm \Cref{lem:comparison}} with an iterated logarithm is not important \emph{per se}. Indeed, the proof goes through as long as the upper bound is of the form $g(y)$ for some positive, increasing, strictly concave map $g$, null at $0$, growing strictly slower at $\infty$ than $\log(y)$.
\end{remark}
\begin{proof}[Proof of {\rm\Cref{lem:comparison}}]
We adapt the arguments of \citeauthor*{crandall1992user} \cite[Section 3]{crandall1992user} to our context. 

\medskip
\noindent \emph{Step $1$.} Let $\mu:=2\vee (\gamma+\nu)$, for some $\nu>0$. Define for any $\alpha>0$ and $\eps>0$ the map
\[
\psi_{\alpha,\eps}(x,y):=u(x) - v(y) -\frac\alpha\mu|x-y|^\mu-\eps \log(1+y),\; (x,y)\in\R_+^2,
\;\mbox{and}\;
M_{\alpha,\eps}:=\sup_{(x,y)\in\R_+^2}\psi_{\alpha,\eps}(x,y)
=\psi_{\alpha,\eps}(x_{\alpha,\eps},y_{\alpha,\eps}),
\]
for some $(x_{\alpha,\eps},y_{\alpha,\eps})\in\R_+^2$, by the growth assumptions on $u$ and $v$. Notice in addition that the supremum is attained on a compact set, then we can find for any $\eps>0$ a further subsequence, denoted by $(x^\eps_n,y^\eps_n)_{n\in\N}:=(x_{\alpha_n,\eps},y_{\alpha_n,\eps})_{n\in\N}$, converging to some $(\hat x^\eps,\hat y^\eps)$. By standard arguments from viscosity solution theory (see for instance \citeauthor*{crandall1992user} \cite[Proposition 3.7]{crandall1992user}), we have
\[
\hat x^\eps=\hat y^\eps,\; \lim_{n\to \infty}\alpha_n\big|x^\eps_n-y^\eps_n\big|^\mu=0,\; M_\eps:=\lim_{n\to\infty}M_{\alpha_n,\eps}= \sup_{ y\geq 0}\; (u-v)(y) -\eps \log\big(1+\hat y^\eps\big).
\]
Let us now assume that there is some $y_o>0$ such that $\eta:=(u-v)(y_o)>0$. Then, we have for any $n\in\N$ and $\eps>0$
\[
\eta-\eps\log(1+y_o)\leq M_{\alpha_n,\eps}=\psi_{\alpha,\eps}(x^\eps_{\alpha_n},y^\eps_{\alpha_n}).
\]
In particular, for $n$ sufficiently large, we have that $x^\eps_n$ and $y^\eps_n$ are both positive, and we assume for notational simplicity that we took the appropriate subsequence. Using Crandall--Ishii's lemma (see \citeauthor*{crandall1992user} \cite[Theorem 3.2]{crandall1992user}), we can find for each integer $n$, an $\R^2$-valued sequence $(X^\eps_n,Y^\eps_n)_{n\in\N}$ such that
\[
\big(\alpha_n(x^\eps_n-y^\eps_n)^{\mu-1},X^\eps_n)\in \overline J^{2,+}u(x^\eps_n),\; \bigg(\alpha_n(x^\eps_n-y^\eps_n)^{\mu-1}-\frac{\eps}{1+y_n^\eps},Y^\eps_n\bigg)\in \overline J^{2,-}v(y^\eps_n),
\]
with the notation $x^\phi:=\mathrm{sign}(x)|x|^\phi$ for all $\phi>0$ and $x\in\R$, and 
\begin{align*}
-\bigg(\frac1\lambda +\|C_n^\eps\|\bigg)\leq\begin{pmatrix}X^\eps_n & 0 \\
0 & -Y^\eps_n\end{pmatrix}\leq C^\eps_n\big(I_2+\lambda C^\eps_n\big),
\;\mbox{for all}\;\lambda>0,
\end{align*}
where $I_2$ is the two-dimensional identity matrix, and 
\[
 C_n^\eps:=a_n^\eps A+b_n^\eps B,\; a_n^\eps:=\alpha_n(\mu-1)|x^\eps_n-y^\eps_n|^{\mu-2},\; b_n^\eps:=\frac{\eps}{(1+y^\eps_n)^2},
 \; A:=\begin{pmatrix} 1 & -1\\-1&1\end{pmatrix},
 \; B:=\begin{pmatrix} 0 & 0\\0&1\end{pmatrix},
\]
and where we use the spectral norm for symmetric matrices. Take $\lambda=\|C_n^\eps\|^{-1}$, and multiply the above inequality by $(1,1)$ to the left and $(1,1)^\top$ to the right, this implies in particular that $X^\eps_n- Y^\eps_n\leq \frac{(b_n^\eps)^2}{\|C_n^\eps\|}+b_n^\eps,\; \mbox{for all}\;n\in\N.$

\medskip
\noindent \emph{Step $2$.} Denoting $G(y,p,q):=-\delta yp+ F^\star(\delta p)-\mathfrak I ( p, q)^+$, it follows from the sub-solution and super-solution properties of $u$ and $v$, that
\[
u(x^\eps_n)+G\big(x_n^\eps,\alpha_n(x^\eps_n-y^\eps_n)^{\mu-1},X^\eps_n\big)
\leq 0
\leq v(y^\eps_n)+G\bigg(y^\eps_n,\alpha_n(x^\eps_n-y^\eps_n)^{\mu-1}-\frac{\eps}{1+y_n^\eps},Y^\eps_n\bigg),
\;\mbox{for all}\;n\in\N^\star.
\]
We deduce that $\eta-\eps\log(1+y_o)\leq u(x^\eps_n)-v(y^\eps_n)\leq F_{n,\eps}$, where
\begin{align*}
F_{n,\eps}:= G\bigg(y^\eps_n,\alpha_n(x^\eps_n-y^\eps_n)^{\mu-1}-\frac\eps{1+y_n^\eps},Y^\eps_n\bigg) - G\big(x_n^\eps,\alpha_n(x^\eps_n-y^\eps_n)^{\mu-1},X^\eps_n\big).
\end{align*}
Now notice that since $\mathfrak I^+$ is Lipschitz continuous and non-decreasing with respect to its second variable, and since $F^\star$ is non-decreasing, we have for some $c_o>0$
\begin{align*}
F_{n,\eps}&= \delta \alpha_n|x^\eps_n-y^\eps_n|^{\mu}+\delta\eps\frac{y_n^\eps}{1+y^\eps_n}+F^\star\bigg(\delta \alpha_n(x^\eps_n-y^\eps_n)^{\mu-1}-\frac{\delta\eps}{1+y_n^\eps}\bigg)-F^\star\big(\delta \alpha_n(x^\eps_n-y^\eps_n)^{\mu-1}\big)\\
&\quad+\mathfrak{I}\big( \alpha_n(x^\eps_n-y^\eps_n)^{\mu-1},X^\eps_n\big)^+ -\mathfrak{I}\bigg( \alpha_n(x^\eps_n-y^\eps_n)^{\mu-1}-\frac{\delta\eps}{1+y_n^\eps},Y^\eps_n\bigg)^+ \\
&\leq  \delta \alpha_n|x_n^\eps-y_n^\eps|^{\mu}+\delta\eps+\mathfrak{I}\bigg( \alpha_n(x^\eps_n-y^\eps_n)^{\mu-1},Y^\eps_n+\frac{(b_n^\eps)^2}{\|C_n^\eps\|}+b_n^\eps\bigg)^+-\mathfrak{I}\bigg( \alpha_n(x^\eps_n-y^\eps_n)^{\mu-1}-\frac{\delta\eps}{1+y_n^\eps},Y^\eps_n\bigg)^+\\
&\leq  \delta \alpha_n|x_n^\eps-y_n^\eps|^{\mu}+(1+c_o)\delta\eps +c_ob_n^\eps +c_o\frac{(b_n^\eps)^2}{\|C_n^\eps\|}.
\end{align*}
We now want to let $n$ go to $\infty$, and will distinguish two cases. 

\medskip
$\bullet$ Either $\big(\|C_n^\eps\|\big)_{n\in\N}$ is unbounded, then by sending $n\to\infty$ along some subsequence, we get $\eta-\eps\log(1+y_o)\leq (1+c_o)\delta\eps+c_o\frac{\eps}{(1+\hat y^\eps)^2},$ which leads to a contradiction when $\eps$ goes to $0$.

\medskip
$\bullet$ Or $\big(\|C_n^\eps\|\big)_{n\in\N}$ is bounded, we take a converging subsequence, and notice that we then have for some $a^\eps\in\R$, $\|C_n^\eps\|\underset{n\to\infty}{\longrightarrow}\|C^\eps\|:=\|a^\eps A+b^\eps B\|,$ with $b^\eps:=\frac{\eps}{(1+\hat y^\eps)^2}.$
\begin{enumerate}[leftmargin=4.5em]
\item If $(a^\eps)_{\eps>0}$ is unbounded, we take again a subsequence and get a contradiction by letting $\eps$ go to $0$ in 
\begin{equation}\label{eq:comp}
\eta-\eps\log(1+y_o)\leq (1+c_o)\delta\eps +c_ob^\eps +c_o\frac{(b^\eps)^2}{\|C^\eps\|}.
\end{equation}
\item If otherwise $(a^\eps)_{\eps>0}$ is bounded and has a subsequence converging to some $a\neq 0$, then taking another subsequence, we obtain a contradiction by letting $\eps$ go to $0$ in \Cref{eq:comp}. Finally, if $\lim_{\eps \to 0}a^\eps=0$, then three cases can occur 
\begin{itemize}
\item[$(i)$] first, if $a^\eps\underset{\eps \to 0}{=}o\big(b^\eps\big)$, then
$\frac{(b^\eps)^2}{\|C^\eps\|}\underset{\eps\to 0}{\sim}b^\eps\underset{\eps\to 0}{\longrightarrow}0,$
and we conclude again by letting $\eps$ go to $0$ in \Cref{eq:comp};
\item[$(ii)$] if instead $b^\eps\underset{\eps \to 0}{=}o\big(a^\eps\big)$, then $
\frac{(b^\eps)^2}{\|C^\eps\|}\underset{\eps\to 0}{\sim}\frac{b^\eps}{a^\eps}b^\eps\underset{\eps\to 0}{\longrightarrow}0,$ and we conclude similarly;
\item[$(iii)$] finally, if $a^\eps\underset{\eps \to 0}{\sim}c b^\eps$ for some $c\neq 0$, then $\frac{(b^\eps)^2}{\|C^\eps\|}\underset{\eps\to 0}{\sim}\frac{b^\eps}{\|cA+B\|}\underset{\eps\to 0}{\longrightarrow}0,$
and we get once more a contradiction.
\end{itemize}
\end{enumerate}
\end{proof}
\noindent We next provide the comparison result for the approximating PDE. Recall the operator $\mathfrak I_n$ from \eqref{mathfrakIn}.

\begin{lemma}\label{lem:comparisonb}{\rm (Comparison in bounded domain)}
Fix some $n\in\N\cup\{+\infty\}$ and some $\bar y>0$. Let $u$ and $v$ be respectively an upper-semicontinuous viscosity sub-solution and a lower-semicontinuous viscosity super-solution of 
\begin{equation}\label{eq:dppeqn}
w(y)-\delta yw^\prime(y)+F^\star\big(\delta w^\prime(y)\big)-\mathfrak I_n\big( w^\prime(y), w^{\prime\prime}(y)\big)=0,\; y\in(0,\bar y),\; w(0)=0,\; w(\bar y)=\bar F(\bar y),
\end{equation}
such that $u(0)\leq v(0)$ and $u(\bar y)\leq v(\bar y)$. Then $u\le v$ on $[0,\bar y]$.

\end{lemma}

\begin{proof}
We argue exactly as in \emph{Step $1$} of the proof of \Cref{lem:comparison}, with $\mu=2$, $\eps=0$, and with maximisation in the definition of $M_\alpha=M_{\alpha,0}$ confined to the compact subset $[0,\bar y]^2$. With $G_n(y,p,q):=-\delta yp+ F^\star(\delta p)-\mathfrak I_n ( p, q)$, the sub-solution and super-solution properties of $u$ and $v$ induce
\[
u(x_m)+G_n\big(x_m,\alpha_m(x_m-y_m),X_m\big)
\leq 0
\leq v(y_m)+G_n\big(y_m,\alpha_m(x_m-y_m),Y_m\big).
\]
Then, whenever there is some $\eta:=(u-v)(y_o)>0$, for some $y_o>0$, we deduce that \begin{align*}
\eta
\leq 
u(x_m)-v(y_m)
\leq 
F_{m}
&:= G_n\big(y_m,\alpha_m(x_m-y_m),Y_m\big) - G_n\big(x_m,\alpha_m(x_m-y_m),X_m\big).
\\
&= \delta \alpha_m|x_m-y_m|^{2}
+\mathfrak{I}_n\big( \alpha_m(x_m-y_m),X_m\big) -\mathfrak{I}_n\big( \alpha_m(x_m-y_m),Y_m\big) 
\end{align*}
Now notice that since $\mathfrak I_n$ is Lipschitz continuous and non-decreasing with respect to its second variable, and since $F^\star$ is non-decreasing, it follows that $F_{m}\leq \delta \alpha_m|x_m-y_m|^{2}
\underset{m\to\infty}{\longrightarrow} 0$, contradiction!
\end{proof}

\noindent We finish this section by proving a comparison result for the convex dual equation of \Cref{eq:dppeqn}.
\begin{lemma}\label{lem:comparisondual}{\rm (Comparison dual equation)}
Fix some $n\in\N\cup\{+\infty\}$ and some $(\underline p,\bar p,\underline w,\bar w)\in\R^4$. Let $u$ and $v$ be respectively an upper semi-continuous viscosity sub-solution and a lower semi-continuous viscosity super-solution of 
\begin{equation}\label{eq:dppeqndual}
w^\star(p)-(1-\delta) p(w^\star)^\prime(y)-F^\star(\delta p)+\mathfrak I_n\bigg(  p,\frac{1}{(w^\star)^{\prime\prime}}\bigg)=0=0,\; p\in(\underline p,\bar p),\; w^\star(\underline p)=\underline w,\; w^\star(\bar p)=\bar w,
\end{equation}
such that $u(\underline p)\leq v(\underline p)$ and $u(\bar p)\leq v(\bar p)$. Then $u\le v$ on $[\underline p,\bar p]$.
\end{lemma}
\begin{proof}
We again argue exactly as in \emph{Step $1$} of the proof of \Cref{lem:comparison}, with $\mu=2$, $\eps=0$, and with maximisation in the definition of $M_\alpha=M_{\alpha,0}$ confined to the compact subset $[\underline p,\bar p]^2$, and $G^\star_n(p,r,q):=-(1-\delta) pr- F^\star(\delta p)+\mathfrak I_n \big( p,\frac{1} q\big)$. By the sub-solution and super-solution properties of $u$ and $v$, we see that
\[
u(x_m)+G_n\big(x_m,\alpha_m(x_m-y_m),X_m\big)
\leq 0
\leq v(y_m)+G_n\big(y_m,\alpha_m(x_m-y_m),Y_m\big),
\;\mbox{for all}\;m\in\N^\star.
\]
Then, whenever there is some $\eta:=(u-v)(p_o)>0$, for some $p_o>0\in[\underline p,\bar p]$, we deduce that 
\begin{align*}
\eta\leq u(x_m)-v(y_m)\leq
F_{m} &:= G_n\big(y_m,\alpha_m(x_m-y_m),Y_m\big) - G_n\big(x_m,\alpha_m(x_m-y_m),X_m\big)
\\
&=
(1-\delta) \alpha_m|x_m-y_m|^{2}+F^\star(\delta x_m)-F^\star(\delta y_m)+\mathfrak{I}_n\bigg( y_m,\frac1{Y_m}\bigg) -\mathfrak{I}_n\bigg( x_m,\frac1{Y_m}\bigg).
\end{align*}
Now $\mathfrak I_n$ is Lipschitz-continuous and non-decreasing with respect to its second variable, and $F^\star$ is Lipschitz continuous on compacts. Then $F_{m}\leq (1-\delta) \alpha_m|x_m-y_m|^{2}+c_o|x_m-y_m|$, for $c_o>0$. The contradiction follows by letting $m\to+\infty$.
\end{proof}

\subsection{The differentiated dual dynamic programming equation}\label{sec:diffDPPeq}
This section provides all the technical results needed to prove that the first-order derivative of the solution to the dual equation \eqref{DPEvstar} is concave under appropriate assumptions, whenever $\delta\in(1/\gamma,1)$. Throughout this section, we assume that $A=[0,\bar a]$, so that we can finally write the operator $\mathfrak I_n$ of \eqref{mathfrakIn} as
\[
\mathfrak I_n(p,q)=+\infty\mathbf{1}_{\{q>0\}}+\mathbf{1}_{\{q\leq 0\}}\sup_{a\in[a_n,\bar a]}\big\{a+\delta ph(a)+\eta\delta (h^\prime(a))^2 q\big\},\; (p,q)\in\R^2,
\;\mbox{with}\;
a_n:=(h^\prime)^{-1}(\beta\vee1/n).
\]
Using the envelop theorem, we also see that $\mathfrak I_n$ is {differentiable} and that for any $(p,q)\in\R\times\R_-$
\[
b_n(p,q):=\partial_p\mathfrak I_n(p,q)= \delta h\big(a^\star_n(p,q)\big),
\;\mbox{and}\;
\gamma_n(p,q):=\partial_q\mathfrak I_n(p,q)=\eta\delta \big(h^\prime(a^\star_n(p,q))\big)^2,
\]
where $a^\star_n(p,q):=(a_n\vee \hat a(p,q))\wedge \bar a$, and $\hat a(p,q)$ is the unique solution of the equation
\begin{equation}\label{eq:critiI}
1+\big(p\delta h^\prime+\eta q\delta [(h^\prime)^2]^\prime\big)\big(\hat a(p,q)\big)=0.
\end{equation}
Our first result shows that under appropriate assumptions on $h$, the functions $b_n$ and $\gamma_n$ are monotone with respect to both their arguments, and that $b_n$ is convex with respect to its first argument.
\begin{lemma}\label{lemma:convexityG}
Assume that $h$ is $C^4$ and that
\[
\min\big\{(h^{\prime\prime})^2+h^\prime h^{(3)},2(h^{\prime\prime})^2-h^\prime h^{(3)},2(h^{\prime\prime})^3-h^\prime h^{\prime\prime}h^{(3)}-(h^\prime)^2h^{(4)}\big\}\geq 0.
\]
Then for any $q\leq 0$, the maps $\R_-\ni p\longmapsto b_n(p,q)$ and $\R_-\ni p\longmapsto \gamma_n(p,q)$ are non-decreasing and convex {\color{black}on the domain where they are not constant}, while for any $p\leq 0$, the maps $\R_-\ni q\longmapsto b_n(p,q)$ and $\R_-\ni q\longmapsto \gamma_n(p,q)$ are non-decreasing.
\end{lemma}

\begin{proof}
First notice that under our assumptions, the maps $A\ni a\longmapsto h(a)$ and $A\ni a\longmapsto (h^\prime(a))^2$ are both increasing and convex. It is therefore enough to prove that for any $q\leq 0$, $\R_-\ni p\longmapsto \hat a(p,q)$ is non-decreasing and convex, and that for any $p\leq 0$, $\R_-\ni q\longmapsto \hat a(p,q)$ is non-decreasing. Differentiating \Cref{eq:critiI} with respect to $p$ and $q$, we get
\[
\partial_p\hat a(p,q)
=
\frac{h^\prime}
       {-ph^{\prime\prime}-\eta q[(h^{\prime})^2]^{\prime\prime}}
       (\hat a(p,q)),\; 
\partial_q\hat a(p,q)
=\frac{\eta (h^{\prime})^2}
         {-ph^{\prime\prime}-\eta q[(h^{\prime})^2]^{\prime\prime}}
         (\hat a(p,q)),
\]
from which we deduce immediately that for any $q\leq 0$, $\R_-\ni p\longmapsto \hat a(p,q)$ is non-decreasing, and for any $p\leq 0$, $\R_-\ni q\longmapsto \hat a(p,q)$ is non-decreasing as well. Next, we compute
\[
\partial^2_{pp}\hat a(p,q)
=
\frac{-ph^\prime\big(2(h^{\prime\prime})^2-h^\prime h^{(3)}\big)
        -2\eta q h^\prime\big(2(h^{\prime\prime})^3
        -h^\prime h^{\prime\prime}h^{(3)}-(h^\prime)^2h^{(4)}\big)}
      {-ph^{\prime\prime}-\eta q[(h^{\prime})^2]^{\prime\prime}}
(\hat a(p,q)),
\]
which is non-negative for any $(p,q)\in\R_-^2$ by assumption.
\end{proof}
\noindent We now recall, for any $n\in\N^\star\cup\{+\infty\}$, the dual ODEs from \Cref{eq:dppeqndual}
\begin{equation*}
w^\star(p)-(1-\delta) p(w^\star)^\prime(y)-F^\star(\delta p)+\mathfrak I_n\bigg(  p,\frac{1}{(w^\star)^{\prime\prime}}\bigg)=0=0,\; p\in(\underline p,\bar p),\; w^\star(\underline p)=\underline w,\; w^\star(\bar p)=\bar w,
\end{equation*}
We next write the ODE that is formally satisfied by $\phi:=-(w^\star)^\prime$. Since we will use these results on compact sets over which $w$ will be decreasing, we know by $\bar p<0$ and that $w^\star$ is non-increasing. By formal differentiation of the equation, we are reduced to searching for a negative and non-decreasing solution $\phi$ of the following equation
\begin{equation}\label{eq:DPEdiff}
\delta\phi(p)
-H_n\big(p,\phi^\prime(p),\phi^{\prime\prime}(p)\big)
= 0,\; p\in(\underline p,\bar p),\; \phi(\bar p)=-(w^\star)^\prime(\bar p),\; \phi(\underline p)=-(w^\star)^\prime(\underline p),
\end{equation}
where we introduced the map
\[
H_n(y,p,q):=-(1-\delta)yp +\delta(F^\star)^\prime(\delta y)-b_n\bigg(y,-\frac1p\bigg)-\gamma_n\bigg(y,-\frac1p\bigg)\frac{q}{p^2}.
\]
Our next result is a concavity property for $G_n$.
\begin{lemma}\label{lemma:fundconc}
Under the assumptions of {\rm \Cref{lemma:convexityG}} and assuming that $F^\prime$ is concave, we have that for any $p> 0$, the map $\R_-^2\times(0,+\infty)\ni (y,r,q)\longmapsto G_n\big(y,r,p,q^{-1}\big)$ is concave. 
\end{lemma}

\begin{proof}
$G_n$ is constituted of several terms. It is obvious that $(y,r,q)\longmapsto \delta r-(1-\delta)yp +\delta(F^\star)^\prime(\delta y)$ is concave since $F^\prime$ is concave. Moreover, by \Cref{lemma:convexityG}, $(y,r,q)\longmapsto-b_n(y,-1/p)$ is also concave. It remains to verify the concavity of the map $f(y,q):=-\gamma_n(y,-1/p)/ q$ for fixed $p< 0$. We have using \Cref{lemma:convexityG}
\begin{align*}
\partial_{yy}f(y,q)&=-2\eta \delta q^{-1}\bigg(\frac{\big((h^\prime)^2((h^{\prime\prime})^2+h^\prime h^{(3)})\big)(\hat a(y,-1/p))}{\big(rh^{\prime\prime}(\hat a(y,-1/p))+2\eta p((h^{\prime\prime})^2+h^\prime h^{(3)})(\hat a(y,-1/p))\big)^2}\\
&\quad+\big((h^\prime)^2 h^{\prime\prime}\big)(\hat a(y,-1/p))\frac{r\big(2(h^{\prime\prime})^2-h^\prime h^{(3)}\big)(\hat a(y,-1/p))+2\eta p\big(2(h^{\prime\prime})^3-h^\prime h^{\prime\prime}h^{(3)}-(h^\prime)^2h^{(4)}\big)(\hat a(y,-1/p))}{\big(rh^{\prime\prime}(\hat a(y,-1/p))+2\eta p((h^{\prime\prime})^2+h^\prime h^{(3)})(\hat a(y,-1/p))\big)^3}\bigg)\\
&
=\frac{-2\eta\delta(h^\prime)^2}
{q(rh^{\prime\prime}+\eta p[(h^{\prime})^2]^{\prime\prime})^3}
\big[ 3r(h^{\prime\prime})^3+2\eta p\big(3(h^{\prime\prime})^4+h^\prime (h^{\prime\prime})^2h^{(3)}+(h^\prime h^{(3)})^2-(h^\prime)^2 h^{\prime\prime}h^{(4)}\big)\big]
(\hat a(y,-1/p)).
\end{align*}
Direct computation of the remaining derivatives provides the following Hessian matrix for $f$
\[
\nabla^2f(y,q)
=
\frac{2\eta\delta {h^\prime}^2}
       {q}
\left( \begin{array}{cc}
-\frac{3r{(h^{\text{\fontsize{4}{4}\selectfont $\prime\prime$}})}^{\text{\fontsize{4}{4}\selectfont $3$}}
        +2\eta p(3({h^{\text{\fontsize{4}{4}\selectfont $\prime\prime$}}})^{\text{\fontsize{4}{4}\selectfont $4$}}+h^{\text{\fontsize{4}{4}\selectfont $\prime$}} (h^{\text{\fontsize{4}{4}\selectfont $\prime\prime$}})^{\text{\fontsize{4}{4}\selectfont $2$}}h^{\text{\fontsize{4}{4}\selectfont $(3)$}}
                            +(h^{\text{\fontsize{4}{4}\selectfont $\prime$}} h^{\text{\fontsize{4}{4}\selectfont $(3)$}})^{\text{\fontsize{4}{4}\selectfont $2$}}-(h^{\text{\fontsize{4}{4}\selectfont $\prime$}})^{\text{\fontsize{4}{4}\selectfont $2$}} h^{\text{\fontsize{4}{4}\selectfont $\prime\prime$}}h^{\text{\fontsize{4}{4}\selectfont $(4)$}})}
       {(rh^{\text{\fontsize{4}{4}\selectfont $\prime\prime$}}+\eta p({h^{\text{\fontsize{4}{4}\selectfont $\prime$}}}^{\text{\fontsize{4}{4}\selectfont $2$}})^{\text{\fontsize{4}{4}\selectfont $\prime\prime$}})^{\text{\fontsize{4}{4}\selectfont $3$}}}
& \times
\\[0.3em]
\frac{-h^{\text{\fontsize{4}{4}\selectfont $\prime\prime$}}}
       {q(rh^{\text{\fontsize{4}{4}\selectfont $\prime\prime$}}+\eta p({h^{\text{\fontsize{4}{4}\selectfont $\prime$}}}^{\text{\fontsize{4}{4}\selectfont $2$}})^{\text{\fontsize{4}{4}\selectfont $\prime\prime$}})}
&
-\frac{1}{q^{\text{\fontsize{4}{4}\selectfont $2$}}}
\end{array}
\right)
(\hat a(y,-1/p)).
\]
By \Cref{lemma:convexityG}, $f$ is partially concave in $y$ and $q$. Then, $\nabla^2 f\le 0$ if and only if det$(\nabla^2 f)\ge 0$, i.e. $\big(2r(h^{\prime\prime})^3
+2\eta p(2(h^{\prime\prime})^4+(h^\prime h^{(3)})^2-(h^\prime)^2 h^{\prime\prime}h^{(4)})\big)
(\hat a(y,-1/p))
\ge 0,$ which holds true under our assumptions.
\end{proof}

\noindent We finally prove the comparison theorem for \Cref{eq:DPEdiff}.
\begin{lemma}\label{lem:comparison2}{\rm (Comparison)}
Let $u$ and $v$ be respectively a viscosity sub-solution and a viscosity super-solution of \eqref{eq:DPEdiff} with $\delta<1$. Assume further that $u$, $v$ are non-negative, non-decreasing, and $u\leq v$ on $\{\underline p,\bar p\}$. Then $u\le v$ on $[\underline p,\bar p]$.
\end{lemma}

\begin{proof}
We again argue exactly as in \emph{Step $1$} of the proof of \Cref{lem:comparison}, with $\mu=2$, $\eps=0$, and with maximisation in the definition of $M_\alpha=M_{\alpha,0}$ confined to the compact subset $[\underline p,\bar p]^2$. Then, whenever $\eta:=(u-v)(p_o)>0$, for some $p_o\in\R$, it follows from the sub-solution and super-solution properties of $u$ and $v$, we have for any $m\in\N$ that \begin{align*}
\eta\leq \delta\big(u(x_m)-v(y_m)\big)\leq
F_{m}:= H_n\big(y_m,\alpha_m(x_m-y_m),Y_m\big) - H_n\big(x_m,\alpha_m(x_m-y_m),X_m\big).
\end{align*}
Now we can use the fact that $(F^\star)^\prime$ and $b_n$ are Lipschitz continuous on compact sets to deduce that there is some $c_o>0$ (which can change values from line to line but is independent of $m$) such that
\begin{align*}
F_{m}&\leq (1-\delta)\alpha_m(x_m-y_m)^2 +c_o|x_m-y_m| \\
&\quad+\gamma_n\bigg(x_m,-\frac1{\alpha_m(x_m-y_m)}\bigg)\frac{X_m}{\alpha_m^2(x_m-y_m)^2}-\gamma_n\bigg(y_m,-\frac1{\alpha_m(x_m-y_m)}\bigg)\frac{Y_m}{\alpha_m^2(x_m-y_m)^2}.
\end{align*}
Now using standard arguments as in \cite[Example 3.6]{crandall1992user}, we know that for any $(A,B)\in\R^2$, $A^2X_m-B^2Y_m\leq 2\alpha_m(A-B)^2.$ We thus deduce using in addition the Lipschitz continuity of $\gamma_n^{1/2}$ on compact sets
\begin{align*}
F_{m}&\leq (1-\delta)\alpha_m(x_m-y_m)^2 +c_o|x_m-y_m| \\
&\quad +2\alpha_m\bigg(\gamma_n^{1/2}\bigg(x_m,-\frac1{\alpha_m(x_m-y_m)}\bigg)\frac{1}{\alpha_m|x_m-y_m|}-\gamma_n^{1/2}\bigg(y_m,-\frac1{\alpha_m(x_m-y_m)}\bigg)\frac{1}{\alpha_m|x_m-y_m|}\bigg)^2\\
&\leq (1-\delta)\alpha_m(x_m-y_m)^2 +c_o|x_m-y_m| +\frac{c_o}{\alpha_m},
\end{align*}
and it suffices to let $m$ go to $+\infty$ to end the proof.
\end{proof}

{
\footnotesize
\bibliography{bibliographyDylan}

}

 \end{document}